\documentclass{amsart}
\usepackage{amsfonts}
\usepackage{amsmath}
\usepackage{amssymb}

\theoremstyle{plain}
\newtheorem{theorem}{Theorem}[section]
\newtheorem{lemma}[theorem]{Lemma}
\newtheorem{proposition}[theorem]{Proposition}
\theoremstyle{definition}
\newtheorem{definition}[theorem]{Definition}

\theoremstyle{remark}
\newtheorem{remark}[theorem]{Remark}
\numberwithin{equation}{section}

\input{tcilatex}

\begin{document}
\title[ The Lagrangian and Hamiltonian Aspects of Electrodynamics]{The
Lagrangian and Hamiltonian Aspects of the Electrodynamic Vacuum-Field Theory
Models }
\author{Nikolai N. Bogolubov (Jr.)}
\address{V.A. Steklov Mathematical Institute of RAS, Moscow, Russian
Federation\\
and\\
The Abdus Salam International Centre of Theoretical Physics, Trieste, Italy}
\email{nikolai\_bogolubov@hotmail.com}
\author{Denis L. Blackmore}
\address{Department of Mathematical Sciences and Center for Applied
Mathematics and Statistics, New Jersey Institute of Technology, Newark, NJ
07102-1982 USA}
\email{deblac@m.njit.edu}
\author{Anatolij K. Prykarpatski}
\address{The Department of Applied Mathematics at AGH University of Science
and Technology, Krakow 30059, Poland }
\email{pryk.anat@ua.fm, prykanat@cybergal.com}
\date{present}

\begin{abstract}
We review the modern classical electrodynamics problems and present the
related main fundamental principles characterizing the electrodynamical
vacuum-field structure. We analyze the models of the vacuum field medium and
charged point particle dynamics using the developed field theory concepts.
There is also described a new approach to the classical Maxwell theory based
on the derived and newly interpreted basic equations making use of the
vacuum field theory approach. In particular, there are obtained the main
classical special relativity theory relations and their new explanations.
The well known Feynman approach to Maxwell electromagnetic equations and the
Lorentz type force derivation is also discussed in detail. A related charged
point particle dynamics and a hadronic string model analysis is also
presented. We also revisited and reanalyzed the classical Lorentz force
expression in arbitrary non-inertial reference frames and present some new
interpretations of the relations between special relativity theory and its
quantum mechanical aspects. Some results related with the charge particle
radiation problem and the magnetic potential topological aspects are
discussed. The electromagnetic Dirac-Fock-Podolsky problem of the Maxwell
and Yang-Mills type dynamical systems is analyzed within the classical
Dirac-Marsden-Weinstein symplectic reduction theory. The problem of
constructing Fock type representations and retrieving their
creation-annihilation operator structure is analyzed. An application of the
suitable current algebra representation to describing the non-relativistic
Aharonov-Bohm paradox is presented. The current algebra coherent functional
representations are constructed and their importance subject to the
linearization problem of nonlinear dynamical systems in Hilbert spaces is
demonstrated.
\end{abstract}

\maketitle

\section{Introduction}

Classical electrodynamics is nowadays considered \cite{LaLi,Paul,Jack,Rohr}
as the most fundamental physical theory, largely owing to the depth of its
theoretical foundations and wealth of experimental verifications. In the
work we describe a new approach to the classical Maxwell theory, based on a
vacuum field medium model, and reanalyze some of the modern classical
electrodynamics problems related with the description of a charged point
particle dynamics under external electromagnetic field. We remark here that
under \textit{"a charged point particle"} we as usually understand an
elementary material charged particle whose internal spatial structure is
assumed to be unimportant and is not taken into account, if the contrary is
not specified.

The important physical principles, characterizing the related
electrodynamical vacuum field structure and based on the least action
principle, we discuss subject to different charged point particle dynamics,
based on the fundamental least action principle. In particular, the main
classical relativistic relationships, characterizing the charge point
particle dynamics, we obtain by means of the \ least action principle within
the Feynman's approach to the Maxwell electromagnetic equations and the
Lorentz type force derivation. \ Moreover, for each least action principle
constructed in the work, we describe the corresponding Hamiltonian pictures
and present the related energy conservation laws. Making use of the
developed modified least action approach a classical hadronic string model
is analyzed in detail.

As the classical Lorentz force expression with respect to an arbitrary
inertial reference frame is related with many theoretical and experimental
controversies, \ such as the relativistic potential energy impact into the
charged point particle mass, the Aharonov-Bohm effect \cite{AhBo} and the
Abraham-Lorentz-Dirac radiation force \cite{Jack,Barr,LaLi} expression, the
analysis of its structure subject to \ the assumed vacuum field medium
structure is \ a very interesting and important problem, which was discussed
by many physicists including E. Fermi, G. Schott, R. Feynman, F. Dyson \cite%
{Ferm,Scho,Feyn-1,Dyso-1,Dyso-2,Glau} and many others. Trying the latter to
explain R. Feynman \cite{Feyn-1} \ in his "Lectures on Physics" wrote:%
\textit{\ }

\textit{"Now we would like to state the law that for quantum mechanics
replaces the law }$F=q\mathbf{v}\times \mathbf{B}.$\textit{\ It will be the
law that determines the behavior of quantum mechanical particles in an
electromagnetic field. Since what happens is determined by amplitudes, the
law must tell us how the magnetic influences affect the amplitudes; we are
no longer dealing with the acceleration of the particle. The law is the
following: the phase of the amplitude to arrive via any trajectory is
changed by the presence of a magnetic field by an amount equal to the
integral of the vector potential along the whole trajectory times the charge
of the particle over Planck's constant. That is, }

\begin{equation*}
\mathit{\ }\text{Magnetic change in phase }=-\frac{q}{\hbar }\int \mathbf{A}%
\cdot d\mathbf{s}\text{ \ \ \ \ \ \ \ \ \ \ \ \ \ \ \ \ \ \ }(15.29)
\end{equation*}

\textit{If there were no magnetic eld there would be a certain phase of
arrival. If there is a magnetic eld anywhere, the phase of the arriving wave
is increased by the integral in Eq. (15.29). Although we will not need to
use it for our present discussion, let us mention that the effect of an
electrostatic field is to produce a phase change given by the negative of
the time integral of the scalar potential :}%
\begin{equation*}
\text{Electric change in phase }=-\frac{q}{\hbar }\int \phi \cdot dt
\end{equation*}

\textit{These two expressions are correct not only for static fields, but
together give the correct result for any electromagnetic field, static or
dynamic. This is the law that replaces }$F=q(\mathbf{E}+\mathbf{v}\times 
\mathbf{B})."$\textit{\ \ }

To describe the essence of the electrodynamic problems related with the
description of a charged point particle dynamics under external
electromagnetic field, let us begin with analyzing the classical Lorentz
force expression 
\begin{equation}
dp/dt=F_{\xi }:=\xi E+\xi u\times B,  \label{L0.1}
\end{equation}%
where $\xi \in \mathbb{R}$ is a particle electric charge, $u\in T(\mathbb{R}%
^{3})$ is its velocity \cite{AbMa,BlPrSa} vector, expressed here in the
light speed $c$ units,%
\begin{equation}
E:=-\partial A/\partial t-\nabla \varphi  \label{L0.2}
\end{equation}%
is the corresponding external electric field and 
\begin{equation}
B:=\nabla \times A  \label{L0.3}
\end{equation}%
is the corresponding external magnetic field, acting on the charged
particle, expressed in terms of suitable vector $A:M^{4}\rightarrow \mathbb{E%
}^{3}$ and scalar $\varphi :M^{4}\rightarrow \mathbb{R}$ potentials. Here
\textquotedblright $\nabla $\textquotedblright\ \ is the standard gradient
operator with respect to the spatial variable $r\in \mathbb{E}^{3},$ $%
"\times "$is the usual vector product in three-dimensional Euclidean vector
space $\mathbb{E}^{3}:=(\mathbb{R}^{3},<\cdot ,\cdot >),$ which is naturally
endowed with the classical scalar product $\ <\cdot ,\cdot >.$ These
potentials are defined on the Minkowski space $M^{4}\simeq \mathbb{R}\times 
\mathbb{E}^{3}$, which models a chosen laboratory reference frame $\mathcal{K%
}.$ \ Now, it is a well known fact \cite{LaLi,Paul,Feyn-1,Thir} that the
force expression (\ref{L0.1}) does not take into account the dual influence
of the charged particle on the electromagnetic field and should be
considered valid only if the particle charge $\xi \rightarrow 0.$ This also
means that expression (\ref{L0.1}) cannot be used for studying the
interaction between two different moving charged point particles, as was
pedagogically demonstrated in classical manuals \cite{LaLi,Feyn-1}. As the
classical Lorentz force \ expression (\ref{L0.1}) is a natural consequence
of the interaction of a charged point particle with an ambient
electromagnetic field, its corresponding derivation based on the general
principles of dynamics, was deeply analyzed by R. Feynman and F. Dyson \cite%
{Feyn-1,Dyso-1,Dyso-2}.

Taking this into account, it is natural to reanalyze this problem from the
classical, taking only into account the Maxwell-Faraday wave theory aspect,
specifying the corresponding vacuum field medium. \ Other questionable
inferences from the classical electrodynamics theory, which strongly
motivated the analysis in this work, are related both with an alternative
interpretation of the well-known \textit{Lorenz condition}, imposed on the
four-vector of electromagnetic observable potentials $(\varphi
,A):M^{4}\rightarrow T^{\ast }(M^{4})$ and the classical Lagrangian
formulation \cite{LaLi} of charged particle dynamics under external
electromagnetic field. The Lagrangian approach latter is strongly dependent
on an important Einsteinian notion of the rest reference frame $\mathcal{K}%
_{\tau }$ and the related least action principle, so before explaining it in
more detail, we first to analyze the classical Maxwell electromagnetic
theory from a strictly dynamical point of view.

\section{\label{Sec_1}Classical relativistic electrodynamics models
revisiting: Lagrangian and Hamiltonian analysis}

\subsection{\label{Subsec_1.1}Introductory setting}

Let us consider with respect to a laboratory reference frame $\mathcal{K}%
_{t} $ \ the additional \textit{Lorenz condition} 
\begin{equation}
\partial \varphi /\partial t+<\nabla ,A>=0,  \label{L0.4}
\end{equation}%
\ \textit{a priori} assumed the Lorentz invariant wave scalar field equation 
\begin{equation}
\partial ^{2}\varphi /\partial t^{2}-\nabla ^{2}\varphi =\rho \text{ \ \ }
\label{L0.5}
\end{equation}%
and the charge continuity equation 
\begin{equation}
\partial \rho /\partial t+<\nabla ,j>=0,  \label{L0.6}
\end{equation}%
where $\rho :M^{4}\rightarrow \mathbb{R}$ and $j:M^{4}\rightarrow \mathbb{E}%
^{3}$ are, respectively, the charge and current densities of the ambient
matter. Then one can derive \cite{BoPrTa,BoPrTaPr-Lore} that the Lorentz
invariant wave equation 
\begin{equation}
\partial ^{2}A/\partial t^{2}-\nabla ^{2}A=j\   \label{L0.5a}
\end{equation}%
and the classical electromagnetic Maxwell field equations \cite%
{Jack,LaLi,Feyn-1,Paul,Thir}%
\begin{align}
\nabla \times E+\partial B/\partial t& =0,\text{ \ \ \ }<\nabla ,E>=\rho ,
\label{L0.7} \\
\nabla \times B-\partial E/\partial t& =j,\text{ \ \ \ \ \ \ }<\nabla ,B>=0,
\notag
\end{align}%
hold for all $(t,r)\in M^{4}$ with respect to the chosen laboratory
reference frame $\mathcal{K}_{t}.$

Notice here that, inversely, Maxwell's equations (\ref{L0.7}) do not
directly reduce, via definitions (\ref{L0.2}) and (\ref{L0.3}), to the wave
field equations (\ref{L0.5}) and \ (\ref{L0.5a}) without the Lorenz
condition (\ref{L0.4}). This fact \ is very important and suggests that when
it comes to a choice of governing equations, it may be reasonable to replace
Maxwell's equations (\ref{L0.7}) with the Lorenz condition (\ref{L0.4}) and
the charge continuity equation (\ref{L0.6}). To make the equivalence
statement, claimed above, more transparent we formulate it as the following
proposition.

\begin{proposition}
The Lorentz invariant wave equation (\ref{L0.5}) together with the Lorenz
condition (\ref{L0.4}) for the observable potentials $(\varphi
,A):M^{4}\rightarrow T^{\ast }(M^{4})$ and the charge continuity
relationship (\ref{L0.6}) are completely equivalent to the Maxwell field
equations (\ref{L0.7}).
\end{proposition}

\begin{proof}
Substituting (\ref{L0.4}), into (\ref{L0.5}), one easily obtains 
\begin{equation}
\partial ^{2}\varphi /\partial t^{2}=-<\nabla ,\partial A/\partial
t>=<\nabla ,\nabla \varphi >+\rho ,  \label{L0.8}
\end{equation}%
which implies the gradient expression 
\begin{equation}
<\nabla ,-\partial A/\partial t-\nabla \varphi >=\rho .  \label{L0.9}
\end{equation}%
Taking into account the electric field definition (\ref{L0.2}), expression (%
\ref{L0.9}) reduces to 
\begin{equation}
<\nabla ,E>=\rho ,  \label{L0.10}
\end{equation}%
which is the second of the first pair of Maxwell's equations (\ref{L0.7}).

Now upon applying $\nabla \times $ to definition (\ref{L0.2}), we find,
owing to definition (\ref{L0.3}), that 
\begin{equation}
\nabla \times E+\partial B/\partial t=0,  \label{L0.11}
\end{equation}%
which is the first pair of the Maxwell equations (\ref{L0.7}). Having
differentiated with respect to the temporal variable $t\in \mathbb{R}$ the
equation \ (\ref{L0.5}) and taken into account the charge continuity
equation \ (\ref{L0.6}), one finds that 
\begin{equation}
<\nabla ,\partial ^{2}A/\partial t^{2}-\nabla ^{2}A-j>=0.  \label{L0.11a}
\end{equation}%
The latter is equivalent to the wave equation \ (\ref{L0.5a}) if to observe
that the current vector $j:M^{4}\rightarrow \mathbb{E}^{3}$ is defined by
means of the charge continuity equation \ (\ref{L0.6}) up to a vector
function $\nabla \times S:M^{4}\rightarrow \mathbb{E}^{3}.\ $Now applying
operation $\nabla \times $ to the definition (\ref{L0.3}), owing to the wave
equation \ (\ref{L0.5a}) one obtains%
\begin{align}
\nabla \times B& =\nabla \times (\nabla \times A)=\nabla <\nabla ,A>-\nabla
^{2}A=  \notag \\
& =-\nabla (\partial \varphi /\partial t)-\partial ^{2}A/\partial
t^{2}+(\partial ^{2}A/\partial t^{2}-\nabla ^{2}A)=  \notag \\
& =\frac{\partial }{\partial t}(-\nabla \varphi -\partial A/\partial
t)+j=\partial E/\partial t+j,  \label{L0.12}
\end{align}%
which leads directly to 
\begin{equation*}
\nabla \times B=\partial E/\partial t+j,
\end{equation*}%
which is the first of the second pair of the Maxwell equations (\ref{L0.7}).
The final \textit{"no magnetic charge}" equation 
\begin{equation*}
<\nabla ,B>=<\nabla ,\nabla \times A>=0,
\end{equation*}%
in (\ref{L0.7}) \ follows directly from the elementary identity $<\nabla
,\nabla \times >=0,$ thereby completing the proof.
\end{proof}

This proposition allows us to consider the observable potential functions $%
(\varphi ,A):M^{4}\rightarrow T^{\ast }(M^{4})$ as fundamental ingredients
of the ambient \textit{vacuum field medium}, by means of which we can try to
describe the related physical behavior of charged point particles imbedded
in space-time $M^{4}.$ The following observation provides strong support for
this approach:

\textbf{Observation.} \textit{The Lorenz condition (\ref{L0.4}) actually
means that the scalar potential field }$\varphi :$\textit{\ }$%
M^{4}\rightarrow \mathbb{R}$\textit{\ continuity relationship, whose origin
lies in some new field conservation law, characterizes the deep intrinsic
structure of the vacuum field medium.}

To make this observation more transparent and precise, let us recall the
definition \cite{LaLi,Paul,Feyn-1,Thir} of the electric current $J:$\textit{%
\ }$M^{4}\rightarrow \mathbb{E}^{3}$ in the dynamical form 
\begin{equation}
J:=\rho u,  \label{L0.13}
\end{equation}%
where the vector $u\in T(\mathbb{R}^{3})$ is the corresponding charge
velocity. Thus, the following continuity relationship%
\begin{equation}
\partial \rho /\partial t+<\nabla ,\rho u>=0  \label{L0.14}
\end{equation}%
holds, which can easily be rewritten \cite{MaCh} \ as the integral
conservation law 
\begin{equation}
\frac{d}{dt}\int_{\Omega _{t}}\rho (t,r)d^{3}r=0  \label{L0.15}
\end{equation}%
for the charge inside of any bounded domain $\Omega _{t}\subset \mathbb{E}%
^{3},$ moving in the space-time $M^{4}$ with respect to the natural
evolution equation 
\begin{equation}
dr\ /dt:=u.  \label{L0.16}
\end{equation}%
Following the above reasoning, we obtain the following result.

\begin{proposition}
The Lorenz condition (\ref{L0.4}) is equivalent to the integral conservation
law%
\begin{equation}
\frac{d}{dt}\int_{\Omega _{t}}\varphi (t,r)d^{3}r=0,  \label{L0.17}
\end{equation}%
where $\Omega _{t}\subset \mathbb{E}^{3}$ \ is any bounded domain, moving
with respect to the charged point particle $\xi $ evolution equation 
\begin{equation}
dr/dt=u(t,r),  \label{L0.18}
\end{equation}%
which represents the velocity vector of related local potential field
changes propagating in the Minkowski space-time $M^{4}.$ Moreover, for a
particle$\ $with the distributed $\ $charge density $\rho :M^{4}\rightarrow 
\mathbb{R},$ the following Umov type local energy conservation relationship 
\begin{equation}
\frac{d}{dt}\int_{\Omega _{t}}\frac{\rho (t,r)\varphi (t,r)}{%
(1-|u(t,r)|^{2})^{1/2}}d^{3}r=0\   \label{L0.18a}
\end{equation}%
holds for any $t\in \mathbb{R}.$
\end{proposition}

\begin{proof}
Consider first the corresponding solutions to potential field equations (\ref%
{L0.5}), taking into account condition (\ref{L0.13}). Owing to the standard
results from \cite{Feyn-1,LaLi}, one finds that 
\begin{equation}
A=\varphi u,  \label{L0.19}
\end{equation}%
which gives rise to the following form of the Lorenz condition (\ref{L0.4}):%
\begin{equation}
\partial \varphi /\partial t+<\nabla ,\varphi u>=0,  \label{L0.20}
\end{equation}%
This obviously can be rewritten \cite{MaCh} as the integral conservation law
(\ref{L0.17}), so the expression \ (\ref{L0.17}) is stated.

To state the local energy conservation relationship \ (\ref{L0.18a}) it is
necessary to combine the conditions \ \ (\ref{L0.14}), (\ref{L0.20}) and
find that 
\begin{equation}
\partial (\rho \varphi )/\partial t+<u,\nabla (\rho \varphi )>+2\rho \varphi
<\nabla ,u>=0.  \label{L0.20a}
\end{equation}%
Taking into account that the infinitesimal volume transformation $%
d^{3}r=\chi (t,r)d^{3}r_{0},$ where the Jacobian $\ \chi (t,r):=|\partial
r(t;r_{0})/\partial r_{0}|$ of the corresponding transformation $r:\mathbb{\ 
}\Omega _{t_{0}}\rightarrow $ $\Omega _{t},$\ induced by the Cauchy problem
for the differential relationship \ (\ref{L0.18}) for any $t\in \mathbb{R},$
satisfies the evolution equation 
\begin{equation}
d\chi /dt=<\nabla ,u>\chi ,  \label{L0.20b}
\end{equation}%
easily following from \ (\ref{L0.18}), and applying to the equality \ (\ref%
{L0.20a}) the operator $\int_{\Omega _{t_{0}}}(...)\chi ^{2}d^{3}r_{0},$ one
obtains that 
\begin{equation}
\begin{array}{c}
0=\int_{\Omega _{t_{0}}}\frac{d}{dt}(\rho \varphi \chi ^{2})d^{3}r_{0}=\frac{%
d}{dt}\int_{\Omega _{t_{0}}}(\rho \varphi \chi )Jd^{3}r_{0}= \\ 
=\frac{d}{dt}\int_{\Omega _{t}}(\rho \varphi \chi ^{\ })d^{3}r\ :=\ \frac{d}{%
dt}\mathcal{E}(\xi ;\Omega _{t}).%
\end{array}
\label{L0.20c}
\end{equation}%
Here we denoted \ the conserved charge $\xi :=\int_{\Omega _{t}}\rho
(t,r)d^{3}r$ and the local energy conservation quantity $\ \mathcal{E}(\xi
;\Omega _{t}):$ $=\int_{\Omega _{t}}(\rho \varphi \chi ^{\ })d^{3}r.$ The
latter quantity can be simplified, owing to the infinitesimal Lorentz
invariance four-volume measure relationship\ $d^{3}r(t,r_{0})\wedge
dt=d^{3}r_{0}\wedge dt_{0},$ where variables $(t,r)\in \mathbb{R}_{t}\times
\Omega _{t}\subset M^{4}$ are, within the present context, taken with
respect to the moving reference frame $\mathcal{K}_{t},$ related to the
infinitesimal charge quantity $d\xi (t,r):=\rho (t,r)d^{3}r,$\ and variables 
$(t_{0},r_{0})\in \mathbb{R}_{t_{0}}\times \Omega _{t_{0}}\subset M^{4}$ are
taken with respect to the laboratory reference frame $\ \mathcal{K}_{t_{0}},$
related to the infinitesimal charge quantity $d\xi (t_{0},r_{0})=\rho
(t_{0},r_{0})d^{3}r_{0},$ satisfying the charge conservation invariance $%
d\xi (t,r)=d\xi (t_{0},r_{0}).$ The mentioned above infinitesimal Lorentz
invariance relationships make it possible to calculate the local energy
conservation quantity $\mathcal{E}(\xi ;\Omega _{0})$ as $\ $ 
\begin{eqnarray}
\mathcal{E}(\xi ;\Omega _{0}) &=&\int_{\Omega _{t}}(\rho \varphi \chi ^{\
})d^{3}r=\int_{\Omega _{t}}(\rho \varphi \frac{d^{3}r}{d^{3}r_{0}}^{\
})d^{3}r=  \label{L0.20d} \\
&=&\int_{\Omega _{t}}(\rho \varphi \frac{d^{3}r\wedge dt}{d^{3}r_{0}\wedge dt%
}^{\ })d^{3}r=\int_{\Omega _{t}}(\rho \varphi \frac{d^{3}r_{0}\wedge dt_{0}}{%
d^{3}r_{0}\wedge dt}^{\ })d^{3}r=  \notag \\
&=&\int_{\Omega _{t}}(\rho \varphi \frac{\ dt_{0}}{\ dt})^{\
}d^{3}r=\int_{\Omega _{t}}\frac{\rho \varphi \ d^{3}r\ }{\ (1-|u|^{2})^{1/2}}%
,  \notag
\end{eqnarray}%
where we took into account that $dt=dt_{0}(1-|u|^{2})^{1/2}.$\ Thus, owing
to \ (\ref{L0.20c}) and \ (\ref{L0.20d}) the local energy conservation
relationship \ (\ref{L0.18a}) is satisfied, proving the proposition.
\end{proof}

The constructed above local energy conservation quantity \ (\ref{L0.20d})
can be rewritten as 
\begin{equation}
\mathcal{E}(\xi ;\Omega _{t})=\int_{\Omega _{t}}\frac{d\xi (t,r)\varphi
(t,r)\ \ \ }{\ (1-|u|^{2})^{1/2}}:=\int_{\Omega _{t}}d\mathcal{E}(t,r)
\label{L0.20e}
\end{equation}%
where $d\mathcal{E}(t,r)=d\xi (t,r)\varphi (t,r)(1-|u|^{2})^{-1/2}$ is the
distributed in vacuum electromagnetic field energy density, related with the
electric charge $d\xi (t,r),$ located at point $(t,r)\in M^{4}.$

The above proposition suggests a physically motivated interpretation of
electrodynamic phenomena in terms of what should naturally be called \textit{%
the vacuum potential field}, which determines the observable interactions
between charged point particles. More precisely, we can \textit{a priori }\
endow the ambient vacuum medium with a scalar potential energy field density
function $W:=\xi \varphi :M^{4}\rightarrow \mathbb{R},$ where $\xi \in 
\mathbb{R}_{+}$ is the value of an elementary charge quantity, and
satisfying the governing \textit{vacuum field equations} 
\begin{equation}
\begin{array}{c}
\partial ^{2}W/\partial t^{2}-\nabla ^{2}W=\rho \xi ,\text{ \ \ }\partial
W/\partial t+<\nabla ,\hat{A}>=0, \\ 
\partial ^{2}\hat{A}/\partial t^{2}-\nabla ^{2}\hat{A}=\xi \rho v,\text{ \ \ 
}\hat{A}=Wv,%
\end{array}
\label{L0.21}
\end{equation}%
taking into account the external charged sources, which possess a virtual
capability for disturbing the vacuum field medium. Moreover, \ this vacuum
potential field function $W:M^{4}\rightarrow \mathbb{R}$ allows the natural
potential energy interpretation, whose origin should be assigned not only to
the charged interacting medium, but also to any other medium possessing
interaction capabilities, including for instance, material particles,
interacting through the gravity.

The latter leads naturally to the next important step, consisting in
deriving the equation governing the corresponding potential field $\bar{W}%
:M^{4}\rightarrow \mathbb{R},$ assigned to a charged point particle moving
in the vacuum field medium with velocity $u\in \ \ T(\mathbb{R}^{3})$ and
located at point $r(t)=R(t)\in \mathbb{E}^{3}$ at time $t\in \mathbb{R}.$ As
can be readily shown \cite{BoPrTa,BoPrTaPr-Lore,Repc}, the corresponding
evolution equation governing the related potential field function $\bar{W}%
:M^{4}\rightarrow \mathbb{R},$ assigned to a \ moving in the space $\mathbb{E%
}^{3}$ charged particle $\xi $ under the sationarily distributed field
sources, \ has the form 
\begin{equation}
\frac{d}{dt}(-\bar{W}u)=-\nabla \bar{W},  \label{L0.22}
\end{equation}%
where $\bar{W}:=W(t,r)|_{r\rightarrow R(t)},$ $u(t):=dR(t)/dt$ at point
particle location $(t,R(t))\in M^{4}.$

Similarly, if there are two interacting charged point particles, located at
points $r(t)=R(t)$ and $r_{f}(t)=R_{f}(t)\in \mathbb{E}^{3}$ at time $t\in 
\mathbb{R}$ and moving, respectively, with velocities $u:=dR(t)/dt$ \ and \ $%
u_{f}:=dR_{f}(t)/dt,$ the corresponding potential field function $\bar{W}%
^{\prime }:M^{4}\rightarrow \mathbb{R},$ considered with respect to the
reference frame $\mathcal{K}^{\prime }$ specified by Euclidean coordinates $%
(t^{\prime },r-r_{f})\in \mathbb{E}^{4}$ and moving with the velocity $%
u_{f}\in \ T(\mathbb{R}^{3})$ \ subject to the laboratory reference frame $%
\mathcal{K}_{t},$ should satisfy \cite{BoPrTa,BoPrTa-1} with respect to the
reference frame $\mathcal{K}^{\prime }$ the dynamical equality 
\begin{equation}
\frac{d}{dt^{\prime }}[-\bar{W}^{\prime }(u^{\prime }-u_{f}^{\prime
})]=-\nabla \bar{W}^{\prime },  \label{L0.23}
\end{equation}%
where, by definition, we have denoted the velocity vectors $u^{\prime
}:=dr/dt^{\prime },u_{f}^{\prime }:=dr_{f}/dt^{\prime }\in T(\mathbb{R}%
^{3}). $ The latter cames with respect to the laboratory reference frame $%
\mathcal{K}_{t}$ about the dynamical equality 
\begin{equation}
\frac{d}{dt}[-\bar{W}\ (u\ -u_{f})]=-\nabla \bar{W}(1-|u_{f}|^{2}).
\label{L0.24}
\end{equation}%
\ The dynamical potential field equations (\ref{L0.22}) and (\ref{L0.23})
appear to have important properties and can be used as means for
representing classical electrodynamic phenomena. Consequently, we shall
proceed to investigate their physical properties in more detail and compare
them with classical results for Lorentz type forces arising in the
electrodynamics of moving charged point particles in an external
electromagnetic field.

In this investigation, we were in part inspired by works \cite%
{DeJa,DuJa,Puth} and studies \cite{Wilc-1,Wilc-2} devoted to solving the
classical problem of reconciling gravitational and electrodynamic charges
within the Mach-Einstein ether paradigm. First, we will revisit the
classical Mach-Einstein relativistic electrodynamics of a moving charged
point particle, and second, we study the resulting electrodynamic theories
associated with our vacuum potential field dynamical equations (\ref{L0.22})
and (\ref{L0.23}), making use of the fundamental Lagrangian and Hamiltonian
formalisms which were specially devised in \cite{BoPr-1,BoPrTaPr-Lore}.

\subsection{\label{Subsec_1.2}Classical relativistic electrodynamics
revisited}

The classical relativistic electrodynamics of a freely moving charged point
particle in the Minkowski space-time $M^{4}:=\mathbb{R}\times \mathbb{E}^{3}$
is based \ on the Lagrangian approach \cite{LaLi,Feyn-1,Paul,Thir} with
Lagrangian function 
\begin{equation}
\mathcal{L}:=-m_{0}(1-|u|^{2})^{1/2},  \label{L1.1}
\end{equation}%
where $m_{0}\in \mathbb{R}_{+}$ is the so-called particle rest mass and $%
u\in T(\mathbb{R}^{3})$ is its spatial velocity in the Euclidean space $%
\mathbb{E}^{3},$ expressed here and in the sequel in light speed units (with
light speed $c$). The least action principle in the form 
\begin{equation}
\delta S=0,\text{ \ \ }S:=-m_{0}\int_{t_{1}}^{t_{2}}(1-|u|^{2})^{1/2}dt
\label{L1.2}
\end{equation}%
for any fixed \ temporal interval $[t_{1},t_{2}]\subset \mathbb{R}$ gives
rise to the well-known relativistic relationships for the mass of the
particle 
\begin{equation}
m=m_{0}(1-|u|^{2})^{-1/2},  \label{L1.3}
\end{equation}%
the momentum of the particle

\begin{equation}
p:=mu=m_{0}u(1-|u|^{2})^{-1/2}  \label{L1.4}
\end{equation}%
and the energy of the particle%
\begin{equation}
\mathcal{E}_{0}=m=m_{0}(1-|u|^{2})^{-1/2}.  \label{L1.5}
\end{equation}%
It follows from \cite{LaLi,Paul}, that the origin of the Lagrangian (\ref%
{L1.1}) can be extracted from the action%
\begin{equation}
S:=-m_{0}\int_{t_{1}}^{t_{2}}(1-|u|^{2})^{1/2}dt=-m_{0}\int_{\tau
_{1}}^{\tau _{2}}d\tau ,  \label{L1.6}
\end{equation}%
on the suitable temporal interval $[\tau _{1,}\tau _{2}]\subset \mathbb{R}%
\mathbf{,}$ where,\textbf{\ }by definition, 
\begin{equation}
d\tau :=dt(1-|u|^{2})^{1/2}  \label{L1.6a}
\end{equation}%
and $\tau \in \mathbb{R}$ is the so-called, proper temporal parameter
assigned to a freely moving particle with respect to the \ rest \ reference
frame $\mathcal{K}_{\tau }.$ \ The action (\ref{L1.6}) is rather
questionable from the dynamical point of view, since it is physically
defined with respect to the \ rest \ reference frame $\mathcal{K}_{\tau },$
giving rise to the constant action $S=-m_{0}(\tau _{2}-\tau _{1}),$ as the
limits of integrations $\tau _{1}<\tau _{2}\in \mathbb{R}$ were taken to be
fixed from the very beginning. Moreover, considering this particle to have
charge $\xi \in \mathbb{R}$ and be moving in the Minkowski space-time $M^{4}$
under action of an electromagnetic field $(\varphi ,A)\in \mathbb{R}\times 
\mathbb{E}^{3},$ the corresponding classical (relativistic) action
functional is chosen\ (see \cite{LaLi,Feyn-1,Paul,Thir,BoPr-1,BoPrTaPr-Lore}%
) \ as follows: 
\begin{equation}
S:=\int_{\tau _{1}}^{\tau _{2}}[-m_{0}d\tau +\xi <A,\dot{r}>d\tau -\xi
\varphi (1-|u|^{2})^{-1/2}d\tau ],  \label{L1.7}
\end{equation}%
with respect to the \textit{rest reference system}, parameterized by the
Euclidean space-time variables $(\tau ,r)\in \mathbb{E}^{4},$ where we have
denoted $\dot{r}:=dr/d\tau $ in contrast to the definition $u:=dr/dt.$ The
action (\ref{L1.7}) can be rewritten with respect to the laboratory
reference frame $\mathcal{K}_{t}$ \ the moving with velocity\ vector $u\in 
\mathbb{E}^{3}$ as 
\begin{equation}
S=\int_{t_{1}}^{t_{2}}\mathcal{L}dt,\text{ \ }\mathcal{L}%
:=-m_{0}(1-|u|^{2})^{1/2}+\xi <A,u>-\xi \varphi ,  \label{L1.8}
\end{equation}%
on the suitable temporal interval \bigskip $\lbrack t_{1},t_{2}]\subset 
\mathbb{R},$ which gives rise to the following \cite{LaLi,Feyn-1,Paul,Thir}
\ dynamical expressions 
\begin{equation}
P=p+\xi A,\text{ \ \ \ \ }p=mu,\text{ \ \ }m=m_{0}(1-|u|^{2})^{-1/2},
\label{L1.9}
\end{equation}%
for the particle momentum and 
\begin{equation}
\mathcal{E}_{0}=(m_{0}^{2}+|P-\xi A|^{2})^{1/2}+\xi \varphi  \label{L1.10}
\end{equation}%
for the charged particle $\xi $ \ energy, where, by definition, $P\in 
\mathbb{E}^{3}$ is the common momentum of the particle and the ambient
electromagnetic field at a space-time point $(t,r)\in M^{4}.$

The expression (\ref{L1.10}) for the particle energy $\mathcal{E}_{0}$ also
appears open to question, since the potential energy $\xi \varphi ,$
entering additively, has no affect on the particle mass $%
m=m_{0}(1-|u|^{2})^{-1/2}.$ This was noticed by L. Brillouin \cite{Bril},
who remarked that the fact that the potential energy has no affect on the
particle mass tells us that "... any possibility of existence of a particle
mass related with an external potential energy, is completely excluded".
Moreover, it is necessary to stress here that the least action principle (%
\ref{L1.8}), formulated with respect to the laboratory reference frame $%
\mathcal{K}_{t}$ \ time parameter $t\in \mathbb{R},$ appears logically
inadequate, for there is a strong physical inconsistency with other time
parameters of the Lorentz equivalent reference frames. This was first
mentioned by R. Feynman in \cite{Feyn}, in his efforts to rewrite the
Lorentz force expression with respect to the rest reference frame $\mathcal{K%
}_{\tau }.$ This and other special relativity theory and electrodynamics
problems stimulated many prominent physicists\ of the past \cite%
{Bril,Feyn,Weyl,Paul,BrDi} and present \cite%
{Hajr,Hamm-1,Merm,Merm-1,Wilc-1,DuJa,Logu-2,Logu-3,BoBo,Repc,Neum,UrCoSaDj}
to try to develop alternative relativity theories based on completely
different space-time and matter structure principles.

There also is another controversial inference from the action expression (%
\ref{L1.8}). As one can easily show \cite{LaLi,Paul,Feyn-1,Thir}, the
corresponding dynamical equation \ for the Lorentz force is given as 
\begin{equation}
dp/dt=F:=\xi E+\xi u\times B.  \label{L1.11}
\end{equation}%
We have defined here, as before, 
\begin{equation}
E:=-\partial A/\partial t-\nabla \varphi  \label{L1.12}
\end{equation}%
for the corresponding electric field and 
\begin{equation}
B:=\nabla \times A  \label{L1.13}
\end{equation}%
for the related magnetic field, acting on the charged point particle $\xi $.
The expression (\ref{L1.11}) means, in particular, that the Lorentz force $F$%
\ depends linearly on the particle velocity vector $u\in T(\mathbb{R}^{3}),$
and so there is a strong dependence on the reference frame with respect to
which the charged particle $\xi $ moves. Attempts to reconcile this and some
related controversies \cite{Bril,Feyn,Repc,Klym} forced Einstein to devise
his special relativity theory and proceed further to creating his general
relativity theory trying to explain the gravity by means of geometrization
of space-time and matter in the Universe. Here we must mention that the
classical Lagrangian function $\mathcal{L}$ \ in (\ref{L1.8}) is written in
terms of a combination of terms expressed by means of both the Euclidean
rest reference frame variables $(\tau ,r)\in \mathbb{E}^{4}$ and arbitrarily
chosen Minkowski reference frame variables $(t,r)\in M^{4}.$

These problems were recently analyzed using a completely different
"no-geometry" \ approach \cite{BoPrTa,BoPrTa-1,Repc}, where new dynamical
equations were derived, which were free of the controversial elements
mentioned above. Moreover, this approach avoided the introduction of the
well known Lorentz transformations of the space-time reference frames with
respect to which the action functional (\ref{L1.8}) is invariant. From this
point of view, there are interesting for discussion conclusions in \cite%
{JaPo,Sant,Hamm-1,Aqui}, in which some electrodynamic models, possessing
intrinsic Galilean and Poincar\'{e}-Lorentz symmetries, are reanalyzed from
diverse geometrical points of view. Subject to a possible geometric
space-type structure and the related vacuum field background, exerting the
decisive influence on the particle dynamics, we need to mention here recent
works \cite{Smol,AmFrKoSm} and the closely related with their ideas the
classical articles \cite{Kibb,Penr}. Next, we shall revisit the results
obtained in \cite{BoPrTa,BoPrTaPr-Lore} from the classical Lagrangian and
Hamiltonian formalisms \cite{BoPr-1} in order to shed new light on the
physical underpinnings of the vacuum field theory approach to the study of
combined electromagnetic and gravitational effects.

\section{\label{Sec_2}The vacuum field theory electrodynamics equations:
Lagrangian analysis}

\subsection{\label{Subsec_2.1}A moving in vacuum point particle - an
alternative electrodynamic model}

In the vacuum field theory approach to combining electromagnetism and the
gravity, devised in \cite{BoPrTa,BoPrTaPr-Lore}, the main vacuum potential
field function $\bar{W}:M^{4}\mathbb{\rightarrow R},$ related to a charged
point particle $\xi $ under the external stationarily distributed field
sources, satisfies the dynamical equation \ (\ref{L0.21}), namely 
\begin{equation}
\frac{d}{dt}(-\bar{W}u)=-\nabla \bar{W}  \label{L2.1}
\end{equation}%
in the case when the external charged particles are at rest, where, as
above, $u:=dr/dt$ is the particle velocity with respect to some reference
system.

To analyze the dynamical equation (\ref{L2.1}) from the Lagrangian point of
view, we write the corresponding action functional as 
\begin{equation}
S:=-\underset{t_{1}}{\overset{t_{2}}{\int }}\bar{W}dt=-\underset{\tau _{1}}{%
\overset{\tau _{2}}{\int }}\bar{W}(1+|\dot{r}|^{2})^{1/2}\text{ }d\tau ,
\label{L2.2}
\end{equation}%
expressed with respect to the rest reference frame $\mathcal{K}_{\tau }.$
Fixing the proper temporal parameters $\tau _{1}<\tau _{2}\in \mathbb{R},$
one finds from the least action principle ( $\delta S=0$) that 
\begin{align}
p& :=\partial \mathcal{L}/\partial \dot{r}=-\bar{W}\dot{r}(1+|\dot{r}%
|^{2})^{-1/2}=-\bar{W}u,  \label{L2.3} \\
\dot{p}& :=dp/d\tau =\partial \mathcal{L}/\partial r=-\nabla \bar{W}(1+|\dot{%
r}|^{2})^{1/2},  \notag
\end{align}%
where, owing to (\ref{L2.2}), the corresponding Lagrangian function is 
\begin{equation}
\mathcal{L}:=-\bar{W}(1+|\dot{r}|^{2})^{1/2}.  \label{L2.4}
\end{equation}%
Recalling now the definition of the particle mass 
\begin{equation}
m:=-\bar{W}  \label{L2.5}
\end{equation}%
and \ the relationships%
\begin{equation}
d\tau =dt(1-|u|^{2})^{1/2},\text{ }\dot{r}d\tau =udt,  \label{L2.6}
\end{equation}%
from (\ref{L2.3}) we easily obtain exactly the dynamical equation (\ref{L2.1}%
). Moreover, one now readily finds that the dynamical mass, defined by means
of expression (\ref{L2.5}), is given as 
\begin{equation*}
m=m_{0}(1-|u|^{2})^{-1/2},
\end{equation*}%
which coincides with the equation (\ref{L1.3}) of the preceding section. Now
one can formulate the following proposition using the above results

\begin{proposition}
The alternative freely moving point particle electrodynamic model (\ref{L2.1}%
) allows the least action formulation (\ref{L2.2}) with respect to the
"rest\textquotedblright\ \ reference frame variables, where the Lagrangian
function is given by expression (\ref{L2.4}). Its electrodynamics is
completely equivalent to that of a classical relativistic freely moving
point particle, described in Subsection \ \ref{Subsec_1.2}.
\end{proposition}

\subsection{\label{Subsec_2.2}A moving in vacuum interacting two charge
system - an alternative electrodynamic model}

We proceed now to the case when our charged point particle $\xi $ moves in
the space-time with velocity vector $u\in T(\mathbb{R}^{3})$ and interacts
with another external charged point particle $\xi _{f},$ moving with
velocity vector $u_{f}\in T(\mathbb{R}^{3})$ with respect to a common
reference frame $\mathcal{K}.$ As was shown in \cite{BoPrTa,BoPrTaPr-Lore},
the respectively modified dynamical equation for the vacuum potential field
function $\bar{W}^{\prime }:M^{4}\mathbb{\rightarrow R}$ subject to the
moving reference frame $\mathcal{K}^{\prime }$ is given by \ equality (\ref%
{L0.23}), or 
\begin{equation}
\frac{d}{dt^{\prime }}[-\bar{W}^{\prime }(u^{\prime }-u_{f}^{\prime
})]=-\nabla \bar{W}^{\prime },  \label{L2.7}
\end{equation}%
where, as before, the velocity vectors $u^{\prime }:=dr/dt^{\prime
},u_{f}^{\prime }:=dr_{f}/dt^{\prime }\in T(\mathbb{R}^{3}).$ \ Since the
external charged particle $\xi _{f}$ \ moves in the space-time $M^{4},$ it
generates the related magnetic field $B:=\nabla \times A,$ whose magnetic
vector potentials $A:M^{4}\mathbb{\rightarrow E}^{3}$ and $A^{\prime }:M^{4}%
\mathbb{\rightarrow E}^{3}$ \ are defined, owing to the results of \cite%
{BoPrTa,BoPrTaPr-Lore,Repc}, as 
\begin{equation}
\xi A:=\bar{W}u_{f},\text{ \ \ }\xi A^{\prime }:=\bar{W}^{\prime
}u_{f}^{\prime },  \label{L2.8}
\end{equation}%
Whence, taking into account that the field potential 
\begin{equation}
\bar{W}=\bar{W}^{\prime }(1-|u_{f}|^{2})^{-1/2}  \label{L2.8a}
\end{equation}%
and the particle momentum $p^{\prime }=-\bar{W}^{\prime }u^{\prime }=-\bar{W}%
u,$ equality (\ref{L2.7}) becomes equivalent to 
\begin{equation}
\frac{d}{dt^{\prime }}(p^{\prime }+\xi A^{\prime })=-\nabla \bar{W}^{\prime
},  \label{L2.9}
\end{equation}%
if considered with respect to the moving reference frame $\mathcal{K}%
^{\prime },$ or to the Lorentz type force equality 
\begin{equation}
\frac{d}{dt}(p+\xi A)=-\nabla \bar{W}(1-|u_{f}|^{2}),  \label{L2.9a}
\end{equation}%
if considered with respect to the laboratory reference frame $\mathcal{K},$
owing to the classical Lorentz invariance relationship (\ref{L2.8a}), as the
corresponding magnetic vector potential, generated by the external charged
point test particle $\xi _{f}$ with respect to the reference frame $\mathcal{%
K}^{\prime },$ is identically equal to zero. To imbed the dynamical equation
(\ref{L2.9a}) into the classical Lagrangian formalism, we start from the
following action functional, which naturally generalizes the functional (\ref%
{L2.2}):%
\begin{equation}
S:=-\underset{\tau _{1}}{\overset{\tau _{2}}{\int }}\bar{W}^{\prime }(1+|%
\dot{r}-\dot{r}_{f}|^{2})^{1/2}\text{ }d\tau .  \label{L2.10}
\end{equation}%
Here, as before, $\bar{W}^{\prime }$ is the respectively calculated vacuum
field potential $\bar{W}\ $ subject to the moving reference frame $\mathcal{K%
}^{\prime },$ $\dot{r}=u^{\prime }dt^{\prime }/d\tau ,\dot{r}%
_{f}=u_{f}^{\prime }dt^{\prime }/d\tau ,$ $d\tau =dt^{\prime }(1-|u^{\prime
}-u_{f}^{\prime }|^{2})^{1/2},$ which take into account the \ relative
velocity of the charged point particle $\xi $ subject to the reference frame 
$\mathcal{K}^{\prime },$ \ specified by the Euclidean coordinates $%
(t^{\prime },r-r_{f})\in \mathbb{R}^{4},$ and moving simultaneously with
velocity vector $u_{f}\in T(\mathbb{R}^{3})$ \ with respect to the
laboratory reference frame $\mathcal{K},$ specified by the Minkowski
coordinates $(t,r)\in M^{4}$ and related to those of the reference frame $%
\mathcal{K}^{\prime }$ and $\mathcal{K}_{\tau }$ by means of the following
infinitesimal relationships:%
\begin{equation}
dt^{2}=(dt^{\prime })^{2}+|dr_{f}|^{2},\text{ \ }(dt^{\prime })^{2}=d\tau
^{2}+\newline
|dr-dr_{f}|^{2}.  \label{L2.10a}
\end{equation}%
So, it is clear in this case that our charged point particle $\xi $ moves
with the velocity vector $u^{\prime }-u_{f}^{\prime }\in T(\mathbb{R}^{3})$
\ with respect to the reference frame $\mathcal{K}^{\prime }$ in which the
external charged particle $\xi _{f}$ is at rest. Thereby, we have reduced
the problem of deriving the charged point particle $\xi $ \ dynamical
equation to that before solved in Subsection \ \ref{Subsec_2.1}.

Now we can compute the least action variational condition $\delta S=0,$
taking into account that, owing to (\ref{L2.10}), the corresponding
Lagrangian function with respect to the rest reference frame $\mathcal{K}%
_{\tau }$ is given as 
\begin{equation}
\mathcal{L}:=-\bar{W}^{\prime }(1+|\dot{r}-\dot{r}_{f}|^{2})^{1/2}.
\label{L2.11}
\end{equation}%
As a result of simple calculations, the generalized momentum of the charged
particle $\xi $ equals 
\begin{align}
P& :=\partial \mathcal{L}/\partial \dot{r}=-\bar{W}^{\prime }(\dot{r}-\dot{r}%
_{f})(1+|\dot{r}-\dot{r}_{f}|^{2})^{-1/2}=  \label{L2.12} \\
& =-\bar{W}^{\prime }\dot{r}(1+|\dot{r}-\dot{r}_{f}|^{2})^{-1/2}+\bar{W}%
^{\prime }\dot{r}_{f}(1+|\dot{r}-\dot{r}_{f}|^{2})^{-1/2}=  \notag \\
& =m^{\prime }u^{\prime }+\xi A^{\prime }:=p^{\prime }+\xi A^{\prime }=p+\xi
A,  \notag
\end{align}%
where, owing to (\ref{L2.8a}) the vectors $p^{\prime }:=-\bar{W}^{\prime
}u^{\prime }=-\bar{W}u=p\in \mathbb{E}^{3},$ $\ A^{\prime }=\bar{W}^{\prime
}u_{f}^{\prime }=\bar{W}u_{f}=A\in \mathbb{E}^{3},\ $\ and giving rise to
the dynamical equality 
\begin{equation}
\frac{d}{d\tau }(p^{\prime }+\xi A^{\prime })=-\nabla \bar{W}^{\prime }(1+|%
\dot{r}-\dot{r}_{f}|^{2})^{1/2}\   \label{L2.13}
\end{equation}%
with respect to the\ rest reference frame $\mathcal{K}_{\tau }.$ As $%
dt^{\prime }=d\tau (1+|\dot{r}-\dot{r}_{f}|^{2})^{1/2}$ and $(1+|\dot{r}-%
\dot{r}_{f}|^{2})^{1/2}=(1-|u^{\prime }-u_{f}^{\prime }|^{2})^{-1/2},$ we
obtain from (\ref{L2.13}) the equality 
\begin{equation}
\frac{d}{dt^{\prime }}(p^{\prime }+\xi A^{\prime })=-\nabla \bar{W}^{\prime
},  \label{L2.13a}
\end{equation}%
exactly coinciding with \ equality (\ref{L2.9}) subject to the moving
reference frame $\mathcal{K}^{\prime }.$ Now, making use of \ expressions \ (%
\ref{L2.10a}) and \ (\ref{L2.8a}), one can rewrite \ (\ref{L2.13a}) as that
with respect to the laboratory reference frame $\mathcal{K}:$%
\begin{equation}
\begin{array}{c}
\frac{d}{dt^{\prime }}(p^{\prime }+\xi A^{\prime })=-\nabla \bar{W}^{\prime
}\Rightarrow \\ 
\\ 
\Rightarrow \frac{d}{dt^{\prime }}(\frac{-\bar{W}u^{\prime }}{%
(1+|u_{f}^{\prime }|^{2})^{1/2}}+\frac{\xi \bar{W}u_{f}^{\prime }}{%
(1+|u_{f}^{\prime }|^{2})^{1/2}})=-\frac{\nabla \bar{W}}{(1+|u_{f}^{\prime
}|^{2})^{1/2}}\Rightarrow \\ 
\\ 
\Rightarrow \frac{d}{dt^{\prime }}(\frac{-\bar{W}dr}{(1+|u_{f}^{\prime
}|^{2})^{1/2}dt^{\prime }}+\frac{\xi \bar{W}dr_{f}/}{(1+|u_{f}^{\prime
}|^{2})^{1/2}})=-\frac{\nabla \bar{W}}{(1+|u_{f}^{\prime }|^{2})^{1/2}}%
\Rightarrow \\ 
\\ 
\Rightarrow \frac{d}{dt}(-\bar{W}\frac{dr}{dt}+\xi \bar{W}\frac{dr_{f}}{dt}%
)=-\nabla \bar{W}(1-|u_{f}|^{2}),%
\end{array}
\label{L2.13b}
\end{equation}%
exactly coinciding with (\ref{L2.9a}):%
\begin{equation}
\frac{d}{dt}(p+\xi A)=-\nabla \bar{W}(1-|u_{f}|^{2}).  \label{L2.13c}
\end{equation}

\begin{remark}
The equation \ (\ref{L2.13c}) allows to infer the following \ important and
physically reasonable phenomenon: if the test charged point particle
velocity \ $u_{f}\in T(\mathbb{R}^{3})$ tends to the light velocity $c=1,$
the corresponding acceleration force $F_{ac}:=-\nabla \bar{W}(1-|u_{f}|^{2})$
is vanishing. Thereby, the electromagnetic fields, generated by such rapidly
moving charged point particles, have no influence on the dynamics of charged
objects if observed with respect to an arbitrarily chosen laboratory
reference frame $\mathcal{K}.$
\end{remark}

The latter equation \ (\ref{L2.13c}) can be easily rewritten as%
\begin{eqnarray}
dp/dt &=&-\nabla \bar{W}-\xi dA/dt+\nabla \bar{W}|u_{f}|^{2}=  \label{L2.13d}
\\
&=&\xi (-\xi ^{-1}\nabla \bar{W}\ -\ \partial A/\partial t)-\xi <u,\nabla
>A+\xi \nabla <A,u_{f}>,  \notag
\end{eqnarray}%
or, \ using the well-known \cite{LaLi} identity 
\begin{equation}
\nabla <a,b>=<a,\nabla >b+<b,\nabla >a+b\times (\nabla \times a)+a\times
(\nabla \times b),  \label{L2.13e}
\end{equation}%
where $a,b\in \mathbb{E}^{3}$ are arbitrary vector functions, in the
standard Lorentz type form 
\begin{equation}
dp/dt=\xi E+\xi u\times B-\ \nabla <\xi A,u-u_{f}>.  \label{L2.13f}
\end{equation}

The result \ (\ref{L2.13f}), being before found and written down with
respect to the moving reference frame $\mathcal{K}^{\prime }\mathcal{\ }$in 
\cite{BoPrTa,BoPrTaPr-Lore,Repc} and in \cite{MaPi} yet with some
inconsistency, makes it possible to formulate the next important proposition.

\begin{proposition}
The alternative classical relativistic electrodynamic model (\ref{L2.9})
allows the least action formulation based on the action functional (\ref%
{L2.10}) with respect to the rest \ reference frame $\mathcal{K}_{\tau },$
where the Lagrangian function is given by expression (\ref{L2.11}). The
resulting Lorentz type force expression equals \ (\ref{L2.13f}), being
modified by the additional force component $F_{c}:=-\nabla <\xi A,u-u_{f}>,$
\ important for explanation \cite{AhBo,Boye,Tram} of the well known
Aharonov-Bohm effect.
\end{proposition}

\subsection{\label{Subsec_2.3}A moving charged point particle formulation
dual to the classical alternative electrodynamic model}

It is easy to see that the action functional (\ref{L2.10}) is written
utilizing the classical Galilean transformations of reference frames. If we
now consider the action functional (\ref{L2.2}) for a charged point particle
moving with respect the reference frame $\mathcal{K}_{\tau },$ and take into%
\textit{\ }account its interaction with an external magnetic field generated
by the vector potential $A:$ $M^{4}\rightarrow \mathbb{E}^{3},$ it can be
naturally generalized as 
\begin{equation}
S:=\underset{t_{1}}{\overset{t_{2}}{\int }}(-\bar{W}dt+\xi <A,dr>)=\underset{%
\tau _{1}}{\overset{\tau _{2}}{\int }}[-\bar{W}(1+|\dot{r}|^{2})^{1/2}\text{ 
}+\xi <A,\dot{r}>]d\tau ,  \label{L2.14}
\end{equation}%
where $d\tau =dt(1-|u|^{2})^{1/2}.$

Thus, the corresponding common particle-field momentum takes the form 
\begin{align}
P& :=\partial \mathcal{L}/\partial \dot{r}=-\bar{W}\dot{r}(1+|\dot{r}%
|^{2})^{-1/2}+\xi A=  \label{L2.15} \\
& =mu+\xi A:=p+\xi A,  \notag
\end{align}%
and satisfies%
\begin{align}
\dot{P}& :=dP/d\tau =\partial \mathcal{L}/\partial r=-\nabla \bar{W}(1+|\dot{%
r}|^{2})^{1/2}\text{ }+\xi \nabla <A,\dot{r}>=  \label{L2.16} \\
& =-\nabla \bar{W}(1-|u|^{2})^{-1/2}+\xi \nabla <A,u>(1-|u|^{2})^{-1/2}, 
\notag
\end{align}%
where 
\begin{equation}
\mathcal{L}:=-\bar{W}(1+|\dot{r}|^{2})^{1/2}\text{ }+\xi <A,\dot{r}>
\label{L2.16a}
\end{equation}%
is the corresponding Lagrangian function. Since $d\tau =dt(1-|u|^{2})^{1/2},$
one easily finds from (\ref{L2.16}) that 
\begin{equation}
dP/dt=-\nabla \bar{W}+\xi \nabla <A,u>.  \label{L2.17}
\end{equation}%
Upon substituting (\ref{L2.15}) into (\ref{L2.17}) and making use of the
well-known \cite{LaLi} identity 
\begin{equation}
\nabla <a,b>=<a,\nabla >b+<b,\nabla >a+b\times (\nabla \times a)+a\times
(\nabla \times b),  \label{L2.18}
\end{equation}%
where $a,b\in \mathbb{E}^{3}$ are arbitrary vector functions, we obtain the
classical expression for the Lorentz force $F,$ acting on the moving charged
point particle $\xi :$%
\begin{equation}
dp/dt:=F=\xi E+\xi u\times B,  \label{L2.19}
\end{equation}%
where, by definition, 
\begin{equation}
E:=-\xi ^{-1}\nabla \bar{W}-\partial A/\partial t  \label{L2.20}
\end{equation}%
is its associated electric field and 
\begin{equation}
B:=\nabla \times A  \label{L2.21}
\end{equation}%
is the corresponding magnetic field. This result can be summarized as
follows:

\begin{proposition}
The classical relativistic Lorentz force (\ref{L2.19}) allows the least
action formulation (\ref{L2.14}) with respect to the rest \ \ reference
frame variables, where the Lagrangian function is given by formula (\ref%
{L2.16a}). Yet its electrodynamics, described by the Lorentz force (\ref%
{L2.19}), is not equivalent to the classical relativistic moving point
particle electrodynamics, described by means of the Lorentz force (\ref%
{L1.11}), as the inertial mass expression $m=-\bar{W}$ does not coincide
with that of \ (\ref{L1.3}).
\end{proposition}

Expressions (\ref{L2.19}) and (\ref{L2.13f}) are equal up to the gradient
term $F_{c}:=-\xi \nabla <A,u-u_{f}>,$ which reconciles the Lorentz forces
acting on a charged moving particle $\xi $ with respect to different
reference frames. This fact is important for our vacuum field theory
approach since it uses no special geometry and makes it possible to analyze
both electromagnetic and gravitational fields simultaneously by employing
the new definition of the dynamical mass by means of expression (\ref{L2.5}).

\subsection{\label{Subsec_2.4}The vacuum field theory electrodynamics
equations: Hamiltonian analysis}

Any Lagrangian theory has an equivalent canonical Hamiltonian representation
via the classical Legendre transformation \ \cite{Arno,Thir,AbMa,HePrPr,PrMy}%
. As we have already formulated our vacuum field theory of a moving charged
particle $\xi $ in Lagrangian form, we proceed now to its Hamiltonian
analysis making use of the action functionals (\ref{L2.2}), (\ref{L2.11})
and (\ref{L2.14}).

Take, first, the Lagrangian function (\ref{L2.4}) and the momentum
expression (\ref{L2.3}) for defining the corresponding Hamiltonian function
with respect to the moving reference frame $\mathcal{K}_{\tau }:\mathcal{\ }$
\begin{align}
H& :=<p,\dot{r}>-\mathcal{L}=  \notag \\
& =-<p,p>\bar{W}^{-1}(1-|p|^{2}/\bar{W}^{2})^{-1/2}+\bar{W}(1-|p|^{2}/\bar{W}%
^{2})^{-1/2}=  \notag \\
& =-|p|^{2}\bar{W}^{-1}(1-|p|^{2}/\bar{W}^{2})^{-1/2}+\bar{W}^{2}\bar{W}%
^{-1}(1-|p|^{2}/\bar{W}^{2})^{-1/2}=  \label{L3.1} \\
& =-(\bar{W}^{2}-|p|^{2})(\bar{W}^{2}-|p|^{2})^{-1/2}=-(\bar{W}%
^{2}-|p|^{2})^{1/2}.  \notag
\end{align}%
Consequently, it is easy to show \cite{AbMa,Arno,BlPrSa,Thir,PrMy} \ that
the Hamiltonian function (\ref{L3.1}) is a conservation law of the dynamical
field equation (\ref{L2.1}), that is for all $\tau ,t\in \mathbb{R}$%
\begin{equation}
dH/d\tau =dH/dt=0,  \label{L3.2}
\end{equation}%
which naturally leads to an energy interpretation of $H$. Thus, we can
represent the particle energy as 
\begin{equation}
\mathcal{E}=(\bar{W}^{2}-|p|^{2})^{1/2}.  \label{L3.3}
\end{equation}%
Accordingly the Hamiltonian equivalent to the vacuum field equation (\ref%
{L2.1}) can be written as%
\begin{align}
\dot{r}& :=dr/d\tau =\partial H/\partial p=p(\bar{W}^{2}-|p|^{2})^{-1/2}
\label{L3.4} \\
\dot{p}& :=dp/d\tau =-\partial H/\partial r=\bar{W}\nabla \bar{W}(\bar{W}%
^{2}-|p|^{2})^{-1/2},  \notag
\end{align}%
and we have the following result.

\begin{proposition}
The alternative freely moving point particle electrodynamic model (\ref{L2.1}%
) allows the canonical Hamiltonian formulation (\ref{L3.4}) with respect to
the "rest\textquotedblright\ \ reference frame variables, where the
Hamiltonian function is given by expression (\ref{L3.1}). Its
electrodynamics is completely equivalent to the classical relativistic
freely moving point particle electrodynamics described in Subsection \ \ref%
{Subsec_2.1}.
\end{proposition}

In the analogous manner, one can now use the Lagrangian (\ref{L2.11}) to
construct the Hamiltonian function for the dynamical field equation (\ref%
{2.9}), describing the motion of charged particle $\xi $ in an external
electromagnetic field in the canonical Hamiltonian form:%
\begin{equation}
\dot{r}:=dr/d\tau =\partial H/\partial P,\text{ \ \ \ \ \ }\dot{P}:=dP/d\tau
=-\partial H/\partial r,  \label{L3.5a}
\end{equation}%
where%
\begin{align}
H& :=<P,\dot{r}>-\mathcal{L}=  \notag \\
& =<P,\dot{r}_{f}-P\bar{W}^{\prime ,-1}(1-|P|^{2}/\bar{W}^{\prime
,2})^{-1/2}>+\bar{W}^{\prime }[\bar{W}^{\prime ,2}(\bar{W}^{\prime
,2}-|P|^{2})^{-1}]^{1/2}=  \notag \\
& =<P,\dot{r}_{f}>+|P|^{2}(\bar{W}^{\prime ,2}-|P|^{2})^{-1/2}-\bar{W}%
^{\prime ,2}(\bar{W}^{\prime ,2}-|P|^{2})^{-1/2}=  \notag \\
& =-(\bar{W}^{\prime ,2}-|P|^{2})(\bar{W}^{\prime ,2}-|P|^{2})^{-1/2}+<P,%
\dot{r}_{f}>=  \label{L3.5} \\
& =-(\bar{W}^{\prime ,2}-|P|^{2})^{1/2}-\xi <A^{\prime },P>(\bar{W}^{\prime
,2}-|P|^{2})^{-1/2}=  \notag \\
& =-(\bar{W}^{2}-|\xi A|^{2}-|P|^{2})^{1/2}-\xi <A,P>(\bar{W}^{2}-|\xi
A|^{2}-|P|^{2})^{-1/2},  \notag
\end{align}%
being written with respect to the laboratory reference frame $\mathcal{K}.$
\ Here we took into account that, owing to definitions (\ref{L2.8}), \ (\ref%
{L2.8a}) and (\ref{2.12}),%
\begin{align}
\xi A^{\prime }& :=\bar{W}^{\prime }u_{f}^{\prime }=\bar{W}^{\prime
}dr_{f}/dt^{\prime }=\xi A=  \label{L3.6} \\
& =\bar{W}^{\prime }\frac{dr_{f}}{d\tau }\cdot \frac{d\tau }{dt^{\prime }}=%
\bar{W}^{\prime }\dot{r}_{f}(1-|u-u_{f}|)^{1/2}=  \notag \\
& =\bar{W}^{\prime }\dot{r}_{f}(1+|\dot{r}-\dot{r}_{f}|^{2})^{-1/2}=  \notag
\\
& =-\bar{W}^{\prime }\dot{r}_{f}(\bar{W}^{\prime ,2}-|P|^{2})^{1/2}\bar{W}%
^{\prime ,-1}=-\dot{r}_{f}(\bar{W}^{\prime ,2}-|P|^{2})^{1/2}\ ,  \notag
\end{align}%
and, in particular, 
\begin{equation}
\dot{r}_{f}=-\xi A(\bar{W}^{^{\prime },2}-|P|^{2})^{-1/2},\text{ \ }\bar{W}=%
\bar{W}^{\prime }(1-|u_{f}|^{2})^{-1/2},  \label{L3.7}
\end{equation}%
where $A:M^{4}\mathbb{\rightarrow R}^{3}$ is the related magnetic vector
potential generated by the moving external charged particle $\xi _{f}.$
Equations (\ref{L3.5a}) can be rewritten with respect to the laboratory
reference frame $\mathcal{K}_{t}$ \ in the form 
\begin{equation}
dr/dt=u,\text{ \ \ \ }dp/dt=\xi E+\xi u\times B-\xi \nabla <A,u-u_{f}>,
\label{L3.7a}
\end{equation}%
which coincides with the result (\ref{L2.13f}).

Whence, we see that the Hamiltonian function (\ref{L3.5}) satisfies the
energy conservation conditions%
\begin{equation}
dH/d\tau =dH/dt^{\prime }=dH/dt=0,  \label{L3.8}
\end{equation}%
for all $\tau ,t^{\prime }$ and $t\in \mathbb{R},$ and that the suitable
energy expression is%
\begin{equation}
\mathcal{E}=(\bar{W}^{2}\ -\xi ^{2}|A|^{2}-|P|^{2})^{1/2}+\xi <A,P>(\bar{W}%
^{2}\ -\xi ^{2}|A|^{2}-|P|^{2})^{-1/2},  \label{L3.9}
\end{equation}%
where the generalized momentum $P=p+\xi A.$ The result (\ref{L3.9}) differs
essentially from that obtained in \cite{LaLi}, which makes use of the
Einsteinian Lagrangian for a moving charged point particle $\xi $ in an
external electromagnetic field. Thus, we obtain the following proposition:

\begin{proposition}
The alternative classical relativistic electrodynamic model (\ref{L3.7a}),
which is intrinsically compatible with the classical Maxwell equations \ (%
\ref{L0.6}), allows the Hamiltonian formulation (\ref{L3.5a}) with respect
to the rest reference frame variables, where the Hamiltonian function is
given by expression (\ref{L3.5}).
\end{proposition}

The inference above is a natural candidate for experimental validation of
our theory. It is strongly motivated by the following remark.

\begin{remark}
It is necessary to mention here that the Lorentz force expression (\ref%
{L3.7a}) uses the particle momentum $p=mu,$ where the dynamical
\textquotedblleft mass\textquotedblright $m:=-\bar{W}$ satisfies condition (%
\ref{L3.9}). This gives rise to the following crucial relationship between
the particle energy $\mathcal{E}_{0}$ and its rest mass $m_{0}=$ $-\bar{W}%
_{0}$ (for the velocity $u\ =0$ at the initial time moment $t=0):$%
\begin{equation}
\mathcal{E}_{0}=m_{0}\frac{(1-|\xi A_{0}/m_{0}|^{2})}{(1-2|\xi
A_{0}/m_{0}|^{2})^{1/2}}\ ,  \label{L3.9a}
\end{equation}%
or, equivalently, under the condition $\ |\xi A_{0}/m_{0}|^{2}<1/2$ 
\begin{equation}
m_{0}=\mathcal{E}_{0}\left( \frac{1}{2}+|\xi A_{0}/\mathcal{E}_{0}|^{2}\pm 
\frac{1}{2}\sqrt{1-4|\xi A_{0}/\mathcal{E}_{0}|^{2}}\right) ^{1/2},
\label{L3.9b}
\end{equation}%
where $A_{0}:=A|_{t=0}\in \mathbb{E}^{3},$ which differs markedly from the
classical expression $m_{0}=\mathcal{E}_{0}-\xi \varphi _{0},$ following
from (\ref{L1.10}) and is does not a priori on the external potential energy 
$\xi \varphi _{0}.$ As the quantity $|\xi A_{0}/\mathcal{E}_{0}|\rightarrow
0 $ if the energy modul $|\mathcal{E}_{0}|\rightarrow \infty ,$ the
following asymptotic mass values follow from (\ref{L3.9b}):%
\begin{equation}
\bar{m}_{0}\simeq \ \mathcal{E}_{0},\text{ \ \ \ \ \ \ }m_{0}^{(\pm )}\simeq
\ \pm \sqrt{2}|\xi A_{0}|.  \label{L3.9c}
\end{equation}%
The first mass value $\bar{m}_{0}\simeq \ \mathcal{E}_{0}$ is looking from
the relativistic physics standard, yet the second mass values $m_{0}^{(\pm
)}\simeq \ \pm \sqrt{2}|\xi A_{0}|$ \ give rise to existence at large enough
energies of charged particle excitations of the vacuo with both positive and
negative mass values.
\end{remark}

To make this difference more clear, we now analyze the Lorentz force (\ref%
{L2.19}) from the Hamiltonian point of view based on the Lagrangian function
(\ref{L2.16a}). Thus, we obtain that the corresponding Hamiltonian function%
\begin{align}
H& :=<P,\dot{r}>-\mathcal{L}=<P,\dot{r}>+\bar{W}(1+|\dot{r}|^{2})^{1/2}-\xi
<A,\dot{r}>=  \label{L3.10} \\
& =<P-\xi A,\dot{r}>+\bar{W}(1+|\dot{r}|^{2})^{1/2}=  \notag \\
& =-<p,p>\bar{W}^{-1}(1-|p|^{2}/\bar{W}^{2})^{-1/2}+\bar{W}(1-|p|^{2}/\bar{W}%
^{2})^{-1/2}=  \notag \\
& =-(\bar{W}^{2}-|p|^{2})(\bar{W}^{2}-|p|^{2})^{-1/2}=-(\bar{W}%
^{2}-|p|^{2})^{1/2}.  \notag
\end{align}%
Since $p=P-\xi A,$ expression (\ref{L3.10}) assumes the final "\textit{no
interaction}" \ \cite{LaLi,Paul,Kupe,PrSaPr} form 
\begin{equation}
H=-(\bar{W}^{2}-|P-\xi A|^{2})^{1/2},  \label{L3.11}
\end{equation}%
which is conserved with respect to the evolution equations (\ref{L2.15}) and
(\ref{L2.16}), that is 
\begin{equation}
\ dH/d\tau =dH/dt=0  \label{L3.11a}
\end{equation}%
for all $\tau ,t\in \mathbb{R}.$ These equations latter are equivalent to
the following Hamiltonian system%
\begin{align}
\dot{r}& =\partial H/\partial P=(P-\xi A)(\bar{W}^{2}-|P-\xi A|^{2})^{-1/2},
\label{L3.12} \\
\dot{P}& =-\partial H/\partial r=(\bar{W}\nabla \bar{W}-\nabla <\xi A,(P-\xi
A)>)(\bar{W}^{2}-|P-\xi A|^{2})^{-1/2},  \notag
\end{align}%
as one can readily check by direct calculations. Actually, the first equation%
\begin{align}
\dot{r}& =(P-\xi A)(\bar{W}^{2}-|P-\xi A|^{2})^{-1/2}=p(\bar{W}%
^{2}-|p|^{2})^{-1/2}=  \label{L3.13} \\
& =mu(\bar{W}^{2}-|p|^{2})^{-1/2}=-\bar{W}u(\bar{W}%
^{2}-|p|^{2})^{-1/2}=u(1-|u|^{2})^{-1/2},  \notag
\end{align}%
holds, owing to the condition $d\tau =dt(1-|u|^{2})^{1/2}$ and definitions $%
p:=mu,$ $m=-\bar{W},$ postulated from the very beginning. Similarly we
obtain that 
\begin{align}
\dot{P}& =-\nabla \bar{W}(1-|p|^{2}/\bar{W}^{2})^{-1/2}+\nabla <\xi
A,u>(1-|p|^{2}/\bar{W}^{2})^{-1/2}=  \label{L3.14} \\
& =-\nabla \bar{W}(1-|u|^{2})^{-1/2}+\nabla <\xi A,u>(1-|u|^{2})^{-1/2}, 
\notag
\end{align}%
coincides with equation (\ref{L2.17}) in the evolution parameter $t\in 
\mathbb{R}.$ This can be formulated as the next result.

\begin{proposition}
The dual to the classical relativistic electrodynamic \ model (\ref{L2.19})
allows the canonical Hamiltonian formulation (\ref{L3.12}) with respect to
the rest reference frame variables, where the Hamiltonian function is given
by expression (\ref{L3.11}). Moreover, this formulation circumvents the
"mass-potential energy" controversy attached to the classical
electrodynamical model (\ref{L1.8}).
\end{proposition}

The modified Lorentz force expression (\ref{L2.19}) and the related rest
energy relationship are characterized by the following remark.

\begin{remark}
If we make use of the modified relativistic Lorentz force expression (\ref%
{L2.19}) as an alternative to the classical one of (\ref{L1.11}), the
corresponding charged particle $\xi $ energy expression (\ref{L3.11}) also
gives rise to a true physically reasonable energy expression (at the
velocity $u:=0\in \mathbb{E}^{3}$ at the initial time moment $t=0$); namely, 
$\mathcal{E}_{0}=m_{0}$ instead of the physically controversial classical
expression $\mathcal{E}_{0}=m_{0}+\xi \varphi _{0},$ where $\varphi
_{0}:=\varphi |_{t=0},$ corresponding to the case (\ref{L1.10}).
\end{remark}

\subsection{\label{Subsec_2.5}The quantization of electrodynamics models
within the vacuum field theory approach}

\subsubsection{The problem setting}

Recently \cite{BoPrTa,BoPrTaPr-Lore} \ we devised a new regular no-geometry
approach to deriving the electrodynamics of a moving charged point particle $%
\xi $ in an external electromagnetic field from first principles. This
approach has, in part, reconciled the mass-energy controversy \cite{Bril} in
classical relativistic electrodynamics. Using the vacuum field theory
approach initially proposed in \cite{BoPrTa,BoPrTaPr-Lore,Repc}, we
reanalyzed this problem above both from the Lagrangian and Hamiltonian
perspective and derived key expressions for the corresponding energy
functions and Lorentz type forces acting on a moving charged point particle $%
\xi .$

Since all of our electrodynamics models were represented here in canonical
Hamiltonian form, they are well suited to the application of Dirac
quantization \cite{Dira,BoSh} and the corresponding derivation of related
Schr\"{o}dinger type evolution equations. We describe these procedures in
the section we proceed below.

\subsubsection{Free point particle electrodynamics model and its quantization%
}

The charged point particle electrodynamics models, discussed in detail in
Sections \ \ref{Subsec_2.2} and \ \ref{Subsec_2.3}, were also considered in 
\cite{BoPrTaPr-Lore} from the dynamical point of view, where a Dirac
quantization of the corresponding conserved energy expressions was
attempted. However, from the canonical point of view, \ the true
quantization procedure should be based on the relevant canonical Hamiltonian
formulation of the models given in (\ref{L3.4}), (\ref{L3.5a}) and (\ref%
{L3.12}).

In particular, consider a free charged \ point particle electrodynamics
model characterized by (\ref{L3.4}) and having the Hamiltonian equations%
\begin{align}
dr/d\tau & :=\partial H/\partial p=-p(\bar{W}^{2}-|p|^{2})^{-1/2},
\label{L5.1} \\
dp/d\tau & :=-\partial H/\partial r=-\bar{W}\nabla \bar{W}(\bar{W}%
^{2}-|p|^{2})^{-1/2},  \notag
\end{align}%
where $\bar{W}:M^{4}\rightarrow \mathbb{R}$ defined in the preceding
sections is the corresponding vacuum field potential characterizing medium
field structure, $(r,p)\in T^{\ast }(\mathbb{R}^{3})\simeq \mathbb{E}%
^{3}\times \mathbb{E}^{3}$ are the standard canonical coordinate-momentum
variables on the cotangent space $T^{\ast }(\mathbb{R}^{3}),$ $\tau \in 
\mathbb{R},$ \ is the proper rest reference frame $\mathcal{K}_{\tau }$ time
parameter of the moving particle, and $H:T^{\ast }(\mathbb{R}%
^{3})\rightarrow \mathbb{R}$ is the Hamiltonian function%
\begin{equation}
H:=-(\bar{W}^{2}-|p|^{2})^{1/2},  \label{L5.2}
\end{equation}%
expressed here and hereafter in light speed units. The rest reference frame $%
\mathcal{K}_{\tau },$ parameterized by variables $(\tau ,r)\in \mathbb{E}%
^{4},$ is related to any other reference frame $\mathcal{K}_{t}$ \ in which
our charged point particle $\xi $ \ moves with velocity vector $u\in \mathbb{%
E}^{3}.$ \ The frame $\mathcal{K}_{t}$ \ is parameterized by variables $%
(t,r)\in M^{4}$ via the Euclidean infinitesimal relationship%
\begin{equation}
dt^{2}=d\tau ^{2}+|dr|^{2},  \label{L5.3}
\end{equation}%
which is equivalent to the Minkowskian infinitesimal relationship%
\begin{equation}
d\tau ^{2}=dt^{2}-|dr|^{2}.  \label{L5.4}
\end{equation}%
The Hamiltonian function (\ref{L5.2}) clearly satisfies the energy
conservation conditions 
\begin{equation}
dH/d\tau =dH/dt=0\   \label{L5.5}
\end{equation}%
for all $t,\tau \in \mathbb{R}.$ This means that the suitable energy 
\begin{equation}
\mathcal{E}=(\bar{W}^{2}-|p|^{2})^{1/2}  \label{L5.6}
\end{equation}%
can be treated by means of the Dirac quantization scheme \cite{Dira,Dira-1}
to obtain, as $\hbar \rightarrow 0,$ (or the light speed $c\rightarrow
\infty $) the governing Schr\"{o}dinger type dynamical equation. To do this
following the approach in \cite{BoPrTa,BoPrTaPr-Lore}, we need to make
canonical operator replacements $\mathcal{E}\rightarrow \hat{\mathcal{E}}:=-%
\frac{\hbar }{i}\frac{\partial }{\partial \tau },$ \ $p\rightarrow \hat{p}:=%
\frac{\hbar }{i}\nabla ,$ as $\hbar \rightarrow 0,$ in the following energy
expression:%
\begin{equation}
\mathcal{E}^{2}:=(\hat{\mathcal{E}}\psi ,\hat{\mathcal{E}}\psi )=(\psi ,\hat{%
\mathcal{E}}^{2}\psi )=(\psi ,\hat{H}^{+}\hat{H}\psi ),  \label{L5.7}
\end{equation}%
where $(\cdot ,\cdot )$ is the standard $L_{2}$ - inner product. It follows
from (\ref{L5.6}) that 
\begin{equation}
\hat{\mathcal{E}}^{2}=\bar{W}^{2}-|p|^{2}=\hat{H}^{+}\hat{H}  \label{L5.8}
\end{equation}%
is a suitable operator factorization in the Hilbert space $\mathcal{H}%
:=L_{2}(\mathbb{R}^{3};\mathbb{C})$ and $\psi \in \mathcal{H}$ \ is the
corresponding normalized quantum vector state. Since the following
elementary identity 
\begin{equation}
\bar{W}^{2}-|p|^{2}=\bar{W}(1-\bar{W}^{-1}|p|^{2}\bar{W}^{-1})^{1/2}(1-\bar{W%
}^{-1}|p|^{2}\bar{W}^{-1})^{1/2}\bar{W}  \label{L5.9}
\end{equation}%
holds, we can use (\ref{L5.8}) and (\ref{L5.9}) to define the operator 
\begin{equation}
\hat{H}:=(1-\bar{W}^{-1}|p|^{2}\bar{W}^{-1})^{1/2}\bar{W}.  \label{L5.10}
\end{equation}%
Having calculated the operator expression (\ref{L5.10}) as $\hbar
\rightarrow 0$ up to operator accuracy $O($ $\hbar ^{4}),$ \ it is easy see
that 
\begin{equation}
\hat{H}=\frac{|p|^{2}}{2m(u)}+\bar{W}:=-\frac{\hbar ^{2}}{2m(u)}\nabla ^{2}+%
\bar{W},  \label{L5.11}
\end{equation}%
where we have taken into account the dynamical mass definition $m(u):=-\bar{W%
}$ (in the light speed units). Consequently, using (\ref{L5.7}) and (\ref%
{L5.11}), we obtain up to operator accuracy $O($ $\hbar ^{4})$ the following
Schr\"{o}dinger type evolution equation%
\begin{equation}
i\hbar \frac{\partial \psi }{\partial \tau }:=\hat{\mathcal{E}}\psi =\hat{H}%
\psi =-\frac{\hbar ^{2}}{2m(u)}\nabla ^{2}\psi +\bar{W}\psi  \label{L5.12}
\end{equation}%
with respect to the rest \ reference frame $\mathcal{K}_{\tau }$ evolution
parameter $\tau \in \mathbb{R}.$ For a related evolution parameter $t\in 
\mathbb{R}$ parameterizing a reference frame $\mathcal{K},$ the equation (%
\ref{L5.12}) takes the form%
\begin{equation}
i\hbar \frac{\partial \psi }{\partial t}=-\frac{\hbar ^{2}m_{0}}{2m(u)^{2}}%
\nabla ^{2}\psi -m_{0}\psi .  \label{L5.13}
\end{equation}%
Here we used the fact that it follows from (\ref{L5.6}) that the classical
mass relationship 
\begin{equation}
m(u)=m_{0}(1-|u|^{2})^{-1/2}  \label{L5.14}
\end{equation}%
holds, where $m_{0}\in \mathbb{R}_{+}$ is the corresponding rest mass of our
point particle $\xi .$

The linear Schr\"{o}dinger equation (\ref{L5.13}) for the case $\hbar
/c\rightarrow 0$ actually coincides with the well-known expression \cite%
{LaLi,Dira,Feyn-1} from classical quantum mechanics.

\subsubsection{Classical charged point particle electrodynamics model and
its quantization}

We start here from the first vacuum field theory reformulation of the
classical charged point particle electrodynamics (introduced in Subsection %
\ref{Subsec_2.1} ) and based on the conserved Hamiltonian function (\ref%
{L3.11}) 
\begin{equation}
H:=-(\bar{W}^{2}-|P-\xi A|^{2})^{1/2},  \label{L6.1}
\end{equation}%
where $\xi \in $ $\mathbb{R}$ is the particle charge, $(\bar{W},A)\in 
\mathbb{R\times E}^{3}$ is the corresponding representation of the
electromagnetic field potentials and $P\in \mathbb{E}^{3}$ is the common
generalized particle-field momentum 
\begin{equation}
P:=p+\xi A,\text{ \ \ \ \ }p:=mu,  \label{L6.2}
\end{equation}%
which satisfies the classical Lorentz force equation. Here $m:=-\bar{W}$ is
the observable dynamical mass of our charged particle, and $u\in \mathbb{E}%
^{3}$ is its velocity vector with respect to a chosen reference frame $%
\mathcal{K},$ all expressed in light speed units.

Our electrodynamics based on (\ref{L6.1}) is canonically Hamiltonian, so the
Dirac quantization scheme 
\begin{equation}
P\rightarrow \hat{P}:=\frac{\hbar }{i}\nabla ,\text{ \ \ \ \ \ \ }\mathcal{E}%
\rightarrow \hat{\mathcal{E}}:=-\frac{\hbar }{i}\frac{\partial }{\partial
\tau }  \label{L6.3}
\end{equation}%
should be applied to the energy expression 
\begin{equation}
\mathcal{E}:=(\bar{W}^{2}-|P-\xi A|^{2})^{1/2},  \label{L6.4}
\end{equation}%
following from the conservation conditions 
\begin{equation}
dH/dt=0=dH/d\tau ,  \label{L6.5}
\end{equation}%
satisfied for all $\tau ,t\in \mathbb{R}.$

Proceeding as above, we can factorize the operator $\hat{\mathcal{E}}^{2}$
as 
\begin{equation*}
\begin{array}{c}
\bar{W}^{2}-|\hat{P}-\xi A|^{2}=\bar{W}(1-\bar{W}^{-1}|\hat{P}-\xi A|^{2}%
\bar{W})^{1/2}\times \\ 
\times (1-\bar{W}^{-1}|\hat{P}-\xi A|^{2}\bar{W}^{-1})^{1/2}\bar{W}:=\hat{H}%
^{+}\hat{H},%
\end{array}%
\end{equation*}%
where (as $\hbar /c\rightarrow 0,$ $\hbar c=const$)%
\begin{equation}
\hat{H}:=\frac{1}{2m(u)}|\frac{\hbar }{i}\nabla -\xi A|^{2}+\bar{W}
\label{L6.7}
\end{equation}%
up to operator accuracy $O(\hbar ^{4}).$ Hence, the related Schr\"{o}dinger
type evolution equation in the Hilbert space $\mathcal{H}=L_{2}(\mathbb{R}%
^{3};\mathbb{C})$ is 
\begin{equation}
i\hbar \frac{\partial \psi }{\partial \tau }:=\hat{\mathcal{E}}\psi =\hat{H}%
\psi =\frac{1}{2m(u)}|\frac{\hbar }{i}\nabla -\xi A|^{2}\psi +\bar{W}\psi
\label{L6.8}
\end{equation}%
with respect to the rest reference frame $\mathcal{K}_{\tau }$ evolution
parameter $\tau \in \mathbb{R}$, and corresponding Schr\"{o}dinger type
evolution equation with respect to the evolution parameter $t\in \mathbb{R}$
takes the form 
\begin{equation}
i\hbar \frac{\partial \psi }{\partial t}=\frac{m_{0}}{2m(u)^{2}}|\frac{\hbar 
}{i}\nabla -\xi A|^{2}\psi -m_{0}\psi .  \label{L6.9}
\end{equation}%
The Schr\"{o}dinger equation (\ref{L6.8}) (as $\hbar /c\rightarrow 0$)
coincides \cite{LaLi-1,Dira} with the classical quantum mechanics version.

\subsubsection{Modified charged point particle electrodynamics model and its
quantization}

From the canonical viewpoint, \ we now turn to the true quantization
procedure for the electrodynamics model, characterized by (\ref{L2.13}) and
having the Hamiltonian function (\ref{L3.5}) 
\begin{equation}
H:=-(\bar{W}^{2}-\xi ^{2}|A|^{2}-|P|^{2})^{1/2}-\xi <A,P>(\bar{W}^{2}-\xi
^{2}|A|^{2}-|P|^{2})^{-1/2}.  \label{L7.1}
\end{equation}%
Accordingly the suitable energy function is 
\begin{equation}
\mathcal{E}:=(\bar{W}^{2}-\xi ^{2}|A|^{2}-|P|^{2})^{1/2}+\xi <A,P>(\bar{W}%
^{2}-\xi ^{2}|A|^{2}-|P|^{2})^{-1/2},  \label{L7.2}
\end{equation}%
where, as before, 
\begin{equation}
P:=p+\xi A,\text{ \ \ \ \ }p:=mu,\text{ \ }m:=-\bar{W},  \label{L7.3}
\end{equation}%
is a conserved quantity for (\ref{L2.13}), which we will canonically
quantize via the Dirac procedure (\ref{L6.3}). Toward this end, let us
consider the quantum condition%
\begin{equation}
\mathcal{E}^{2}:=(\hat{\mathcal{E}}\psi ,\hat{\mathcal{E}}\psi )=(\psi ,\hat{%
\mathcal{E}}^{2}\psi ),\text{ \ \ \ \ \ }(\psi ,\psi ):=1,  \label{L7.4}
\end{equation}%
where, by definition, $\hat{\mathcal{E}}:=-\frac{\hbar }{i}\frac{\partial }{%
\partial t}$ and $\psi \in \mathcal{H}=L_{2}(\mathbb{R}^{3};\mathbb{C})$ is
a respectively normalized quantum state vector. Making now use of the energy
function (\ref{L7.2}), one readily computes that 
\begin{equation}
\mathcal{E}^{2}=\bar{W}^{2}-|P-\xi A|^{2}+\xi ^{2}<A,P>(\bar{W}%
^{2}-|P|^{2})^{-1}<P,A>,  \notag
\end{equation}%
which transforms by the canonical Dirac type quantization $P\rightarrow \hat{%
P}:=\frac{\hbar }{i}\nabla $ into the symmetrized operator expression 
\begin{equation}
\hat{\mathcal{E}}^{2}=\bar{W}^{2}-|\hat{P}-\xi A|^{2}+\xi ^{2}<A,\hat{P}>(%
\bar{W}^{2}-|\hat{P}|^{2})^{-1}<\hat{P},A>.  \label{L7.6}
\end{equation}%
Factorizing the operator (\ref{L7.6})\ in the form $\hat{\mathcal{E}}^{2}=%
\hat{H}^{+}\hat{H},$ and retaining only terms up to $O(\hbar ^{4})$ (as $%
\hbar /c\rightarrow 0$), we compute that 
\begin{equation}
\hat{H}:=\frac{1}{2m(u)}|\frac{\hbar }{i}\nabla -\xi A|^{2}-\frac{\xi ^{2}}{%
2m^{3}(u)}<A,\frac{\hbar }{i}\nabla ><\frac{\hbar }{i}\nabla ,A>,
\label{L7.7}
\end{equation}%
where, as before, $m(u)=-\bar{W}$ in light speed units. Thus, owing to (\ref%
{L7.4}) and (\ref{L7.7}), the resulting Schr\"{o}dinger evolution equation
is 
\begin{equation}
i\hbar \frac{\partial \psi }{\partial \tau }:=\hat{H}\psi =\frac{1}{2m(u)}|%
\frac{\hbar }{i}\nabla -\xi A|^{2}\psi -\frac{\xi ^{2}}{2m^{3}(u)}<A,\frac{%
\hbar }{i}\nabla ><\frac{\hbar }{i}\nabla ,A>\psi  \label{L7.8}
\end{equation}%
with respect to the rest reference frame proper evolution parameter $\tau
\in \mathbb{R}.$ The latter can be rewritten in the equivalent form as

\begin{eqnarray}
i\hbar \frac{\partial \psi }{\partial \tau } &=&-\frac{\hbar ^{2}}{2m(u)}%
\Delta \psi -\frac{1}{2m(u)}<[\frac{\hbar }{i}\nabla ,\xi A]_{+}>\psi -
\label{L7.8a} \\
-\frac{\xi ^{2}}{2m^{3}(u)} &<&A,\frac{\hbar }{i}\nabla ><\frac{\hbar }{i}%
\nabla ,A>\psi ,  \notag
\end{eqnarray}%
where $[\cdot ,\cdot ]_{+}$ means the formal anti-commutator of operators.
Similarly one also obtains the related Schr\"{o}dinger equation with respect
to the time parameter $t\in \mathbb{R},$ which we shall not dwell upon here.
The result (\ref{L7.8}) only slightly differs from the classical Schr\"{o}%
dinger evolution equation (\ref{L6.8}). Simultaneously, its form (\ref{L7.8a}%
) almost completely coincides with the classical ones from \cite%
{LaLi-1,Paul,Dira} modulo the evolution considered with respect to the rest
reference time parameter $\tau \in \mathbb{R}.$ This suggests that we must
more thoroughly reexamine the physical motivation of the principles
underlying the classical electrodynamic models, described by the Hamiltonian
functions (\ref{L6.1}) and (\ref{L7.1}), giving rise to different Lorentz
type force expressions. A more deeply considered and extended analysis of
this matter is forthcoming in a paper now in preparation.

\begin{remark}
All of dynamical field equations discussed above are canonical Hamiltonian
systems with respect to the corresponding proper \ rest reference frames $%
\mathcal{K}_{\tau },$ parameterized by suitable time parameters $\tau \in 
\mathbb{R}.$ Upon passing to the basic laboratory reference frame $\mathcal{K%
}_{t}$ \ with the time parameter $t\in \mathbb{R},$naturally the related
Hamiltonian structure is lost, giving rise to a new interpretation of the
real particle motion. Namely, one that has an absolute sense only with
respect to the proper rest \ reference system, and otherwise being
completely relative with respect to all other reference frames. As for the
Hamiltonian expressions (\ref{L3.1}), (\ref{L3.5}) and (\ref{L3.11}), one
observes that they all depend strongly on the vacuum potential energy field
function $\bar{W}:M^{4}\mathbb{\rightarrow R},$ thereby avoiding the mass
problem of the classical energy expression pointed out by L. Brillouin \cite%
{Bril}. It should be noted that the canonical Dirac quantization procedure
can be applied only to the corresponding dynamical field systems considered
with respect to their proper rest reference frames. Some comments are in
order concerning the classical relativity principle. We have obtained our
results relying only on the natural notion of the rest reference frame and
its suitable Lorentzian parametrization with respect to any other moving
reference frames. It seems reasonable then that the true state changes of a
moving charged particle $\xi $\ are exactly realized only with respect to
its proper rest reference system. Then the only remaining question would be
about the physical justification of the corresponding relationship between
time parameters of moving and rest reference frames.
\end{remark}

The relationship between reference frames that we have used through is
expressed as 
\begin{equation}
d\tau =dt(1-|u|^{2})^{1/2},  \label{L4.1}
\end{equation}%
where $u:=dr/dt\in \mathbb{E}^{3}$ is the velocity vector with which the
rest reference frame $\mathcal{K}_{\tau }$ moves with respect to another
arbitrarily chosen reference frame $\mathcal{K}.$ Expression (\ref{L4.1})
implies, in particular, that 
\begin{equation}
dt^{2}-|dr|^{2}=d\tau ^{2},  \label{L4.2}
\end{equation}%
which is identical to the classical infinitesimal Lorentz invariant. This is
not a coincidence, since all our dynamical vacuum field equations were
derived in turn \cite{BoPrTa,BoPrTaPr-Lore} \ from the governing equations
of the vacuum potential field function $W:M^{4}\mathbb{\rightarrow R}$ in
the form 
\begin{equation}
\partial ^{2}W/\partial t^{2}-\nabla ^{2}W=\xi \rho ,\text{ }\partial
W/\partial t+\nabla (vW)=0,\text{ }\partial \rho /\partial t+\nabla (v\rho
)=0,  \label{L4.3}
\end{equation}%
which is \emph{a priori} \ Lorentz invariant. Here $\rho \in $ $\mathbb{R}$
is the charge density and $v:=dr/dt$ the associated local velocity of the
vacuum field potential evolution. Consequently, the dynamical infinitesimal
Lorentz invariant (\ref{L4.2}) reflects this intrinsic structure of
equations (\ref{L4.3}). If it is rewritten in the following nonstandard
Euclidean form:%
\begin{equation}
dt^{2}=d\tau ^{2}+|dr|^{2}  \label{L4.4}
\end{equation}%
it gives rise to a completely different relationship between the reference
frames $\mathcal{K}_{t}$ \ and $\mathcal{K}_{\tau }$, namely%
\begin{equation}
dt=d\tau (1+|\dot{r}|^{2})^{1/2},  \label{L4.5}
\end{equation}%
where $\dot{r}:=dr/d\tau $ is the related particle velocity with respect to
the rest reference system. Thus, we observe that all our Lagrangian analysis
in this Section is based on the corresponding functional expressions written
in these "Euclidean" space-time coordinates and with respect to which the
least action principle was applied. So we see that there are two
alternatives - the first is to apply the least action principle to the
corresponding Lagrangian functions expressed in the Minkowski space-time
variables with respect to an arbitrarily chosen reference frame $\mathcal{K}%
, $ and the second is to apply the least action principle to the
corresponding Lagrangian functions expressed in Euclidean space-time
variables with respect to the rest reference frame $\mathcal{K}_{\tau }.$

This leads us to a slightly amusing but thought-provoking observation: It
follows from our analysis that all of the results of classical special
relativity related with the electrodynamics of charged point particles can
be obtained (in a one-to-one correspondence) using of our new definitions of
the dynamical particle mass and the least action principle with respect to
the associated Euclidean space-time variables in the rest reference system.

An additional remark concerning the quantization procedure of the proposed
electrodynamics models is in order: If the dynamical vacuum field equations
are expressed in canonical Hamiltonian form, as we have done in this paper,
only straightforward technical details are required to quantize the
equations and obtain the corresponding Schr\"{o}dinger evolution equations
in suitable Hilbert spaces of quantum states. There is another striking
implication from our approach: the Einsteinian equivalence principle \cite%
{LaLi,Paul,Feyn-1,Feyn,Klym} is rendered superfluous for our vacuum field
theory of electromagnetism and gravity.

Using the canonical Hamiltonian formalism devised here for the alternative
charged point particle electrodynamics models, we found it rather easy to
treat the Dirac quantization. The results obtained compared favorably with
classical quantization, but it must be admitted that we still have not given
a compelling physical motivation for our new models. This is something that
we plan to revisit in future investigations. Another important aspect of our
vacuum field theory no-geometry (geometry-free) approach to combining the
electrodynamics with the gravity, is the manner in which it singles out the
decisive role of the rest \ reference frame $\mathcal{K}_{\tau }.$ More
precisely, all of our electrodynamics models allow both the Lagrangian and
Hamiltonian formulations with respect to the rest reference \ system
evolution parameter $\tau \in \mathbb{R}$, which are well suited the to
canonical quantization. The physical nature of this fact remains is as yet
not quite clear. In fact, as far as we know \cite%
{Paul,LaLi,Klym,Logu-2,Logu-3}, there is no physically reasonable
explanation of this decisive role of the rest \ reference system, except for
that given by R. Feynman who argued in \cite{Feyn-1} that the relativistic
expression for the classical Lorentz force (\ref{L1.11}) has physical sense
only with respect to \ the rest reference frame variables $(\tau ,r)\in 
\mathbb{R}\times \mathbb{E}^{3}.$ In future research we plan to analyze the
quantization scheme in more detail and begin work on formulating a vacuum
quantum field theory of infinitely many particle systems.

\section{\protect\bigskip \label{Sec_2a}The electromagnetic
Dirac-Fock-Podolsky problem and symplectic properties of the Maxwell and
Yang-Mills type dynamical systems}

\subsection{Introduction}

When investigating different dynamical systems on canonical symplectic
manifolds, invariant under action of certain symmetry groups, additional
mathematical structures often appear, the analysis of which shows their
importance for understanding many related problems under study. Amongst them
\ we here mention the Cartan type connection on an associated principal
fiber bundle, which enables one to study in more detail the properties of
the investigated dynamical system in the case of its reduction upon the
corresponding invariant submanifolds and quotient spaces, associated with
them.

Problems related to the investigation of properties of reduced dynamical
systems on symplectic manifolds were studied, e.g., in \cite%
{AbMa,PrMy,BlPrSa}, where the relationship between a symplectic structure on
the reduced space and the available connection on a principal fiber bundle
was formulated in explicit form. Other aspects of dynamical systems related
to properties of reduced symplectic structures were studied in \cite%
{Kumm,HoKu,HoKu-1}, where, in particular, the reduced symplectic structure
was explicitly described within the framework of the classical Dirac scheme,
and several applications to nonlinear (including celestial) dynamics were
given.

It is well-known \cite{BoSh,Thir, Dira,BoPrTa,BoPrTa-1,Paul} that the
Hamiltonian theory of electromagnetic Maxwell equations faces a very
important classical problem of introducing into the unique formalism the
well known Lorentz conditions, ensuring both the wave structure of
propagating quanta and the positivity of energy. Regretfully, in spite of
classical studies on this problem given by Dirac, Fock and Podolsky \cite%
{DiFoPo}, the problem remains open, and the Lorentz condition is imposed
within the modern electrodynamics as the external constraint not entering a
priori the initial Hamiltonian (or Lagrangian) theory. Moreover, when trying
to quantize the electromagnetic theory, as it was shown by Pauli, Dirac,
Bogolubov and Shirkov and others \cite{BoSh,Paul,Dira}, within the existing
approaches the quantum Lorentz condition could not be satisfied, except in
the average sense, since it becomes not compatible with the related quantum
dynamics. This problem stimulated us to study this problem from the so
called symplectic reduction theory \cite{MaWe,BlPrSa}, which allows the
systematic introduction into the Hamiltonian formalism the external charge
and current conditions, giving rise to a partial solution to the Lorentz
condition problem mentioned above. Some applications of the method to
Yang-Mills type equations interacting with a point charged particle, are
presented in detail. In particular, based on analysis of reduced geometric
structures on fibered manifolds, invariant under the action of a symmetry
group, we construct the symplectic structures associated with connection
forms on suitable principal fiber bundles. We present suitable mathematical
preliminaries of the related Poissonian structures on the corresponding
reduced symplectic manifolds, which are often used \cite{AbMa,MaWe,Kupe} in
various problems of dynamics in modern mathematical physics, and apply them
to study the non-standard Hamiltonian properties of the Maxwell and
Yang-Mills type dynamical systems. We formulate a symplectic analysis of the
important Lorentz type constraints, which describe the electrodynamic vacuum
properties.

We formulate a symplectic reduction theory of the classical Maxwell
electromagnetic field equations and prove \cite{BoPrTa} that the important
Lorentz condition, ensuring the existence of electromagnetic waves \cite%
{BoSh,Feyn-1,LaLi}, can be naturally included into the Hamiltonian picture,
thereby solving the Dirac, Fock and Podolsky problem \cite{DiFoPo} mentioned
above. We also study from the symplectic reduction theory the Poissonian
structures and the classical minimal interaction principle related with
Yang-Mills type equations.

\subsection{\label{Subsec_2a.1}The symplectic reduction on cotangent fiber
bundles with symmetry}

Consider an $m$-dimensional smooth manifold $M$ and the cotangent vector
fiber bundle $T^{\ast }(M).$ We equip (see \cite{Godb}, Chapter VII; \cite%
{DuNoFo}) the cotangent space $T^{\ast }(M)$ with the canonical Liouville
1-form $\lambda (\alpha ^{(1)}):=$ $pr_{M}^{\ast }\alpha ^{(1)}\in \Lambda
^{1}(T^{\ast }(M)),$ where $pr_{M}:T^{\ast }(M)\rightarrow M$ is the
canonical projection, $\ pr_{M}^{\ast }:=(d\circ pr_{M})^{\ast }:T^{\ast
}(M)\rightarrow T^{\ast }(T^{\ast }(M))\ \ $is the adjoint to the standard
tangent mapping $d\circ pr_{M}:=pr_{M,\ast }:T(T^{\ast }(M))\rightarrow T(M)$
\ with respect to the natural convolution on the product $\ T^{\ast
}(M))\otimes T(M).$ Then for a general one-form 
\begin{equation}
\alpha ^{(1)}(u)=\sum_{j=1}^{m}v_{j}du^{j},  \label{1}
\end{equation}%
where $(u,v)\in T^{\ast }(M)$ are the corresponding canonical local
coordinates on $T^{\ast }(M),$ the canonical symplectic structure on $%
T^{\ast }(M)$ will be equal to $\omega ^{(2)}(\alpha ^{(1)}):=d$ $\lambda
(\alpha ^{(1)})=\sum_{j=1}^{m}dv_{j}\wedge du^{j}\in \Lambda ^{2}(T^{\ast
}(M)).$ \ The any group of diffeomorphisms of the manifold $M,$ naturally
lifted to the fiber bundle $T^{\ast }(M),$ preserves the invariance of the
canonical 1-form $\lambda (\alpha ^{(1)})\in \Lambda ^{1}(T^{\ast }(M)).$ In
particular, if a smooth action of a Lie group $G$ is given on the manifold $%
M,$ then every element $a\in \mathcal{G},$ where $\mathcal{G}$ is the Lie
algebra of the Lie group $G,$ generates the vector field $k_{a}:M\rightarrow
T(M)$ in a natural manner. Furthermore, since the group action on $M,$ i.e., 
\begin{equation}
\varphi :G\times M\rightarrow M,  \label{2}
\end{equation}%
generates a diffeomorphism $\ \varphi _{g}\in Diff$ $M$ for every element $%
g\in G,$ this diffeomorphism is naturally lifted to the corresponding
diffeomorphism $\varphi _{g}^{\ast }$ $\in Diff$ $T^{\ast }(M)$ of the
cotangent fiber bundle $T^{\ast }(M),$ which also leaves \ the canonical
1-form $pr_{M}^{\ast }\alpha ^{(1)}\in \Lambda ^{1}(T^{\ast }(M))$
invariant. Namely, the equality

\begin{equation}
\varphi _{g}^{\ast }\lambda (\alpha ^{(1)})=\lambda (\alpha ^{(1)})
\label{3}
\end{equation}%
holds \cite{AbMa,Godb,PrMy} for every 1-form $\alpha ^{(1)}\in \Lambda
^{1}(M).$ Thus, we can define on $T^{\ast }(M)$ the corresponding vector
field $K_{a}:T^{\ast }(M)\rightarrow T(T^{\ast }(M))$ for every element $%
a\in \mathcal{G}.$ Then condition (\ref{3}) can be rewritten in the
following form for all $a\in \mathcal{G}:$

\begin{equation*}
L_{K_{a}}\cdot pr_{M}^{\ast }\alpha ^{(1)}=pr_{M}^{\ast }\cdot
L_{k_{a}}\alpha ^{(1)}=0,
\end{equation*}%
where $L_{K_{a}}$ and $L_{k_{a}}$ are the usual Lie derivatives on $\Lambda
^{1}(T^{\ast }(M))$ and $\Lambda ^{1}(M),$ respectively.

The canonical symplectic structure on $T^{\ast }(M)$ defined above as

\begin{equation}
\omega ^{(2)}(\alpha ^{(1)}):=d\lambda (\alpha ^{(1)})  \label{4}
\end{equation}
is also invariant, i.e., $L_{K_{a}}\omega ^{(2)}=0$ for all $a\in \mathcal{G}%
.$

For any smooth function $H\in D(T^{\ast }(M)),$ \ a Hamiltonian vector field 
$K_{H}:T^{\ast }(M)\rightarrow T(T^{\ast }(M))$ \ such that

\begin{equation}
i_{K_{H}}\omega ^{(2)}=-dH  \label{5}
\end{equation}%
is defined, and vice versa, because the symplectic 2-form (\ref{4}) is
non-degenerate. Using (\ref{5}) and (\ref{4}), we easily establish that the
Hamiltonian function $H:=H_{K}\in D(T^{\ast }(M))$ is given by the
expression $H_{K}=pr_{M}^{\ast }\alpha ^{(1)}(K_{H})=\alpha
^{(1)}(pr_{M}^{\ast }K_{H})=\alpha ^{(1)}(k_{H}),$ where $k_{H}\in T(M)$ is
the corresponding vector field on the manifold $M,$ whose lifting to the
fiber bundle $T^{\ast }(M)$ coincides with the vector field $K_{H}:T^{\ast
}(M)\rightarrow T(T^{\ast }(M)).$ For $K_{a}:T^{\ast }(M)\rightarrow
T(T^{\ast }(M)),$ where $a\in \mathcal{G},$ \ it is easy to establish that
the corresponding Hamiltonian function $H_{a}=\alpha
^{(1)}(k_{a})=pr_{M}^{\ast }$ $\alpha ^{(1)}(K_{a})$ for $a\in \mathcal{G}$
defines \cite{AbMa,PrMy,HePrPr} a linear momentum \ mapping $l:T^{\ast
}(M)\rightarrow \mathcal{G}^{\ast }$ according to the rule

\begin{equation}
H_{a}:=<l,a>,  \label{6}
\end{equation}%
where $<$\textperiodcentered $,$\textperiodcentered $\ >$ \ is the
corresponding convolution on $\mathcal{G}^{\mathcal{\ast }}$ $\times 
\mathcal{G}.$ By virtue of definition (\ref{6}), the momentum mapping $%
l:T^{\ast }(M)\rightarrow \mathcal{G}^{\ast }$ is invariant under the action
of any invariant Hamiltonian vector field $K_{b}:T^{\ast }(M)\rightarrow
T(T^{\ast }(M))$ for any $b\in \mathcal{G}.$ Indeed, $%
L_{K_{b}}<l,a>=L_{K_{b}}H_{a}=-L_{K_{a}}H_{b}=0,$ because, by definition,
the Hamiltonian function $H_{b}\in D(T^{\ast }(M))$ is invariant under the
action of any vector field $K_{a}:T^{\ast }(M)\rightarrow T(T^{\ast }(M)),$ $%
a\in \mathcal{G}.$

We now fix a regular value of the momentum mapping $l(u,v)=\xi \in \mathcal{G%
}^{\ast }$ and consider the corresponding submanifold $\mathcal{M}_{\xi
}:=\{(u,v)\in T^{\ast }(M):l(u,v)=\xi \in \mathcal{G}^{\ast }\}.$ On the
basis of definition (\ref{1}) and the invariance of the 1-form $pr_{M}^{\ast
}$ $\alpha ^{(1)}\in \Lambda ^{1}(T^{\ast }(M))$ under the action of the Lie
group $G$ on $T^{\ast }(M),$ we can write the equalities

\begin{eqnarray}
&<&l(g\circ (u,v)),a>=pr_{M}^{\ast }\alpha ^{(1)}(K_{a})(g\circ (u,v))= 
\notag \\
&=&pr_{M}^{\ast }\alpha ^{(1)}(K_{Ad_{g-1}a})(u,v):=  \label{7} \\
&=&<l(u,v),Ad_{g-1}a>=<Ad_{g-1}^{\ast }l(u,v),a>  \notag
\end{eqnarray}%
for any $g\in G$ and all $a\in \mathcal{G}$ and $(u,v)\in T^{\ast }(M).$
Using (\ref{7}) we establish that, for every $g\in G$ and all $(u,v)\in
T^{\ast }(M),$ the following relation is true: $l(g\circ
(u,v))=Ad_{g^{-1}}^{\ast }l(u,v).$ This means that the diagram%
\begin{equation*}
\begin{array}{ccc}
T^{\ast }(M) & \overset{l}{\rightarrow } & \mathcal{G}^{\mathcal{\ast }} \\ 
g\downarrow &  & \downarrow Ad_{g-1}^{\ast } \\ 
T^{\ast }(M) & \overset{l}{\rightarrow } & \mathcal{G}^{\mathcal{\ast }}%
\end{array}%
\end{equation*}%
is commutative for all elements $g\in G;$ the corresponding action $%
g:T^{\ast }(M)\rightarrow T^{\ast }(M)$ is called equivariant \cite%
{AbMa,PrMy}.

Let $\ G_{\xi }\subset G$ \ denote the stabilizer of a regular element $\xi
\in \mathcal{G}^{\ast }$ with respect to the related co-adjoint action. \ It
is obvious that in this case the action of the Lie subgroup $G_{\xi }$ on
the submanifold $\mathcal{M}_{\xi }\subset \mathcal{M}:=T^{\ast }(M)$ is
naturally defined; we assume that it is free and proper. According to this
action on $\mathcal{M}_{\xi },$ we can define \cite{AbMa,HoKu,Kupe} a
so-called reduced space $\mathcal{\bar{M}}_{\xi }$ by taking the factor with
respect to the action of the subgroup $G_{\xi }$ on $\mathcal{M}_{\xi },$
i.e.,

\begin{equation}
\mathcal{\bar{M}}_{\xi }:=\mathcal{M}_{\xi }/G_{\xi }.  \label{8}
\end{equation}%
The quotient space (\ref{8}) induces a symplectic structure $\bar{\omega}%
_{\xi }^{(2)}\in \Lambda ^{2}(\mathcal{\bar{M}}_{\xi })$ on itself, which is
defined as follows:%
\begin{equation}
\bar{\omega}_{\xi }^{(2)}(\bar{\eta}_{1},\bar{\eta}_{2})=\omega _{\xi
}^{(2)}(\eta _{1},\eta _{2}),  \label{9}
\end{equation}%
where $\bar{\eta}_{1},\bar{\eta}_{2}\in T(\mathcal{\bar{M}}_{\xi })$ are
arbitrary vectors onto which vectors $\eta _{1},\eta _{2}\in T(\mathcal{M}%
_{\xi })$ are projected, taken at any point $(u_{\xi },v_{\xi })\in \mathcal{%
M}_{\xi },$ being uniquely projected onto the point $\bar{\mu}_{\xi }\in 
\mathcal{\bar{M}}_{\xi },$ according to (\ref{8}).

Let $\pi _{\xi }:\mathcal{M}_{\xi }\rightarrow \mathcal{M}\ $ denote the
corresponding imbedding mapping into $\mathcal{M}$ and let $r_{\xi }:%
\mathcal{M}_{\xi }\mathcal{\rightarrow \bar{M}}_{\xi }$ denote the
corresponding reduction to the space $\mathcal{\bar{M}}_{\xi }.$ Then
relation (\ref{9}) can be rewritten equivalently in the form of the equality 
\begin{equation}
r_{\xi }^{\ast }\bar{\omega}_{\xi }^{(2)}=\pi _{\xi }^{\ast }\omega ^{(2)},
\label{10}
\end{equation}%
defined on vectors on the cotangent space $T^{\ast }(\mathcal{M}_{\xi }).$
To establish the symplecticity of the 2-form $\omega _{\xi }^{(2)}\in
\Lambda ^{2}(\mathcal{\bar{M}}_{\xi }),$ we use the corresponding
non-degeneracy of the Poisson bracket $\{$\textperiodcentered $,$%
\textperiodcentered $\}_{\xi }^{r}$ on $\mathcal{\bar{M}}_{\xi }.$ To
calculate it, we use a Dirac type construction, defining functions on $\ 
\mathcal{\bar{M}}_{\xi }$ as certain $G_{\xi }$-invariant functions on the
submanifold $\mathcal{M}_{\xi }.$ Then one can calculate the Poisson bracket 
$\{$\textperiodcentered $,$\textperiodcentered $\}_{\xi }$ of such functions
that corresponds to symplectic structure (\ref{4}) as an ordinary Poisson
bracket on $\mathcal{M},$ arbitrarily extending these functions from the
submanifold $\mathcal{M}_{\xi }$ $\subset \mathcal{M}$ to a certain
neighborhood $U(\mathcal{M}_{\xi })\subset \mathcal{M}.$ It is obvious that
two extensions of a given function to the neighborhood $U(\mathcal{M}_{\xi
}) $ of this type differ by a function that vanishes on the submanifold $%
\mathcal{M}_{\xi }\subset \mathcal{M}.$ The difference between the
corresponding Hamiltonian fields of these two different extensions to $U(%
\mathcal{M}_{\xi })$ is completely controlled by the conditions of the
following lemma (see also \cite{AbMa,PrMy,HoKu,Pryk-Ampe}).

\begin{lemma}
\label{Lm_2.1}Suppose that a function $f:U(\mathcal{M}_{\xi })\rightarrow 
\mathbb{R}$ is smooth and vanishes on $\mathcal{M}_{\xi }$ $\subset T^{\ast
}(M),$ i.e., $f|_{\mathcal{M}_{\xi }}=0.$ Then, at every point $(u_{\xi
},v_{\xi })\in \mathcal{M}_{\xi }$ the corresponding Hamiltonian vector
field $K_{f}\in T(U(\mathcal{M}_{\xi }))$ is tangent to the orbit $%
Or(G;(u_{\xi },v_{\xi })).$
\end{lemma}

\begin{proof}
It is obvious that the submanifold $M_{\xi }\subset T^{\ast }(M)$ is defined
by a certain collection of relations of the type%
\begin{equation}
H_{a_{s}}=\xi _{s},\text{ \ \ \ \ \ \ \ \ }\xi _{s}:=<\xi ,a_{s}>,
\label{11}
\end{equation}%
where $a_{s}\in \mathcal{G},s=\overline{1,dimG},$ is a certain basis of the
Lie algebra $\mathcal{G},$ which follows from definition (\ref{6}). Since a
function $f:U(\mathcal{M}_{\xi })\rightarrow \mathbb{R}$ vanishes on $%
\mathcal{M}_{\xi },$ we can write the following equality:%
\begin{equation*}
f=\sum_{s=1}^{\dim \mathcal{G}}(H_{a_{s}}-\xi _{s})f_{s},
\end{equation*}%
where $\ f_{s}:U(\mathcal{M}_{\xi })\rightarrow \mathbb{R},$ $\ \ s=%
\overline{1,dimG},$ is a certain collection of functions in the neighborhood 
$U(\mathcal{M}_{\xi })$. We take an arbitrary tangent vector $\eta \in T(U(%
\mathcal{M}_{\xi }))$ at the point $(u_{\xi },v_{\xi })\in \mathcal{M}_{\xi
} $ \ and calculate the expression%
\begin{eqnarray}
&<&df(u_{\xi },v_{\xi }),\eta (u_{\xi },v_{\xi })>=\sum_{s=1}^{\dim \mathcal{%
G}}<dH_{a_{s}}(u_{\xi },v_{\xi }),\eta (u_{\xi },v_{\xi })>f_{s}(u_{\xi
},v_{\xi })=  \notag \\
&=&-\sum_{s=1}^{\dim \mathcal{G}}\omega ^{(2)}(K_{a_{i}}(u_{\xi },v_{\xi
}),\eta (u_{\xi },v_{\xi }))f_{s}(u_{\xi },v_{\xi })=  \notag \\
&=&-\omega ^{(2)}(\sum_{s=1}^{\dim \mathcal{G}}K_{a_{s}}(u_{\xi },v_{\xi
})f_{s}(u_{\xi },v_{\xi }),\eta (u_{\xi },v_{\xi }))=  \notag \\
&=&-<i_{\left( \sum_{s=1}^{\dim \mathcal{G}}K_{a_{s}}(u_{\xi },v_{\xi
})f_{s}(u_{\xi },v_{\xi })\right) }\omega ^{(2)},\eta (u_{\xi },v_{\xi })>.
\label{12}
\end{eqnarray}%
It follows from the arbitrariness of the vector $\eta \in T(\mathcal{M}_{\xi
})$ at the point $(u_{\xi },v_{\xi })\in \mathcal{M}_{\xi }$ and relation (%
\ref{12}) that%
\begin{equation*}
K_{f}=\sum_{s=1}^{\dim \mathcal{G}}K_{a_{s}}f_{s},
\end{equation*}%
i.e., $K_{f}:\mathcal{M}_{\xi }\rightarrow T(Or(G)),$ which was to be proved.
\end{proof}

As a corollary of Lemma \ref{Lm_2.1}, we obtain an algorithm for the
determination of the reduced Poisson bracket $\{$\textperiodcentered $,$%
\textperiodcentered $\}_{\xi }^{r}$ on the space $\mathcal{\bar{M}}_{%
\mathcal{\xi }}$ according to definition (\ref{10}). Namely, we choose two
functions defined on $\mathcal{M}_{\xi }$ and invariant under the action of
the subgroup $G_{\xi }$ and arbitrarily smoothly extend them to a certain
open domain $U(\mathcal{M}_{\xi })\subset \mathcal{M}.$ Then we determine
the corresponding Hamiltonian vector fields on $\mathcal{M}$ and project
them onto the space tangent to $\mathcal{M}_{\xi },$ adding, if necessary,
the corresponding vectors tangent to the orbit $Or(G).$ It is obvious that
the projections obtained depend on the chosen extensions to the domain $U(%
\mathcal{M}_{\xi })\subset \mathcal{M}.$ As a result, we establish that the
reduced Poisson bracket $\{$\textperiodcentered $,$\textperiodcentered $%
\}_{\xi }^{r}$ is uniquely defined via the restriction of the initial
Poisson bracket upon $\mathcal{M}_{\xi }$ $\subset \mathcal{M}.$ By virtue
of the non-degeneracy of the latter and the functional independence of the
basis functions (\ref{11}) on the submanifold $U(\mathcal{M}_{\xi })\subset 
\mathcal{M},$ the reduced Poisson bracket $\{$\textperiodcentered $,$%
\textperiodcentered $\}_{\xi }^{r}$ appears to be \cite{AbMa,PrMy}
non-degenerate on $\mathcal{\bar{M}}_{\xi }.$ \ As a consequence of the
non-degeneracy, we establish that the dimension of the reduced space $%
\mathcal{\bar{M}}_{\xi }$ is even. Taking into account that the element $\xi
\in \mathcal{G}^{\mathcal{\ast }}$ is regular and the dimension of the Lie
algebra of the stabilizer $\mathcal{G}_{\xi }$ is equal to $dim$ $G_{\xi },$
we easily establish that $dim$ $\mathcal{\bar{M}}_{\xi }=$ $dim$ $\mathcal{M}%
-2dim$ $\mathcal{G}_{\xi }.$ Since, by construction, $dim$ $\mathcal{M}=2m,$
we conclude that the dimension of the reduced space $\mathcal{\bar{M}}_{\xi
} $ is necessarily even.

For the correctness of the algorithm, it is necessary to establish the
existence of the corresponding projections of Hamiltonian vector fields onto
the tangent space $T(\mathcal{M}_{\xi }).$ The following statement is true.

\begin{theorem}
\label{Tm_2.1}At every point $(u_{\xi },v_{\xi })\in \mathcal{M}_{\xi
}\subset \mathcal{M},$ one can choose a vector $V_{f}\in T(Or(G))$ such that 
$K_{f}(u_{\xi },v_{\xi })$ $+V_{f}(u_{\xi },v_{\xi })\in T_{(u_{\xi },v_{\xi
})}(\mathcal{M}_{\xi }).$ Furthermore, the vector $V_{f}\in T(Or(G))$ is
determined uniquely up to a vector tangent to the orbit $Or(G_{\xi }).$
\end{theorem}

\begin{proof}
Note that the orbit $Or(G;(u_{\xi },v_{\xi }))$ passing through the point $%
(u_{\xi },v_{\xi })\in \mathcal{M}_{\xi }$ is always symplectically
orthogonal to the tangent space $T_{(u_{\xi },v_{\xi })}(\mathcal{M}_{\xi
}). $ Indeed, for any vector $\eta \in T(\mathcal{M}_{\xi })$ and $a\in 
\mathcal{G},$ we have $\omega ^{(2)}(\eta ,K_{a})=-i_{K_{a}}\omega
^{(2)}(\eta )=dH_{a}(\eta )=0,$ because the submanifold $\mathcal{M}_{\xi }$ 
$\subset \mathcal{M}$ is defined by the equality $<\xi ,a>=H_{a}$ for all $%
a\in \mathcal{G},$ i.e., $dH_{a}=0$ on $\mathcal{M}_{\xi }.$ Thus, $T(%
\mathcal{M}_{\xi })\cap T(Or(G))=T(Or(G))$ because $H_{a}\circ g_{\xi
}=H_{a} $ for all $g_{\xi }\in G_{\xi },$ which follows from the invariance
of the element $\xi \in \mathcal{G}^{\ast }$ under the action of the Lie
group $G_{\xi }.$\bigskip\ We now solve the imbedding condition $\
K_{f}+V_{f}\in T(\mathcal{M}_{\xi }),$ or the equation 
\begin{equation}
\omega ^{(2)}(K_{f}+V_{f},K_{a})=0  \label{13}
\end{equation}%
on the manifold $\mathcal{M}_{\xi }$ $\subset T^{\ast }(M)$ for all $a\in 
\mathcal{G}.$ We rewrite equality (\ref{13}) in the form 
\begin{equation}
K_{a}f=\omega ^{(2)}(V_{f},K_{a})  \label{14}
\end{equation}%
on $\mathcal{M}_{\xi }$ for all $a\in \mathcal{G};$ it is obvious that the
2-form on the right-hand side of (\ref{14}) depends only on the element $\xi
\in G^{\ast }.$ Taking into account the equivariance of the group action on $%
\mathcal{M}$ and the obvious equality%
\begin{equation*}
\omega ^{(2)}(K_{a},K_{b})=pr_{M}^{\ast }\alpha
^{(1)}([K_{a},K_{b}])=-pr_{M}^{\ast }\alpha ^{(1)}(K_{[a,b]})
\end{equation*}%
for all $a,b\in \mathcal{G},$ we establish that there exists an element $%
a_{f}\in \mathcal{G}$ such that $V_{f}=K_{a_{f}}\in T(Or(G))$ and%
\begin{eqnarray}
\omega ^{(2)}(V_{f},K_{a}) &=&\omega ^{(2)}(K_{a_{f}},K_{a})=pr_{M}^{\ast
}\alpha ^{(1)}([K_{a},K_{a_{f}}])=  \notag \\
&=&pr_{M}^{\ast }\alpha ^{(1)}(K_{[a_{f},a]})=H_{[a_{f},a]}=<l,[a_{f},a]>= 
\notag \\
&=&<\xi ,[a_{f},a]>=<ad_{a_{f}}^{\ast }\xi ,a>  \label{15}
\end{eqnarray}%
on $\mathcal{M}_{\xi }$ for all $a\in \mathcal{G}.$ Since $ad_{a_{f}}^{\ast
}\xi =0$ for any $a_{f}\in \mathcal{G}_{\xi },$ we conclude that, on the
quotient space $\mathcal{G}/\mathcal{G}_{\xi }$ the right-hand side of (\ref%
{15}) defines a non-degenerate skew-symmetric form associated with the
canonical isomorphism $\hat{\xi}:\mathcal{G}/\mathcal{G}_{\xi }\rightarrow (%
\mathcal{G}/\mathcal{G}_{\xi })^{\ast },$ where, by definition,%
\begin{equation}
<\hat{\xi}(\tilde{a}),\tilde{b}>:=<\xi ,[a,b]>  \label{16}
\end{equation}%
for any $\tilde{a}$ and $\tilde{b}$ $\in $ $\mathcal{G}/\mathcal{G}_{\xi }$
with the corresponding representatives $a$ and $b{\in }\mathcal{G}.$
Further, since the function $f:\mathcal{M}_{\xi }\rightarrow \mathbb{R}$ is $%
G_{\xi }$-invariant on $\mathcal{M}_{\xi }$ $\subset \mathcal{M},$ the
right-hand side of (\ref{14}) defines an element $\mu _{f}\in (\mathcal{G}/%
\mathcal{G}_{\xi })^{\ast }$ by the equality 
\begin{equation*}
\mu _{f}:\tilde{a}:=-K_{a}f
\end{equation*}%
for all $a\in G.$ Using relations (\ref{15}) ) and (\ref{16}), we establish
that there exists the element 
\begin{equation*}
\tilde{a}_{f}=\hat{\xi}^{-1}\circ \mu _{f}\in \mathcal{G}/\mathcal{G}_{\xi }.
\end{equation*}%
Since the element $\tilde{a}_{f}\in \mathcal{G}/\mathcal{G}_{\xi }$ is
associated with the element $a_{f}$ $(mod$ $\mathcal{G}_{\xi })$ $\in 
\mathcal{G},$ which uniquely generates a locally defined vector field $%
K_{a_{f}}:Or(G)\rightarrow T(Or(G)),$ using the fact that $V_{f}=K_{a_{f}}$
on $\mathcal{M}_{\xi }\subset \mathcal{M},$ we complete the proof of the
theorem.
\end{proof}

Now assume that two functions $f_{1},f_{2}\in D(\mathcal{M}_{\xi })$ are $%
G_{\xi }$-invariant. Then their reduced Poisson bracket $\{f_{1},f_{2}\}_{%
\xi }^{r}$ on $\mathcal{\bar{M}}_{\xi }$ is defined according to the rule: 
\begin{equation}
\{f_{1},f_{2}\}_{\xi }^{r}:=-\omega
^{(2)}(K_{f_{1}}+V_{f_{1}},K_{f_{2}}+V_{f_{2}})=\{f_{1},f_{2}\}+\omega
^{(2)}(V_{f_{1}},V_{f_{2}}),  \label{17}
\end{equation}%
where we have used the following identities on $\mathcal{M}_{\xi }$ $\subset
T^{\ast }(M)$: 
\begin{equation*}
\omega ^{(2)}(K_{f_{1}}+V_{f_{1}},V_{f_{2}})=0=\omega
^{(2)}(K_{f_{2}}+V_{f_{2}},V_{f_{1}}),
\end{equation*}%
being simple consequences of equality (\ref{13}) on $\mathcal{M}_{\xi }.$
Regarding (\ref{15}), relation (\ref{17}) takes the form 
\begin{equation}
\{f_{1},f_{2}\}_{\xi }^{r}=\{f_{1},f_{2}\}+\frac{1}{2}%
(V_{f_{1}}f_{2}-V_{f_{2}}f_{1}),  \label{18}
\end{equation}%
where $f_{1},f_{2}\in D(\mathcal{M}_{\xi })$ are arbitrary smooth extensions
of the $G_{\xi }$-invariant functions defined earlier on the domain $U(%
\mathcal{M}_{\xi }).$ Thus, the following theorem holds.

\begin{theorem}
\label{Tm_2.2}The reduced Poisson bracket of two functions on the quotient
space $\mathcal{\bar{M}}_{\xi }$ $=$ $\mathcal{M}_{\xi }/G_{\xi }$ is
determined with the use of their arbitrary smooth extensions to functions on
an open neighborhood $U(\mathcal{M}_{\xi })$ according to the Dirac-type
formula (\ref{18}).
\end{theorem}

\subsection{The symplectic reduction on principal fiber bundles with
connection}

We begin by reviewing the backgrounds of the reduction theory subject to
Hamiltonian systems \ with symmetry on principle fiber bundles. The material
is partly available in \cite{GiSt,Kupe}, so here it will be only sketched in
notations suitable for us.

Let $G$ denote a given Lie group with the unity element $e\in G$ and the
corresponding Lie algebra $\mathcal{G}$ $\simeq T_{e}(G).$ Consider a
principal fiber bundle $p:(M,\varphi )\rightarrow N$ \ with the structure
group $G$ and base manifold $N,$ on which the Lie group $G$ acts by means of
a mapping $\ \varphi :M\times G\rightarrow M.$ Namely, for each $g\in G$
there is a group diffeomorphism $\varphi _{g}:M\rightarrow M,$ generating
for any fixed $u\in M$ the following induced mapping: $\hat{u}:G\rightarrow
M,$ where 
\begin{equation}
\hat{u}(g)=\varphi _{g}(u).  \label{0.1}
\end{equation}

On the principal fiber bundle $p:(M,\varphi )\rightarrow N$ \ a connection $%
\Gamma (${$\mathcal{A}$}$)$ is assigned by means of such a morphism {$%
\mathcal{A}$}$:(T(M),\varphi _{g\ast })\rightarrow (\mathcal{G},Ad_{g^{-1}})$
that for each $u\in M$ a mapping $\mathcal{A}(u):T_{u}(M)\rightarrow 
\mathcal{G}$ is a left inverse one to the \ tangent mapping $d\hat{u}(e):=%
\hat{u}_{\ast }(e):\mathcal{G}\rightarrow T_{u}(M)$ at unity element $e\in
G, $ that is

\begin{equation}
\mathit{\mathcal{A}}(u)\hat{u}_{\ast }(\xi )=1.  \label{0.2}
\end{equation}

As usual, denote by $\varphi _{g}^{\ast }:T^{\ast }(M)\rightarrow T^{\ast
}(M)$ the corresponding cotangent lift of the mapping \ $\varphi
_{g}:M\rightarrow M$ \ at any $g\in G.$ If $\alpha ^{(1)}\in \Lambda ^{1}(M)$
is the canonical $G$ - invariant 1-form on $\ M,$ the canonical symplectic
structure $\omega ^{(2)}\in \Lambda ^{2}(T^{\ast }(M))$ given by 
\begin{equation}
\omega ^{(2)}:=d\text{ }pr^{\ast }\alpha ^{(1)}  \label{0.3}
\end{equation}%
generates the corresponding momentum mapping $l:T^{\ast }(M)\rightarrow 
\mathcal{G}^{\ast },$ where 
\begin{equation}
l(\alpha ^{(1)})(u)=\hat{u}^{\ast }(e)\alpha ^{(1)}(u)  \label{0.4}
\end{equation}%
for all $u\in M.$ Remark here that the principal fiber \ bundle structure $%
p:(M,\varphi )\rightarrow N$ \ means in part the exactness of the following
sequences of mappings: 
\begin{equation}
0\rightarrow \mathcal{G}\overset{\hat{u}_{\ast }(e)}{\rightarrow }T_{u}(M)%
\overset{p_{\ast }(u)}{\rightarrow }T_{p(u)}(N)\rightarrow 0,  \label{0.5}
\end{equation}%
that is 
\begin{equation}
p_{\ast }(u)\hat{u}_{\ast }(e)=0=\hat{u}^{\ast }(e)p^{\ast }(u)  \label{0.6}
\end{equation}%
for all $u\in M.$ Combining (\ref{0.6}) with (\ref{0.2}) and (\ref{0.4}),
one obtains such an embedding: 
\begin{equation}
\lbrack 1-\mathcal{A}^{\ast }(u)\hat{u}^{\ast }(e)]\alpha ^{(1)}(u)\in \text{%
range }p^{\ast }(u)  \label{0.7}
\end{equation}%
for the canonical 1-form $\alpha ^{(1)}\in \Lambda ^{1}(M)$ at $u\in M.$ The
expression (\ref{0.7}) means of course, that 
\begin{equation}
\hat{u}^{\ast }(e)[1-\mathcal{A}^{\ast }(u)\hat{u}^{\ast }(e)]\alpha
^{(1)}(u)=0  \label{0.8}
\end{equation}%
for all $u\in M.$ Now taking into account that the mapping \ $p^{\ast
}(u):T^{\ast }(N)\rightarrow T^{\ast }(M)$ \ is for each $u\in M$ injective,
it has the unique inverse mapping \ $\ (p^{\ast }(u))^{-1}$ upon its image \ 
$p^{\ast }(u)T_{p(u)}^{\ast }(N)\subset T_{u}^{\ast }(M).$ Thereby \ for
each $u\in M$ one can define a morphism $p_{\mathcal{A}}:(T^{\ast
}(M),\varphi _{g}^{\ast })\rightarrow T^{\ast }(N)$ as 
\begin{equation}
p_{\mathcal{A}}(u):\alpha ^{(1)}(u)\rightarrow (p^{\ast }(u))^{-1}[1-%
\mathcal{A}^{\ast }(u)\hat{u}^{\ast }(e)]\alpha ^{(1)}(u).  \label{0.9}
\end{equation}%
Based on the definition (\ref{0.9}) one can easily check that the diagram 
\begin{equation}
\begin{array}{ccc}
T^{\ast }(M) & \overset{p_{\mathcal{A}}}{\rightarrow } & T^{\ast }(N) \\ 
\left. pr_{M}\right\downarrow &  & \left\downarrow pr_{N}\right. \\ 
M & \overset{p}{\rightarrow } & N%
\end{array}
\label{0.10}
\end{equation}%
is commutative.

Let an element $\xi \in \mathcal{G}^{\ast }$ be $G$-invariant, that is $\ \
Ad_{g^{-1}}^{\ast }\xi =\xi $ for all $\ g\in G.$ Denote also by \ $p_{%
\mathcal{A}}^{\xi }$ \ the restriction of the mapping (\ref{0.9}) upon the
subset $\mathcal{M}_{\xi }:=l^{-1}(\xi )\in T^{\ast }(M),$ that is \ \ $p_{%
\mathcal{A}}^{\xi }:\mathcal{M}_{\xi }\rightarrow T^{\ast }(N),$ where for
all $u\in M$%
\begin{equation}
p_{\mathcal{A}}^{\xi }(u):l^{-1}(\xi )\rightarrow (p^{\ast }(u))^{-1}[1-%
\mathcal{A}^{\ast }(u)\hat{u}^{\ast }(e)]l^{-1}(\xi ).  \label{0.11}
\end{equation}%
Now one can characterize the structure of the reduced phase space {$\mathcal{%
\bar{M}}_{\xi }:=$}$l^{-1}(\xi )/G$ \ by means of the following lemma.

\begin{lemma}
{\label{lem_01} {The mapping }$p_{\mathcal{A}}^{\xi }(u):\mathcal{M}_{\xi
}\rightarrow T^{\ast }(N),$ where \ }$\mathcal{M}_{\xi }:=${\ }$l^{-1}(\xi
), ${\ {\ is a principal fiber }$G$ {-bundle with the reduced space \ }$%
\mathcal{\bar{M}}_{\xi },$ {\ being diffeomorphic to }$T^{\ast }(N).$}
\end{lemma}

Denote by $<.,.>_{\mathcal{G}}$ the standard $Ad$-invariant non-degenerate
scalar product on $\mathcal{G}\times \mathcal{G}.$\ \ Based on Lemma \ref%
{lem_01} one derives the following \ characteristic theorem.

\begin{theorem}
{\ \label{th_02}\ {Given a principal fiber \ }$G${-bundle with a connection }%
$\Gamma (\mathcal{A})$ {\ and a }$G${-invariant element }$\ \xi \in \mathcal{%
G}^{\ast },$ {\ then every such connection }$\Gamma (\mathcal{A})$ {\
defines a symplectomorphism }$\nu _{\xi }:{\mathcal{\bar{M}}_{\xi }}%
\rightarrow T^{\ast }(N)$ {\ between the reduced phase space }$\mathcal{\bar{%
M}}_{\xi }$ {\ and \ cotangent bundle \ }$T^{\ast }(N),$ {\ where }$%
l:T^{\ast }(M)\rightarrow \mathcal{G}^{\ast }$ {\ is the naturally
associated momentum mapping for the group }$G${-action on }$M.$ {\ Moreover,
the following equality } 
\begin{equation}
(p_{\mathcal{A}}^{\xi })(d\text{ }pr_{N}^{\ast }\beta ^{(1)}+pr_{N}^{\ast }%
\text{ }\Omega _{\xi }^{(2)})=\left. d\text{ }pr_{M}^{\ast }\alpha
^{(1)}\right\vert _{l^{-1}(\xi )}  \label{0.12}
\end{equation}%
{holds for the canonical 1-forms \ }$\beta ^{(1)}\in \Lambda ^{1}(N)$ {\ and
\ }$\alpha ^{(1)}\in \Lambda ^{1}(M),$ {\ where \ }$\Omega _{\xi
}^{(2)}:=<\xi ,\Omega ^{(2)}>_{\mathcal{G}}$ {\ is the }$\xi ${-component of
the corresponding curvature form }$\Omega ^{(2)}\in \Lambda ^{(2)}(N)\otimes 
\mathcal{G}.$}
\end{theorem}

\begin{proof}
One has that on $l^{-1}(\xi )\subset M$ the following expression, due to (%
\ref{0.9}), holds: 
\begin{equation*}
p^{\ast }(u)p_{\mathcal{A}}^{\xi }(\alpha ^{(1)}(u))=p^{\ast }(u)\beta
^{(1)}(pr_{N}(u))=\alpha ^{(1)}(u)-\mathcal{A}^{\ast }(u)\hat{u}^{\ast
}(e)\alpha ^{(1)}(u)
\end{equation*}%
for \ any $\beta ^{(1)}\in T^{\ast }(N)$ and all $u\in M_{\xi
}:=p_{M}l^{-1}(\xi )\subset M.$ Thus we easily get that 
\begin{equation*}
\alpha ^{(1)}(u)=(p_{\mathcal{A}}^{\xi })^{-1}\beta ^{(1)}(p_{N}(u))=p^{\ast
}(u)\beta ^{(1)})(pr_{N}(u))+<\xi ,\mathcal{A}(u)>_{\mathcal{G}}
\end{equation*}%
for all $u\in M_{\xi }.$ Recall now that in virtute of (\ref{0.10}) on the
manifold $M_{\xi }$ there hold relationships:%
\begin{equation*}
p\text{ }\circ pr_{M_{\xi }}=pr_{N}\text{ }\circ p_{\mathcal{A}}^{\xi },%
\text{ \ }pr_{M_{\xi }}^{\ast }\text{\ }\circ p^{\ast }\text{ }=\text{ }(p_{%
\mathcal{A}}^{\xi })^{\ast }\text{\ }\circ pr_{N}^{\ast }\text{ .\ }
\end{equation*}%
Therefore we can now write down that 
\begin{eqnarray*}
pr_{M_{\xi }}^{\ast }\alpha ^{(1)}(u) &=&pr_{M_{\xi }}^{\ast }\beta
^{(1)}(p_{N}(u))+pr_{M_{\xi }}^{\ast }<\xi ,\mathcal{A}(u)>_{\mathcal{G}} \\
&=&(p_{\mathcal{A}}^{\xi })^{\ast }(pr_{N}^{\ast }\beta
^{(1)})(u)+pr_{M_{\xi }}^{\ast }<\xi ,\mathcal{A}(u)>_{\mathcal{G}},
\end{eqnarray*}%
whence taking the external \ differential, one arrives at the following
equalities: 
\begin{eqnarray*}
d\text{ }pr_{M_{\xi }}^{\ast }\alpha ^{(1)}(u) &=&(p_{\mathcal{A}}^{\xi
})^{\ast }d(pr_{N}^{\ast }\beta ^{(1)})(u)+pr_{M_{\xi }}^{\ast }<\xi ,d\text{
}\mathcal{A}(u)>_{\mathcal{G}}= \\
&=&(p_{\mathcal{A}}^{\xi })^{\ast }d(pr_{N}^{\ast }\beta
^{(1)})(u)+pr_{M_{\xi }}^{\ast }<\xi ,\Omega (p(u))>_{\mathcal{G}}= \\
&=&(p_{\mathcal{A}}^{\xi })^{\ast }d(pr_{N}^{\ast }\beta
^{(1)})(u)+pr_{M_{\xi }}^{\ast }p^{\ast }<\xi ,\Omega >_{\mathcal{G}}(u)= \\
&=&(p_{\mathcal{A}}^{\xi })^{\ast }d(pr_{N}^{\ast }\beta ^{(1)})(u)+(p_{%
\mathcal{A}}^{\xi })^{\ast }pr_{N}^{\ast }<\xi ,\Omega >_{\mathcal{G}}(u)= \\
&=&(p_{\mathcal{A}}^{\xi })^{\ast }[d(pr_{N}^{\ast }\beta
^{(1)})(u)+pr_{N}^{\ast }<\xi ,\Omega >_{\mathcal{G}}(u)].
\end{eqnarray*}%
When deriving the above expression we made use of the following property
satisfied by the curvature 2-form $\Omega \in \Lambda ^{2}(M)\otimes 
\mathcal{G}:$%
\begin{eqnarray*}
&<&\mathcal{\xi },dA(u)\mathcal{>_{\mathcal{G}}=<}\xi ,d\mathcal{A(}u)+%
\mathcal{A(}u)\wedge \mathcal{A(}u)>_{\mathcal{G}}-<\xi ,\mathcal{A(}%
u)\wedge \mathcal{A(}u)>_{\mathcal{G}} \\
&=&<\xi ,\Omega (p_{N}(u))>_{\mathcal{G}}=<\xi ,p_{N}^{\ast }\Omega >_{%
\mathcal{G}}(u)
\end{eqnarray*}%
at any $u\in M_{\xi },$ since for any $A,B\in \mathcal{G}$ there holds $<\xi
,[A,B]>_{\mathcal{G}}=<Ad^{\ast }A\xi ,B>_{\mathcal{G}}=0$ in virtue of the
invariance condition $Ad_{G}^{\ast }\xi =\xi .$ Thereby the proof is
finished.
\end{proof}

\begin{remark}
\label{rem_03} {As the canonical 2-form \ \ \ }$d$ {\ }$pr_{M}^{\ast }\alpha
^{(1)}\in $ {\ }$\Lambda ^{(2)}(T^{\ast }(M))$ {\ \ \ is }$G$ {-invariant on
\ \ \ }$T^{\ast }(M)$ {\ due to construction, it is evident that its
restriction upon the \ }$G$ {-invariant submanifold \ \ \ $\mathcal{M}_{\xi
} $}$\subset T^{\ast }(M)$ {\ \ will be effectively defined only on the
reduced space \ \ $\mathcal{\bar{M}}_{\xi }$}$,$ {\ that ensures the
validity of the equality sign in (\ref{0.12}). \ }
\end{remark}

As a consequence of Theorem \ref{th_02} one can formulate the following
useful for applications theorems.

\begin{theorem}
\label{th_04} {Let an element }$\xi \in \mathcal{G}^{\ast }$ {\ have the
isotropy group }$G_{\xi }$ {\ \ \ acting on the subset \ \ $\mathcal{M}_{\xi
}$}$\subset T^{\ast }(M)$ {\ freely and properly, so that the reduced phase
space \ }$({\mathcal{\bar{M}}_{\xi }}\ ,\bar{\sigma}_{\xi }^{(2)})$ where,
by definition, ${\mathcal{\bar{M}}_{\xi }:=}l^{-1}(\xi )/G_{\xi },$ {\ is
symplectic whose symplectic structure is defined as } 
\begin{equation}
\bar{\sigma}_{\xi }^{(2)}:=\left. d\text{ }pr_{M}^{\ast }\alpha
^{(1)}\right\vert _{{\mathcal{\bar{M}}_{\xi }})}.  \label{0.13}
\end{equation}%
{If a principal fiber bundle \ \ }$p:(M,\varphi )\rightarrow N$ {\ has a
structure group coinciding with }$G_{\xi },${\ then\ \ the reduced
symplectic space \ \ \ \ }$({\mathcal{\bar{M}}_{\xi }},\bar{\sigma}_{\xi
}^{(2)})$ {\ is symplectomorphic to the cotangent symplectic space }$%
(T^{\ast }(N),\bar{\omega}_{\xi }^{(2)}),${\ where } 
\begin{equation}
\bar{\omega}_{\xi }^{(2)}=d\text{ }pr_{N}^{\ast }\beta ^{(1)}+pr_{N}^{\ast
}\Omega _{\xi }^{(2)},  \label{0.14}
\end{equation}%
{and the corresponding symplectomorphism \ is given by a relation like (\ref%
{0.12}).}
\end{theorem}

\begin{theorem}
{\label{th_05} {In order that two symplectic spaces \ \ \ }$({\mathcal{\bar{M%
}}_{\xi }},\bar{\sigma}_{\xi }^{(2)})$ {\ \ \ \ and \ \ \ }$(T^{\ast
}(N),dpr_{N}^{\ast }\beta ^{(1)})$ were {symplectomorphic, it is necessary
and sufficient that the element \ }$\xi \in \ker ${\ }}${h,}${\ where for }${%
G}${{-invariant element \ \ }$\xi \in \mathcal{G}^{\ast }$ {\ the mapping }$%
h:\xi \rightarrow \lbrack \Omega _{\xi }^{(2)}]\in H^{2}(N;\mathbb{Z}),$ {\
with }$H^{2}(N;\mathbb{Z})$\ {\ being the cohomology class of 2-forms \ on
the manifold }$N.$}
\end{theorem}

\subsection{The Hamiltonian analysis of the Maxwell electromagnetic
dynamical systems}

We take the Maxwell electromagnetic equations to be%
\begin{eqnarray}
\partial E/\partial t &=&\nabla \times B-J,\text{ \ \ \ }\partial B/\partial
t=-\nabla \times E,  \label{1.1} \\
&<&\nabla ,E>=\rho ,\text{ \ \ \ \ \ \ \ }<\nabla ,B>=0,\text{\ }  \notag
\end{eqnarray}%
on the cotangent phase space $T^{\ast }(N)$ \ to $N\subset T(D;\mathbb{E}%
^{3}),$ \ being the smooth manifold of smooth vector fields on an open
domain $D\subset \mathbb{R}^{3},$ all expressed in the light speed units.
Here $(E,B)\in T^{\ast }(N)$ \ \ is a vector of electric and magnetic
fields, $\rho :D\rightarrow \mathbb{R}$ and $J:D\rightarrow \mathbb{E}^{3}$\
are, simultaneously, fixed charge and current densities in the domain $D,$
satisfying the equation of continuity 
\begin{equation}
\partial \rho /\partial t+<\nabla ,J>=0,  \label{1.1a}
\end{equation}%
holding for all $t\in \mathbb{R},$ where we denoted by the sign $"\nabla "$
the gradient operation with respect to a variable $x\in $ $D,$ by the sign $%
"\times "$ the usual vector product in $\mathbb{E}^{3}:=(\mathbb{R}%
^{3},<\cdot ,\cdot >),$ \ being the standard three-dimensional Euclidean
vector space $\ \mathbb{R}^{3}$ endowed with the usual scalar product $%
<\cdot ,\cdot >.$

Aiming to represent equations (\ref{1.1}) as those on reduced symplectic
space, we define an appropriate configuration space $M$ $\subset \mathcal{T}%
(D;\mathbb{E}^{3})$ with a vector potential field coordinate \ $A\in M.$ The
cotangent space $T^{\ast }(M)$\ \ may be identified with pairs $(A;Y)\in
T^{\ast }(M),$ where $Y\in \mathcal{T}^{\ast }(D;\mathbb{E}^{3})$ is a
suitable vector field density in $D.$ On the space $T^{\ast }(M)$ there
exists the weak canonical symplectic form $\omega ^{(2)}\in \Lambda
^{2}(T^{\ast }(M)),$ allowing, $\ $owing to the definition$\ $of\ the\
Liouville\ from$\ \ \ \ \ \ \ \ \ \ \ \ \ $%
\begin{equation}
\lambda (\alpha ^{(1)})(A;Y)=\int_{D}d^{3}x(<Y,dA>:=(Y,dA),  \label{1.2}
\end{equation}%
the canonical expression$\ \ \ $%
\begin{equation}
\omega ^{(2)}(\alpha ^{(1)}):=d\lambda (\alpha ^{(1)})=(dY,\wedge dA).
\label{1.2b}
\end{equation}%
Here we denoted by $"\wedge "$ the usual external differentiation, by $%
d^{3}x,$ $x\in D,$ \ the Lebesgue measure in the domain $D$ and by $%
pr:T^{\ast }(M)\rightarrow M$ \ the standard projection upon the base space $%
M.$ Define now a Hamiltonian function $\tilde{H}\in \mathcal{D}(T^{\ast
}(M)) $ as%
\begin{equation}
H(A,Y)=1/2[(Y,Y)+(\nabla \times A,\nabla \times A)+(<\nabla ,A>,<\nabla
,A>)],  \label{1.2c}
\end{equation}%
describing the well-known Maxwell equations in vacuum, if the densities $%
\rho =0$ and $J=0.$ Really, owing to (\ref{1.2b}) one easily obtains from (%
\ref{1.2c}) that 
\begin{eqnarray}
\partial A/\partial t &:&=\delta H/\delta Y=Y,  \label{1.2d} \\
\partial Y/\partial t &:&=-\delta H/\delta A=-\nabla \times B+\nabla <\nabla
,A>,  \notag
\end{eqnarray}%
being true wave equations in vacuum, where we put, by definition, 
\begin{equation}
B:=\nabla \times A,  \label{1.2e}
\end{equation}%
being the corresponding magnetic field. \ Now defining 
\begin{equation}
E:=-Y-\nabla W  \label{1.1f}
\end{equation}%
for some function $W:M\rightarrow \mathbb{R}$ as the corresponding electric
field, the system of equations (\ref{1.2d}) will become, owing to definition
(\ref{1.2e}), 
\begin{equation}
\partial B/\partial t=-\nabla \times E,\text{ \ }\partial E/\partial
t=\nabla \times B,  \label{1.3}
\end{equation}%
exactly coinciding with the Maxwell equations in vacuum, if the Lorentz
condition 
\begin{equation}
\partial W/\partial t+<\nabla ,A>=0  \label{1.3a}
\end{equation}%
is involved.

Since definition (\ref{1.1f}) was essentially imposed rather than arising
naturally from the Hamiltonian approach and our equations are valid only for
a vacuum, we shall try to improve upon these matters by employing the
reduction approach devised in Section 2. Namely, we start with the
Hamiltonian (\ref{1.2c}) and observe that it is invariant with respect to
the following abelian symmetry group $G:=\exp \mathcal{G},$ where $\mathcal{G%
}\simeq C^{(1)}(D;\mathbb{R}),$ acting on the base manifold $M$ naturally
lifted to $T^{\ast }(M):$ for any $\psi \in \mathcal{G}$ and $(A,Y)\in
T^{\ast }(M)$%
\begin{equation}
\varphi _{\psi }(A):=A+\nabla \psi ,\ \ \ \ \varphi _{\psi }(Y)=Y.
\label{1.4}
\end{equation}%
The 1-form (\ref{1.2}) under transformation (\ref{1.4}) also is invariant
since 
\begin{equation}
\begin{array}{c}
\varphi _{\psi }^{\ast }\lambda (\alpha ^{(1)})(A,Y)=(Y,dA+\nabla d\psi )=
\\ 
=(Y,dA)-(<\nabla ,Y>,d\psi )=\lambda (\alpha ^{(1)})(A,Y),%
\end{array}
\label{1.5}
\end{equation}%
where we made use of the condition $d\psi \simeq 0$ in $\Lambda ^{1}(T^{\ast
}(M))$ for any $\psi \in \mathcal{G}.$ Thus, the corresponding momentum
mapping (\ref{0.4}) is given as 
\begin{equation}
l(A,Y)=-<\nabla ,Y>  \label{1.6}
\end{equation}%
for all $(A,Y)\in T^{\ast }(M).$ If $\rho \in \mathcal{G}^{\ast }$is fixed,
one can define the reduced phase space $\mathcal{\bar{M}}_{\rho
}:=l^{-1}(\rho )/G,$ since evidently, the isotropy group $G_{\rho }=G,$
owing to its commutativity and the condition (\ref{1.4}). Consider \ \ now a
principal fiber bundle $p:M\rightarrow N$ with the abelian structure group $%
G $ and a base manifold $N$ taken as 
\begin{equation}
N:=\{B\in \mathcal{T}(D;\mathbb{E}^{3}):\text{ \ }<\nabla ,\text{ }B>=0,%
\text{ \ }<\nabla ,E(S)>=\rho \},  \label{1.7}
\end{equation}%
where, by definition, 
\begin{equation}
p(A)=B=\nabla \times A.  \label{1.8}
\end{equation}%
We can construct a connection 1-form $\ \ \mathcal{A}\in \Lambda ^{1}(M)%
\mathbb{\otimes }\mathcal{G}$ on this bundle, where for all $A\in M$ 
\begin{equation}
\mathcal{A}(A)\cdot \hat{A}_{\ast }(l)=1,\text{ \ \ }d<\rho ,\mathcal{A}%
(A)>_{\mathcal{G}}=\Omega _{\rho }^{(2)}(A)\in H^{2}(M;\mathbb{Z}),
\label{1.9}
\end{equation}%
where $\mathcal{A}(A)\in \Lambda ^{1}(M)$ is some differential 1-form, which
we choose in the following form:%
\begin{equation}
\mathcal{A}(A):=-(W,d<\nabla ,A>),  \label{1.9a}
\end{equation}%
where $W\in C^{(1)}(D;\mathbb{R})$ is some scalar function, still not
defined. As a result, the Liouville form (\ref{1.2}) transforms into 
\begin{equation}
\lambda (\tilde{\alpha}_{\rho }^{(1)}):=(Y,dA)-(W,d<\nabla ,A>)=(Y+\nabla
W,dA):=(\tilde{Y},\text{ }dA),\ \tilde{Y}:=Y+\nabla W,  \label{1.9aa}
\end{equation}%
giving rise to the corresponding canonical symplectic structure on $T^{\ast
}(M)$ as 
\begin{equation}
\tilde{\omega}_{\rho }^{(2)}:=d\lambda (\tilde{\alpha}_{\rho }^{(1)})=(d%
\tilde{Y},\wedge dA).  \label{1.9aaa}
\end{equation}%
Respectively, the Hamiltonian function (\ref{1.2c}), as a function\ on $\
T^{\ast }(M),$ transforms into 
\begin{equation}
\tilde{H}_{\rho }(A,\tilde{Y})=1/2[(\tilde{Y},\tilde{Y})+(\nabla \times
A,\nabla \times A)+(<\nabla ,A>,<\nabla ,A>)],  \label{1.9ab}
\end{equation}%
coinciding with the well-known Dirac-Fock-Podolsky \cite{BoSh,BoSh-1,DiFoPo}
Hamiltonian expression. The corresponding Hamiltonian equations on the
cotangent space $\ \ T^{\ast }(M)$ 
\begin{eqnarray*}
\partial A/\partial t &:&=\delta \tilde{H}/\delta \tilde{Y}=\tilde{Y},\text{
\ \ }\tilde{Y}:=-E-\nabla W, \\
\partial \tilde{Y}/\partial t &:&=-\delta \tilde{H}/\delta A=-\nabla \times
(\nabla \times A)+\nabla <\nabla ,A>,
\end{eqnarray*}%
describe true wave processes related to the Maxwell equations in vacuum,
which do not take into account boundary charge and current densities
conditions. Really, from (\ref{1.9ab}) we obtain that 
\begin{equation}
\partial ^{2}A/\partial t^{2}-\nabla ^{2}A=0\Longrightarrow \partial
E/\partial t+\nabla (\partial W/\partial t\text{\ }+\text{ }<\nabla
,A>)=-\nabla \times B,\text{\ }  \label{1.9abb}
\end{equation}%
giving rise to the true vector potential wave equation, but the
electromagnetic Faraday induction law is satisfied if one to impose
additionally the Lorentz condition (\ref{1.3a}).

To remedy this situation, we will apply to this symplectic space \ the
reduction technique devised in Subsection \ (\ref{Subsec_2a.1}). Namely,
owing to Theorem \ref{th_04}, the constructed above cotangent manifold $%
T^{\ast }(N)\ $ is symplectomorphic to the corresponding reduced phase space 
$\mathcal{\bar{M}}_{\rho },$ that is 
\begin{equation}
\mathcal{\bar{M}}_{\rho }\simeq \{(B;S)\in T^{\ast }(N):\ <\nabla
,E(S)>=\rho ,\text{ \ \ }<\nabla ,B>=0\}  \label{1.9b}
\end{equation}%
with the reduced canonical symplectic 2-form 
\begin{equation}
\omega _{\rho }^{(2)}(B,S)=(dB,\wedge dS=d\lambda (\alpha _{\rho
}^{(1)})(B,S),\text{ \ \ \ }\lambda (\alpha _{\rho }^{(1)})(B,S):=-(S,dB),
\label{1.10}
\end{equation}%
\ where we put, by definition,

\begin{equation}
\nabla \times S+F+\nabla W=-\tilde{Y}:=E+\nabla W,\text{ \ \ }<\nabla
,F>:=\rho ,  \label{1.10a}
\end{equation}%
for some fixed vector mapping $F\in C^{(1)}(D;\mathbb{E}^{3}),$ depending on
the imposed boundary conditions. The result (\ref{1.10}) follows right away
upon substituting the expression for the electric field \ $E=\nabla \times
S+F$ into the symplectic structure (\ref{1.9aaa}), and taking into account
that $dF=0$ in $\Lambda ^{1}(M).$ \ The Hamiltonian function (\ref{1.9ab})
reduces, respectively, to the following symbolic form:%
\begin{eqnarray}
H_{\rho }(B,S) &=&1/2[(B,B)+(\nabla \times S+F+\nabla W,\nabla \times
S+F+\nabla W)+  \notag \\
+( &<&\nabla ,(\nabla \times )^{-1}B>,<\nabla ,(\nabla \times )^{-1}B>)],
\label{1.11}
\end{eqnarray}%
where $"(\nabla \times )^{-1}"$ means, by definition, the corresponding
inverse curl-operation, mapping \cite{MaWe} the divergence-free subspace $C_{%
{div}}^{(1)}(D;\mathbb{E}^{3})\subset C^{(1)}(D;\mathbb{E}^{3})$ into
itself. As a result from (\ref{1.11}), \ the Maxwell equations (\ref{1.1})
become a canonical Hamiltonian system upon the reduced phase space $T^{\ast
}(N),$ endowed with the canonical symplectic structure (\ref{1.10}) and \
the modified Hamiltonian function (\ref{1.11}). Really, one easily obtains
that%
\begin{eqnarray}
\partial S/\partial t &:&=\delta H/\delta B=B-(\nabla \times )^{-1}\nabla
<\nabla ,(\nabla \times )^{-1}B>,  \label{1.11a} \\
\text{\ \ }\partial B/\partial t &:&=-\delta H/\delta S=-\nabla \times
(\nabla \times S+F+\nabla W):=-\nabla \times E,  \notag
\end{eqnarray}%
where we make use of the definition $E=\nabla \times S+F$ and the elementary
identity $\nabla \times \nabla =0.$ Thus, the second equation of (\ref{1.11a}%
) coincides with \ the second Maxwell equation of (\ref{1.1}) in the
classical form 
\begin{equation*}
\partial B/\partial t=-\nabla \times E.
\end{equation*}%
Moreover, from (\ref{1.10a}), owing to (\ref{1.11a}), one obtains via the
differentiation with respect to $t\in \mathbb{R}$ that 
\begin{eqnarray}
\partial E/\partial t &=&\partial F/\partial t+\nabla \times \partial
S/\partial t=  \label{1.11b} \\
&=&\partial F/\partial t+\nabla \times B,  \notag
\end{eqnarray}%
as well as, owing to (\ref{1.1a}), 
\begin{equation}
<\nabla ,\partial F/\partial t>=\partial \rho /\partial t=-<\nabla ,J>.
\label{1.11c}
\end{equation}%
So, we can find from (\ref{1.11c}) that, up to non-essential curl-terms $%
\nabla \times (\cdot ),$ the following relationship%
\begin{equation}
\partial F/\partial t=-J  \label{1.11d}
\end{equation}%
holds. Really, the current density vector $J\in C^{(1)}(D;\mathbb{E}^{3}),$
owing to the equation of continuity (\ref{1.1a}), is defined up to
curl-terms $\nabla \times (\cdot )$ which can be included into the
right-hand side of (\ref{1.11d}). Having now substituted (\ref{1.11d}) into (%
\ref{1.11b}), we obtain exactly the first Maxwell equation of (\ref{1.1}): 
\begin{equation}
\partial E/\partial t=\nabla \times B-J,  \label{1.11cd}
\end{equation}%
being supplemented, naturally, with the external boundary constraint
conditions 
\begin{equation}
\begin{array}{c}
<\nabla ,B>=0,\text{ \ \ }<\nabla ,E>=\rho , \\ 
\partial \rho /\partial t+<\nabla ,J>=0,%
\end{array}
\label{1.11cdd}
\end{equation}%
owing to the continuity relationship (\ref{1.1a}) and definition (\ref{1.9b}%
).

Concerning the wave equations, related to the Hamiltonian system (\ref{1.11a}%
), \ we obtain the following: the electric field $E$ is recovered from the
second equation as 
\begin{equation}
E:=-\partial A/\partial t-\nabla W,  \label{1.11dd}
\end{equation}%
where $W\in C^{(1)}(D;\mathbb{R})$ is some smooth function, depending on the
vector field $A\in M.$ To retrieve this dependence, we substitute (\ref%
{1.11d}) \ into equation (\ref{1.11cd}), having taken into account that $%
B=\nabla \times A:$%
\begin{equation}
\partial ^{2}A/\partial t^{2}-\nabla (\partial W/\partial t+<\nabla
,A>)=\nabla ^{2}A+J.  \label{1.11de}
\end{equation}%
\ With the above, if we now impose the Lorentz condition (\ref{1.3a}), we
obtain from (\ref{1.11de}) \ the corresponding true wave equations in the
space-time, taking into account the external charge and current density
conditions (\ref{1.11cdd}).

Notwithstanding our progress so far, the problem of fulfilling the Lorentz
constraint (\ref{1.3a}) naturally within the canonical Hamiltonian formalism
still remains to be completely solved. To this end, we are compelled to
analyze the structure of the Liouville 1-form\ (\ref{1.9aa}) for Maxwell
equations in vacuum on a slightly extended functional manifold $M\times L.$
As a first step, we rewrite 1-form (\ref{1.9aa}) as 
\begin{eqnarray}
\lambda (\tilde{\alpha}_{\rho }^{(1)}) &:&=(\tilde{Y},dA)=(Y+\nabla
W,dA)=(Y,dA)+  \notag \\
+(W,-d &<&\nabla ,A>):=(Y,dA)+(W,d\chi ),  \label{1.11e}
\end{eqnarray}%
where we put, by definition, 
\begin{equation}
\chi :=-<\nabla ,A>.  \label{1.11f}
\end{equation}%
Considering now the elements $(Y,A;\chi ,W)$ $\in T^{\ast }(M\times L)$ as
new canonical variables on the extended cotangent phase space $T^{\ast
}(M\times L),$ where $L:=C^{(1)}(D;\mathbb{R}),$ we can rewrite the
symplectic structure (\ref{1.9aaa}) in the following canonical form%
\begin{equation}
\tilde{\omega}_{\rho }^{(2)}:=d\lambda (\tilde{\alpha}_{\rho
}^{(1)})=(dY,\wedge dA)+(dW,\wedge d\chi ).  \label{1.11g}
\end{equation}%
Subject to the Hamiltonian function (\ref{1.9ab}) we obtain the expression 
\begin{equation}
H(A,Y;\chi ,W)=1/2[(Y-\nabla W,Y-\nabla W)+(\nabla \times A,\nabla \times
A)+(\chi ,\chi )],  \label{1.11h}
\end{equation}%
with respect to which the corresponding Hamiltonian equations take the form:%
\begin{eqnarray}
\partial A/\partial t &:&=\delta H/\delta Y=Y-\nabla W,\text{ \ \ }Y:=-E, 
\notag \\
\partial Y/\partial t &:&=-\delta H/\delta A=-\nabla \times (\nabla \times
A),  \notag \\
\partial \chi /\partial t &:&=\delta H/\delta W=<\nabla ,Y-\nabla W>,  \notag
\\
\partial W/\partial t &:&=-\delta H/\delta \chi =-\chi .  \label{1.11hh}
\end{eqnarray}%
From (\ref{1.11hh}) \ we obtain, owing to external boundary conditions (\ref%
{1.11cdd}), \ successively that%
\begin{eqnarray}
\partial B/\partial t+\nabla \times E &=&0,\text{ \ }\partial ^{2}W/\partial
t^{2}-\nabla ^{2}W=\rho ,  \label{1.11hi} \\
\partial E/\partial t-\nabla \times B &=&0,\text{ \ \ }\partial
^{2}A/\partial t^{2}-\nabla ^{2}A=-\nabla (\partial W/\partial t+<\nabla
,A>).  \notag
\end{eqnarray}%
As is seen, these equations describe electromagnetic Maxwell equations in
vacuum, but without the Lorentz condition (\ref{1.3a}). Thereby, as above,
we will apply to the symplectic structure (\ref{1.11g}) the reduction
technique devised in Section 2. We obtain that under transformations (\ref%
{1.10a}) the corresponding reduced manifold $\mathcal{\bar{M}}_{\rho }$
becomes endowed with the symplectic structure 
\begin{equation}
\bar{\omega}_{\rho }^{(2)}:=(dB,\wedge dS)+(dW,\wedge d\chi ),  \label{1.11i}
\end{equation}%
and the Hamiltonian (\ref{1.11h}) assumes the form 
\begin{equation}
H(S,B;\chi ,W)=1/2[(\nabla \times S+F+\nabla W,\nabla \times S+F+\nabla
W)+(B,B)+(\chi ,\chi )],  \label{1.11j}
\end{equation}%
whose Hamiltonian equations 
\begin{eqnarray}
\partial S/\partial t &:&=\delta H/\delta B=B,\text{ \ \ \ \ \ \ }\partial
W/\partial t:=-\delta H/\delta \chi =-\chi ,  \label{1.11k} \\
\partial B/\partial t &:&=-\delta H/\delta S=-\nabla \times (\nabla \times
S+F+\nabla W)=-\nabla \times E,  \notag \\
\partial \chi /\partial t &:&=\delta H/\delta W=-<\nabla ,\nabla \times
S+F+\nabla W>=-<\nabla ,E>-\Delta W,  \notag
\end{eqnarray}%
coincide completely with Maxwell equations (\ref{1.1}) under conditions (\ref%
{1.10a}), describing true wave processes in vacuum, as well as the
electromagnetic Maxwell equations, taking into account \textit{a priori}
both the imposed external boundary conditions (\ref{1.11cdd}) and the
Lorentz condition (\ref{1.3a}), solving the problem mentioned in \cite%
{BoSh,DiFoPo}. Really, it is easy to obtain from (\ref{1.11k}) that 
\begin{eqnarray}
\partial ^{2}W/\partial t^{2}-\Delta W &=&\rho ,\text{ \ \ \ \ \ \ \ \ \ }%
\partial W/\partial t+<\nabla ,A>=0,  \label{1.11l} \\
\nabla \times B &=&J+\partial E/\partial t,\text{ \ \ \ \ \ \ \ }\partial
B/\partial t=-\nabla \times E,  \notag
\end{eqnarray}%
Based now on (\ref{1.11l}) and (\ref{1.11cdd}) one can easily calculate \cite%
{BoPrTa-1,BoPrTa} the magnetic wave equation%
\begin{equation}
\partial ^{2}A/\partial t^{2}-\Delta A=J,  \label{1.11m}
\end{equation}%
supplementing the suitable wave equation on the scalar potential $W\in L,$
finishing the calculations. Thus, we can formulate the following proposition.

\begin{proposition}
The electromagnetic Maxwell equations (\ref{1.1}) jointly with Lorentz
condition (\ref{1.3a}) are equivalent to the Hamiltonian system (\ref{1.11k}%
) \ with respect to the canonical symplectic structure \ (\ref{1.11i}) \ and
Hamiltonian function (\ref{1.11j}), which correspondingly reduce to
electromagnetic equations (\ref{1.11l}) and (\ref{1.11m}) under external
boundary conditions (\ref{1.11cdd}).
\end{proposition}

The obtained above result can be, eventually, used\ for developing an \
alternative quantization procedure\ of Maxwell \ electromagnetic equations,
being free of some quantum operator problems, discussed in detail in \cite%
{BoSh}. We hope to consider this aspect of quantization problem in a
specially devoted study.

\begin{remark}
{\ If one considers a motion of a charged point particle under a Maxwell
field, it is convenient to introduce a trivial fiber bundle structure $\pi
:M\rightarrow N,$ \ such that $M=N\times G$}${,}${\ $\ N:=D\subset \mathbb{R}%
^{3},$ with $G:=\mathbb{R}\backslash \{0${$\}$}\ being the corresponding
(abelian) structure Lie group. An analysis similar to the above gives rise
to the reduced (on the space }$\mathcal{\bar{M}}_{\xi }{:=l}${$^{-1}(\xi
)/G\simeq T^{\ast }(N),$ $\xi \in \mathcal{G}^{\ast }\mathcal{)}$\
symplectic structure }%
\begin{equation*}
{\bar{\omega}_{\xi }^{(2)}(q,p)=<dp,\wedge dq>+d<\xi ,\mathcal{A}(q,g)>_{%
\mathcal{G}},}
\end{equation*}%
{\ where $\mathcal{A}(q,g):=g^{-1}(d+{\xi }<A(q),dq>)g\in \mathcal{G}$ \ \ \
\ is a suitable connection 1-form on phase\ space $\ M,$ with $(q,p)\in
T^{\ast }(N)$ and $g\in G.$ The corresponding canonical Poisson brackets on $%
T^{\ast }(N)$ are easily found to be 
\begin{equation}
\{q^{i},q^{j}\}=0,\text{ \ \ }\{p_{j},q^{i}\}=\delta _{j}^{i},\text{ \ \ \ \
\ \ }\{p_{i},p_{j}\}=F_{ji}(q)  \label{1.21}
\end{equation}%
for all $(q,p)\in T^{\ast }(N),i,j=\overline{1,3}.$ If to introduce a new
momentum variable $\tilde{p}:=p+\xi A(q)$ on $T^{\ast }(N)\ni (q,p),$ it is
easy to verify that \ $\omega _{\xi }^{(2)}\rightarrow \tilde{\omega}_{\xi
}^{(2)}:=<d\tilde{p},\wedge dq>$}${,}${\ giving rise to the following
Poisson brackets \cite{Kupe,Pryk-Ampe,BlPrSa}: 
\begin{equation}
\{q^{i},q^{j}\}=0,\text{ \ \ \ \ }\{\tilde{p}_{j},q^{i}\}=\delta _{j}^{i},%
\text{ \ \ \ \ \ \ }\{\tilde{p}_{i},\tilde{p}_{j}\}=0,  \label{1.22}
\end{equation}%
where $i,j=\overline{1,3},$ iff \ \ the standard Maxwell field equations 
\begin{equation}
\partial F_{ij}/\partial q_{k}+\partial F_{jk}/\partial q_{i}+\partial
F_{ki}/\partial q_{j}=0  \label{1.23}
\end{equation}%
are satisfied on $N$ for all $i,j,k=\overline{1,3}$ with the curvature
tensor $F_{ij}(q):=\partial A_{j}/\partial q^{i}-\partial A_{i}/\partial
q^{j},$ \ $i,j=\overline{1,3},$ $q\in N.$}
\end{remark}

Such a construction permits a natural generalization to the case of
non-abelian structure Lie group yielding a description of Yang-Mills field
equations within the reduction approach, to which we proceed below.

\subsection{The Hamiltonian analysis of the Yang-Mills type dynamical systems%
}

As above, we start with defining a phase space $M$ of a particle \ \ \ under
a Yang-Mills field in a region $D\subset \mathbb{R}^{3}$ as $M:=D\mathbb{%
\times }G,$ where $G$ is a (not in general semi-simple) Lie group, acting on 
$M$ \ from the right. Over the space $M$ one can define quite naturally a
connection $\Gamma (\mathcal{A})$ \ if we consider the following trivial
principal fiber bundle $p:M\rightarrow N,$ where $N:=D,$ with the structure
group $G.$ Namely, if $g\in G,$ $\ q\in N,$ then a connection 1-form on $%
M\ni (q,g)$ can be expressed \cite{GiSt,PrMy,HePrPr} as 
\begin{equation}
\mathcal{A}(q;g):=g^{-1}(d+\sum_{i=1}^{n}a_{i}A^{(i)}(q))g,  \label{2.1Y}
\end{equation}%
where $\{a_{i}\in \mathcal{G}:i=\overline{1,n}\}$ is a basis of the Lie
algebra $\mathcal{G}$ \ of the Lie group $G$, and $A_{i}:D\rightarrow
\Lambda ^{1}(D),$ $i=\overline{1,n},$ are the Yang-Mills fields on the
physical space \ \ $D\subset \mathbb{R}^{3}.$

Now one defines the natural left invariant Liouville form on $M$ \ as 
\begin{equation}
\alpha ^{(1)}(q;g):=<p,dq>+<y,g^{-1}dg>_{\mathcal{G}},  \label{2.2Y}
\end{equation}%
where $y\in T^{\ast }(G)$ and $\ <\cdot ,\cdot >_{\mathcal{G}}$ denotes, \
as before, the usual Ad-invariant non-degenerate bilinear form on $\mathcal{G%
}^{\ast }\times \mathcal{G},$ as, evidently, $g^{-1}dg\in \Lambda
^{1}(G)\otimes \mathcal{G}\mathbf{.}$ The main assumption we need to proceed
is that the connection 1-form is compatible with the Lie group $G$ \ action
on $M.$ The latter means that the condition 
\begin{equation}
R_{h}^{\ast }\mathcal{A}(q;g)=Ad_{h^{-1}}\mathcal{A}(q;g)  \label{2.3Y}
\end{equation}%
is satisfied for all \ \ $(q,g)\in M$ and $h\in G,$ where $%
R_{h}:G\rightarrow G$ means the right translation by an element $h\in G$ on
the Lie group $\ G.$

Having stated all preliminary conditions needed for the reduction Theorem %
\ref{th_04} to be applied to our model, suppose that the Lie group $G$
canonical action on $M$ is naturally lifted to that on the cotangent space $%
T^{\ast }(M)$ endowed due to (endowed owing to ({\ref{1.2}) with the
following $G$-invariant canonical symplectic structure: \ \ \ \ 
\begin{eqnarray}
\omega ^{(2)}(q,p;g,y) &:&=d\text{ }pr^{\ast }\alpha
^{(1)}(q,p;g,y)=<dp,\wedge dq>+  \label{2.4Y} \\
+ &<&dy,\wedge g^{-1}dg>_{\mathcal{G}}+<ydg^{-1},\wedge dg>_{\mathcal{G}} 
\notag
\end{eqnarray}%
for all $(q,p;g,y)\in T^{\ast }(M).$ Take now an element $\xi \in \mathcal{G}%
^{\ast }$ and assume that its isotropy subgroup $G_{\xi }=G,$ that is $%
Ad_{h}^{\ast }\xi =\xi $ for all $h\in G.$ In the general case such an
element $\xi \in \mathcal{G}^{\ast }$ cannot exist but trivial $\xi =0,$ as
it happens, for instance, in the case of the Lie group $G=SL_{2}(\mathbb{R}%
). $ Then one can construct the reduced phase space $l^{-1}(\xi )/G$
symplectomorphic to $(T^{\ast }(N),\omega _{\xi }^{(2)}),$ where owing to (%
\ref{0.12}) for any $(q,p)\in T^{\ast }(N)$%
\begin{eqnarray}
\omega _{\xi }^{(2)}(q,p) &=&<dp,\wedge dq>+<\Omega ^{(2)}(q),\xi >_{%
\mathcal{G}}=  \label{2.5Y} \\
&=&<dp,\wedge
dq>+\sum_{s=1}^{n}\sum_{i,j=1}^{3}e_{s}F_{ij}^{(s)}(q)dq^{i}\wedge dq^{j}. 
\notag
\end{eqnarray}%
In the above we have expanded the element $\xi =\sum_{i=1}^{n}e_{i}a^{i}$ }$%
\in \mathcal{G}^{\ast }${{\ }with respect to the bi-orthogonal basis $%
\{a^{i}\in \mathcal{G}^{\ast },a_{j}\in \mathcal{G}:$ $\ <a^{i},a_{j}>_{%
\mathcal{G}}=\delta _{j}^{i},$ $i,j=\overline{1,n}\},$ \ with $e_{i}\in 
\mathbb{R},$ $i=\overline{1,3},$ being some constants, and we, as well,
denoted by $F_{ij}^{(s)}(q),$ $i,j=\overline{1,3},$ $s=\overline{1,n},$ the
corresponding curvature 2-form $\Omega ^{(2)}\in \Lambda ^{2}(N)\otimes 
\mathcal{G}$ components, that is 
\begin{equation}
\Omega ^{(2)}(q):=\sum_{s=1}^{n}\sum_{i,j=1}^{3}a_{s\text{ }%
}F_{ij}^{(s)}(q)dq^{i}\wedge dq^{j}  \label{2.6Y}
\end{equation}%
for any point $q\in N.$ Summarizing the calculations accomplished above, we
can formulate the following result. }

\begin{theorem}
\label{th_2.1} {\ Suppose the Yang-Mills field (\ref{2.1Y}) on the fiber
bundle }$p:M\rightarrow N$ {\ with }$M=D\times G$ {\ is invariant with
respect to the Lie group }$G$ {\ action }$G\times M\rightarrow M.$ {\
Suppose also that an element }$\xi \in \mathcal{G}^{\ast }$ {\ is chosen so
that }$Ad_{G}^{\ast }\xi =\xi .$ {\ Then for the naturally constructed
momentum mapping }$l:T^{\ast }(M)\rightarrow G^{\ast }$ {\ (being
equivariant) the reduced phase space }$l^{-1}(\xi )/G\simeq T^{\ast }(N)$ {\
is endowed with the symplectic structure (\ref{2.5}), having the following
component-wise Poisson brackets form:} 
\begin{equation}
\{p_{i},q^{j}\}_{\xi }=\delta _{i}^{j},\text{ \ \ }\{q^{i},q^{j}\}_{\xi
}=0,\ \ \{p_{i},p_{j}\}_{\xi }=\sum_{s=1}^{n}e_{s}F_{ji}^{(s)}(q)
\label{2.7Y}
\end{equation}%
{for all }$i,j=\overline{1,3}$ {\ and }$(q,p)\in T^{\ast }(N).$
\end{theorem}

The respectively extended Poisson bracket on the whole cotangent space $%
T^{\ast }(M)$ amounts owing to (\ref{1.4}) into the following set of Poisson
relationships: 
\begin{eqnarray}
\{y_{s},y_{k}\}_{\xi } &=&\sum_{r=1}^{n}c_{sk\text{ }}^{r}y_{r},\text{ \ \ \
\ \ \ \ \ \ \ \ \ }\ \ \{p_{i},q^{j}\}_{\xi }=\ \delta _{i}^{j}\ \ ,\text{ }
\label{2.8Y} \\
\text{\ }\{y_{s},p_{j}\}_{\xi } &=&0=\{q^{i},q^{j}\},\text{\ \ }%
\{p_{i},p_{j}\}_{\xi }=\sum_{s=1}^{n}y_{s\text{ }}F_{ji}^{(s)}(q),  \notag
\end{eqnarray}%
where $i,j=\overline{1,n},$ $\ c_{sk}^{r}\in \mathbb{R},$ \ $s,k,r=\overline{%
1,m},$ are the structure constants of the Lie algebra $\mathcal{G},$ and we
made use of the expansion \ $A^{(s)}(q)=\sum_{j=1}^{n}A_{j}^{(s)}(q)$ $%
dq^{j} $ as well we introduced alternative fixed values $e_{i}:=y_{i},$ $i=%
\overline{1,n}.$ The result (\ref{2.8Y}) can be easily seen if one one makes
a shift within the expression (\ref{2.4Y}) as $\sigma ^{(2)}\rightarrow
\sigma _{ext}^{(2)},$ where $\sigma _{ext}^{(2)}:=\left. \sigma
^{(2)}\right\vert _{\mathcal{A}_{0}\rightarrow \mathcal{A}}$ $,$ $\mathcal{A}%
_{0}(g):=g^{-1}dg,$ $g\in G.$ Thereby one can obtain in virtue of the
invariance properties of the connection $\Gamma (\mathcal{A})$ that 
\begin{equation*}
\sigma _{ext}^{(2)}(q,p;u,y)=<dp,\wedge dq>+d<y(g),Ad_{g^{-1}}\mathcal{A}%
(q;e)>_{\mathcal{G}}=
\end{equation*}%
\begin{equation*}
=<dp,\wedge dq>+<d\text{ }Ad_{g^{-1}}^{\ast }y(g),\wedge \mathcal{A}(q;e)>_{%
\mathcal{G}}=<dp,\wedge dq>+\sum_{s=1}^{m}dy_{s}\wedge du^{s}+
\end{equation*}

\begin{equation*}
+\sum_{j=1}^{n}\sum_{s=1}^{m}A_{j}^{(s)}(q)dy_{s}\wedge
dq-<Ad_{g^{-1}}^{\ast }y(g),\mathcal{A}(q,e)\wedge \mathcal{A}(q,e)>_{%
\mathcal{G}}+
\end{equation*}%
\begin{equation}
+\sum_{k\geq s=1}^{m}\sum_{l=1}^{m}y_{l}\text{ }c_{sk}^{l}\text{ }%
du^{k}\wedge du^{s}+\sum_{s=1}^{n}\sum_{i\geq
j=1}^{3}y_{s}F_{ij}^{(s)}(q)dq^{i}\wedge dq^{j},  \label{2.9Y}
\end{equation}%
where coordinate points $(q,p;u,y)\in T^{\ast }(M)$ \ are defined as
follows: $\mathcal{A}_{0}(e):=\sum_{s=1}^{m}du^{i}$ $a_{i},$ $%
Ad_{g^{-1}}^{\ast }y(g)=y(e):=\sum_{s=1}^{m}y_{s}$ $a^{s}$ for any element $%
g\in G.$ Hence one gets straightaway the Poisson brackets (2.8) plus
additional brackets connected with conjugated sets of variables $\{u^{s}\in 
\mathbb{R}:$\texttt{\ }$s=\overline{1,m}\}$ $\in \mathcal{G}^{\ast }$\ and $%
\{y_{s}\in \mathbb{R}:$\texttt{\ }$s=\overline{1,m}\}\in \mathcal{G}:$

\begin{equation}
\{y_{s},u^{k}\}_{\xi }=\delta _{s}^{k},\text{ \ }\{u^{k},q^{j}\}_{\xi }=0,%
\text{ \ }\{p_{j},u^{s}\}_{\xi }=A_{j}^{(s)}(q),\text{ \ }%
\{u^{s},u^{k}\}_{\xi }=0,  \label{2.10Y}
\end{equation}%
where $j=\overline{1,n},$ \ $k,s=\overline{1,m},\ $and $\ \ q\in N.$

Note here that the transition suggested above from the symplectic structure $%
\sigma ^{(2)}$ \ on $T^{\ast }(N)$ to its extension $\sigma _{ext}^{(2)}$ on
\ $T^{\ast }(M)$ just consists formally in adding to the symplectic
structure $\sigma ^{(2)}$ \ an exact part, which transforms it into an
equivalent one. Looking now at the expressions (\ref{2.9Y}), one can infer
immediately that an element \ $\xi :=\sum_{s=1}^{m}e_{s}a^{s}\in \mathcal{G}%
^{\ast }$ \ will be invariant with respect to the $Ad^{\ast }$-action of the
Lie group $\ G$ \ iff 
\begin{equation}
\left. \{y_{s},y_{k}\}_{\xi }\right\vert
_{y_{s}=e_{s}}=\sum_{r=1}^{m}c_{sk}^{r}\text{ }e_{r}\text{ }=0  \label{2.11Y}
\end{equation}%
identically for all $s,k=\overline{1,m},$ \ $j=\overline{1,n}$ \ and $\ q\in
N.$ In this, and only this case, the reduction scheme elaborated above will
go through.

Returning our attention to expression (\ref{2.10Y}), one can easily write
the following exact expression: 
\begin{equation}
\omega _{ext}^{(2)}(q,p;u,y)=\omega ^{(2)}(q,p+\sum_{s=1}^{n}y_{s}\text{ }%
A^{(s)}(q)\text{ };u,y),  \label{2.12Y}
\end{equation}%
on the phase space $T^{\ast }(M)\ni (q,p;u,y),$ where we abbreviated $\
<A^{(s)}(q),dq>$ as $\sum_{j=1}^{n}A_{j}^{(s)}(q)$ $dq^{j}.$ The
transformation like (\ref{2.12Y}) was discussed within somewhat different
contexts in articles \cite{Kupe,Pryk-Ampe} containing also a good background
for the infinite dimensional generalization of symplectic structure
techniques. Having observed from (\ref{2.12Y}) that the simple change of
variable 
\begin{equation}
\tilde{p}:=p+\sum_{s=1}^{m}y_{s}\text{ }A^{(s)}(q)  \label{2.13Y}
\end{equation}%
of the cotangent space $T^{\ast }(N)$ recasts our symplectic structure (\ref%
{2.9Y}) into the old canonical form (\ref{2.4Y}), one obtains that the
following new set of canonical Poisson brackets on $T^{\ast }(M)$ $\ni (q,%
\tilde{p};u,y):$%
\begin{eqnarray}
\{y_{s},y_{k}\}_{\xi } &=&\sum_{r=1}^{n}c_{sk}^{r}\text{ }y_{r},\text{ \ \ \ 
}\{\tilde{p}_{i},\tilde{p}_{j}\}_{\xi }=0,\text{\ \ \ \ \ }\{\tilde{p}%
_{i},q^{j}\}=\delta _{i}^{j},\text{ }  \label{2.14Y} \\
\{y_{s},q^{j}\}_{\xi } &=&0\text{ }=\{q^{i},q^{j}\}_{\xi },\text{\ }%
\{u^{s},u^{k}\}_{\xi }=0,\text{ \ \ }\{y_{s},\tilde{p}_{j}\}_{\xi }=0,\text{
\ }  \notag \\
\{u^{s},q^{i}\}_{\xi } &=&0,\text{ \ \ \ \ \ \ \ \ \ }\{y_{s},u^{k}\}_{\xi
}=\delta _{s}^{k},\text{ \ \ \ \ \ \ \ \ \ \ }\{u^{s},\tilde{p}_{j}\}_{\xi
}=0,  \notag
\end{eqnarray}%
where $\ k,s=\overline{1,m}$ \ and $i,j=\overline{1,n},$ holds iff the
non-abelian Yang-Mills type field equations 
\begin{equation}
\partial F_{ij}^{(s)}/\partial q^{l}+\partial F_{jl}^{(s)}/\partial
q^{i}+\partial F_{li}^{(s)}/\partial q^{j}+  \label{2.15Y}
\end{equation}%
\begin{equation*}
+\sum_{k,r=1}^{m}c_{kr}^{s}(F_{ij}^{(k)}A_{l}^{(r)}+F_{jl}^{(k)}A_{i}^{(r)}+F_{li}^{(k)}A_{j}^{(r)})=0
\end{equation*}%
are fulfilled for all $\ s=\overline{1,m}$ \ and $i,j,l=\overline{1,n}$ on
the base manifold $\ N.$ This effect of complete reduction of gauge
Yang-Mills type variables from the symplectic structure (\ref{2.9Y}) is
known in literature \cite{Kupe} as the principle of minimal interaction and
appeared to be useful enough for studying different interacting systems as
in \cite{MaWe,PrZa}. We plan to continue further the study of the geometric
properties of reduced symplectic structures connected with such interesting
infinite-dimensional coupled dynamical systems of Yang-Mills-Vlasov,
Yang-Mills-Bogolubov and Yang-Mills-Josephson types \cite{MaWe,PrZa} as well
as their relationships with associated principal fiber bundles endowed with
canonical connection structures.

\bigskip

\section{\label{Sec_3}The modified Lorentz force, radiation theory and the
Abraham--Lorentz electron inertia problem}

\subsection{Introductory setting}

The elementary point charged particle, like electron, mass problem was
inspiring many physicists \cite{Jamm} \ from the past as J. J. Thompson,
G.G. Stokes, H.A. Lorentz, E. Mach, M. Abraham, P.A. M. Dirac, G.A. Schott
and others. Nonetheless, their studies have not given rise to a clear
explanation of this phenomenon that stimulated new researchers to tackle it
from different approaches based on new ideas stemming both from the
classical Maxwell-Lorentz electromagnetic theory, as in \cite%
{Bril,GiZa-1,GiZaLi,Pryk-Ampe,Feyn-1,Feyn-2,Hamm-1,Hamm-2,Kosy-1,Kosy-2,MaPi,Medi,Moro,Page,Papp,Pegg,Puth,Simu,Teit,WhFe,YaTr}%
, and modern quantum field theories of Yang-Mills and Higgs type, as in \cite%
{Anni,Higg,Hoof,Wilc-2} and others, whose recent and extensive review is
done in \cite{Wilc-1}.

In the present work I will mostly concentrate on detail analysis and
consequences of the Feynman proper time paradigm \cite%
{Feyn-1,Feyn-2,Dyso-1,Dyso-2} subject to deriving the electromagnetic
Maxwell equations and the related Lorentz like force expression considered
from the vacuum field theory approach, developed in works \cite%
{BoPr-foun,BoPr-Feyn,BoPrTaPr-Lore}, and further, on its applications to the
electromagnetic mass origin problem. Our treatment of this and related
problems, based on the least action principle within the Feynman proper time
paradigm \cite{Feyn-1}, has allowed to construct the respectively modified
Lorentz type equation for a moving in space and radiating energy charged
point particle. Our analysis also elucidates, in particular, the
computations of the self-interacting electron mass term in \cite{MaPi},
where there was proposed a not proper solution to the well known classical
Abraham-Lorentz \cite{Abra,Lore-1,Lore-2,Lore-3} and Dirac \cite{Dira}
electron electromagnetic "4/3-electron mass" problem. As a result of our
scrutinized studying the classical electromagnetic mass problem \ we have
stated that it can be satisfactory solved within the classical H. Lorentz
and M. Abraham reasonings augmented with the additional electron stability
condition, which was not taken before into account yet appeared to be very
important for balancing the related electromagnetic field and mechanical
electron momenta. The latter, following recent enough works \cite{Puth,Moro}%
, devoted to analyzing the electron charged shell model, can be realized
within there suggested \textit{pressure-energy compensation principle, }%
suitably\textit{\ } \ applied to the \ ambient electromagnetic energy
fluctuations and the own electrostatic Coulomb electron energy.

\subsection{Feynman proper time paradigm geometric analysis}

In this section, we will develop further the vacuum field theory approach
within the Feynman proper time paradigm, devised before in \cite%
{BoPrTaPr-Lore,BoPr-Feyn}, to the electromagnetic J.C. Maxwell and H.
Lorentz electron theories and show that they should be suitably modified:
namely, the basic Lorentz force equations should be generalized following
the Landau-Lifschitz least action recipe \cite{LaLi}, taking also into
account the pure electromagnetic field impact. When applied the devised
vacuum field theory approach to the classical electron shell model, the
resulting Lorentz force expression appears to \ \ satisfactorily explaine
the electron inertial mass term exactly coinciding with the electron
relativistic mass, thus confirming the well known assumption \cite{Jack,Rohr}
by M. Abraham and H. Lorentz.

As was reported by F. Dyson \cite{Dyso-1,Dyso-2}, the original Feynman
approach derivation of the electromagnetic Maxwell equations was based on an 
\textit{a priori} general form of the classical Newton type force, acting on
a charged point particle moving in three-dimensional \ space $\mathbb{R}^{3}$
endowed with the canonical Poisson brackets on the phase variables, \
defined on the associated tangent space $T\mathbb{(R}^{3}).$ As a result of
this approach there only the first part of the Maxwell equations were
derived, as the second part, owing to F. Dyson \cite{Dyso-1}, is related
with the charged matter nature, which appeared to be hidden. Trying to
complete this Feynman approach to the derivation of Maxwell's equations more
systematically we have observed \cite{BoPr-Feyn} that the original Feynman's
calculations, based on Poisson brackets analysis, were performed on the 
\textit{tangent space} $T\mathbb{(R}^{3})$ which is, subject to the problem
posed, not physically proper. The true Poisson brackets can be correctly
defined only on the \textit{coadjoint phase space} $T^{\ast }\mathbb{(R}%
^{3}),$ as seen from the classical Lagrangian equations and the related
Legendre transformation \cite{AbMa,Arno,Godb,BlPrSa} from $T\mathbb{(R}^{3})$
to $T^{\ast }\mathbb{(R}^{3}).$ Moreover, within this observation, the
corresponding dynamical Lorentz type equation for a charged point particle
should be written for the particle momentum, not for the particle velocity,
whose value is well defined only with respect to the proper relativistic
reference frame, associated with the charged point particle owing to the
fact that the Maxwell equations are Lorentz invariant.

Thus, from the very beginning, we shall reanalyze the structure of the
Lorentz force exerted on a moving charged point particle with a charge $\xi
\in\mathbb{R}$ by another point charged particle with a charge $\xi_{f}\in%
\mathbb{R}$, making use of the classical Lagrangian approach, and rederive
the corresponding electromagnetic Maxwell equations. The latter appears to
be strongly related to the charged point mass structure of the
electromagnetic origin as was suggested by R. Feynman and F. Dyson.

Consider a charged point particle moving in an electromagnetic field. For
its description, it is convenient to introduce a trivial fiber bundle
structure ${\pi }${$:\mathcal{M}\rightarrow \mathbb{R}^{3},\mathcal{M}=%
\mathbb{R}^{3}\times G$}$,$ {with the abelian structure group $G:=\mathbb{R}%
\backslash \{0${$\},$} equivariantly acting on the canonically symplectic
coadjoint space }$T^{\ast }(\mathcal{M})$ {\ endowed both with the canonical
symplectic structure }%
\begin{align}
\omega ^{(2)}(p,y;r,g)& :=d\text{ }pr^{\ast }\alpha ^{(1)}(r,g)=<dp,\wedge
dr>+  \label{2.0} \\
+& <dy,\wedge g^{-1}dg>_{\mathcal{G}}+<ydg^{-1},\wedge dg>_{\mathcal{G}} 
\notag
\end{align}%
{for all $(p,y;r,g)\in T^{\ast }(\mathcal{M}),$ \ where }$\alpha
^{(1)}(r,g):=<p,dr>+<y,g^{-1}dg>_{\mathcal{G}}\in T^{\ast }(\mathcal{M})$ is
the corresponding Liouville form on $\mathcal{M},${\ and with a connection
one-form }$\mathcal{A}:{M}\rightarrow {T^{\ast }(M)\times \mathcal{G}}\ \ $%
as 
\begin{equation}
{\mathcal{A}}(r,g):=g^{-1}<\xi A(r),dr>g+g^{-1}dg,  \label{2.1}
\end{equation}%
with $\xi \in \mathcal{G}^{\ast },(r,g)\in \mathbb{R}^{3}\times G,$ and $\ \
<\cdot ,\cdot >$ $\ $being the scalar product in $\mathbb{E}^{3}.$ The
corresponding curvature 2-form $\Sigma ^{(2)}\in \Lambda ^{2}(\mathbb{R}%
^{3})\otimes $ $\mathcal{G}$ \ is 
\begin{equation}
\Sigma ^{(2)}(r):=d{\mathcal{A}}(r,g)+{\mathcal{A}}(r,g)\wedge {\mathcal{A}}%
(r,g)=\xi \sum_{i,j=1}^{3}F_{ij}(r)dr^{i}\wedge dr^{j},  \label{2.2}
\end{equation}%
where%
\begin{equation}
F_{ij}(r):=\frac{\partial A_{j}}{\partial r^{i}}-\frac{\partial A_{i}}{%
\partial r^{j}}  \label{2.3}
\end{equation}%
{for $i,j=\overline{1,3}$ \ is the electromagnetic tensor with respect to
the reference frame }$\mathcal{K}_{t},$ characterized by the phase space {%
coordinates $(r,p)\in T^{\ast }(\mathbb{R}^{3})$. }As an element $\xi \in 
\mathcal{G}^{\ast }$ is still not fixed, it is natural to apply the standard 
\cite{AbMa,Arno,BlPrSa} invariant Marsden--Weinstein--Meyer reduction to the
orbit factor space $\ \ \tilde{P}_{\xi }:=P_{\xi }/G_{\xi }\ $ subject to
the related momentum mapping $l:T^{\ast }\mathbb{(}\mathcal{M})\rightarrow 
\mathcal{G}^{\ast },$ constructed with respect to the canonical symplectic
structure \ (\ref{2.0}) \ on $T^{\ast }\mathbb{(}\mathcal{M}),$ where, by
definition, $\xi \in \mathcal{G}^{\ast }$ is constant, $P_{\xi }:=l^{-1}(\xi
)\ \subset $ $T^{\ast }\mathbb{(}\mathcal{M})$ and $G_{\xi }=\{g\in
G:Ad_{G}^{\ast }\xi \}$ is the isotropy group of the element $\xi \in 
\mathcal{G}^{\ast }.$

As a result of the Marsden--Weinstein--Meyer reduction, one finds that $%
G_{\xi }\simeq G,$ the factor-space $\tilde{P}_{\xi }\simeq T^{\ast }\mathbb{%
(R}^{3})$ is endowed with a suitably reduced symplectic structure {$\bar{%
\omega}_{\xi }^{(2)}\in T^{\ast }($}$\tilde{P}_{\xi })$ and the
corresponding Poisson brackets on the reduced manifold $\tilde{P}_{\xi }\ $%
are 
\begin{align}
\{r^{i},r^{j}\}_{\xi }& =0,\text{ }\{p_{j},r^{i}\}_{\xi }=\delta _{j}^{i},
\label{2.4} \\
\{p_{i},p_{j}\}_{\xi }& =\xi F_{ij}(r)  \notag
\end{align}%
for $i,j=\overline{1,3},$ considered with respect to the reference frame $%
\mathcal{K}_{t}.$ {Introducing a new momentum variable }%
\begin{equation}
{\tilde{\pi}:=p+\xi A(r)}  \label{2.4a}
\end{equation}%
{\ on }$\tilde{P}_{\xi }${$,$} {it is easy to verify that \ $\bar{\omega}%
_{\xi }^{(2)}\rightarrow \tilde{\omega}_{\xi }^{(2)}:=<d{\tilde{\pi}},\wedge
dr>$}${,}${\ giving rise to the following \textit{\textquotedblleft minimal
interaction\textquotedblright } \ \ canonical Poisson brackets: \ 
\begin{equation}
\{r^{i},r^{j}\}_{{\tilde{\omega}_{\xi }^{(2)}}}=0,\;\{{\tilde{\pi}}%
_{j},r^{i}\}_{{\tilde{\omega}_{\xi }^{(2)}}}=\delta _{j}^{i},\;\{{\tilde{\pi}%
}_{i},{\tilde{\pi}}_{j}\}_{{\tilde{\omega}_{\xi }^{(2)}}}=0  \label{2.4b}
\end{equation}%
for $i,j=\overline{1,3}$ \ with respect to some new reference frame $%
\mathcal{\tilde{K}}_{t^{\prime }},$ characterized by the phase space
coordinates }${(r,\tilde{\pi})}\in \tilde{P}_{\xi }$ and a new evolution
parameter $t^{\prime }\in \mathbb{R}${\ \ if and only if the Maxwell field
compatibility equations 
\begin{equation}
\partial F_{ij}/\partial r_{k}+\partial F_{jk}/\partial r_{i}+\partial
F_{ki}/\partial r_{j}=0  \label{2.5}
\end{equation}%
are satisfied on }$\mathbb{R}^{3}${\ for all $i,j,k=\overline{1,3}$ with the
curvature tensor \ (\ref{2.3}).}

Now we proceed to a dynamic description of the interaction between two
moving charged point particles $\xi $ and $\xi _{f},$ moving respectively,
with the velocities $u:=dr/dt$ and $u_{f}:=dr_{f}/dt$ subject to the
reference frame $\mathcal{K}_{t}.$ Unfortunately, there is a fundamental
problem in correctly formulating a physically suitable action functional and
the related least \ action condition. There are clearly possibilities such as%
\begin{equation}
S_{p}^{(t)}:=\int_{t_{1}}^{t_{2}}dt\mathcal{L}_{p}^{(t)}[r;dr/dt]
\label{2.6a}
\end{equation}%
on a temporal interval $[t_{1},t_{2}]\subset \mathbb{R}$ with respect to the
laboratory reference frame $\mathcal{K}_{t},$%
\begin{equation}
S_{p}^{(t^{\prime })}:=\int_{t_{1}^{\prime }}^{t_{2}^{\prime }}dt^{\prime }%
\mathcal{L}_{p}^{(t^{\prime })}[r;dr/dt^{\prime }]  \label{2.6b}
\end{equation}%
on a temporal interval $[t_{1}^{\prime },t_{2}^{\prime }]\subset \mathbb{R}$
with respect to the moving reference frame $\mathcal{K}_{t^{\prime }}$ and 
\begin{equation}
S_{p}^{(\tau )}:=\int_{\tau _{1}}^{\tau _{2}}d\tau \mathcal{L}_{p}^{(\tau
)}[r;dr/d\tau ]  \label{2.6c}
\end{equation}%
on a temporal interval $[\tau _{1},\tau _{2}]\subset \mathbb{R}$ with
respect to the proper time reference frame $\mathcal{K}_{\tau },$ naturally
related to the moving charged point particle $\xi .$

It was first observed by Poincar\'{e} and Minkowski \- \cite{Paul} that the
temporal differential $\ d\tau $ is not a closed differential one-form,
which physically means that a particle can traverse many different paths in
space $\mathbb{R}^{3}$ with respect to the reference frame $\mathcal{K}_{t}$
during any given proper time interval $d\tau ,$ \textit{naturally} related
to its motion. This fact was stressed \cite{Eins-1,Eins-2,Mink,Paul,Poin} by
Einstein, Minkowski and Poincar\'{e}, and later exhaustively analyzed by R.
Feynman, who argued \cite{Feyn-1} that the dynamical equation of a moving
point charged particle is physically sensible only with respect to its
proper time reference frame. This is Feynman's proper time reference frame
paradigm, which was recently further elaborated and applied both to the
electromagnetic Maxwell equations in \cite{GiZa,GiZa-1,GiZaLi} and to the
Lorentz type equation for a moving charged point particle under external
electromagnetic field in \cite{BoPr-Feyn,BoPrTaPr-Lore,BoPr-foun,BlPrSa}. As
it was there argued from a physical point of view, the least action
principle should be applied only to the expression \ (\ref{2.6c}) written
with respect to the proper time reference frame $\mathcal{K}_{\tau },$ whose
temporal parameter $\tau \in \mathbb{R}$ is independent of an observer and
is a closed differential one-form. Consequently, this action functional is
also mathematically sensible, which in part reflects the Poincar\'{e}'s and
Minkowski's observation that the infinitesimal quadratic interval%
\begin{equation}
d\tau ^{2}=(dt^{\prime })^{2}-|dr-dr_{f}|^{2},  \label{2.6d}
\end{equation}%
relating the reference frames $\mathcal{K}_{t^{\prime }}$ and $\mathcal{K}%
_{\tau },$ can be invariantly used for the four-dimensional relativistic
geometry. The most natural way to contend with this problem is to first
consider the quasi-relativistic dynamics of the charged point particle $\xi $
\ with respect to the moving reference frame $\mathcal{K}_{t^{\prime }}$
subject to which the charged point particle $\xi _{f}$ \ is at rest.
Therefore, it possible to write down a suitable action functional \ (\ref%
{2.6b}), up to $O(1/c^{4}),$ as the light velocity $c\rightarrow \infty $,
where the quasi-classical Lagrangian function $\mathcal{L}_{p}^{(t^{\prime
})}[r;dr/dt^{\prime }]$ can be naturally chosen as 
\begin{equation}
\mathcal{L}_{p}^{(t^{\prime })}[r;dr/dt^{\prime }]:=m^{\prime }(r)\left\vert
dr/dt^{\prime }-dr_{f}/dt^{\prime }\right\vert ^{2}/2-\xi \varphi ^{\prime
}(r).  \label{2.8}
\end{equation}%
where $m^{\prime }(r)\in \mathbb{R}_{+}$ is the charged particle $\xi $
ineryial mass parameter and $\varphi ^{\prime }(r)$ is the potential
function generated by the charged particle $\xi _{f}$ \ at a point $r\in $ $%
\mathbb{R}^{3}$ with respect to the reference frame $\mathcal{K}_{t^{\prime
}}.$ Since the standard temporal relationships between reference frames $%
\mathcal{K}_{t}$ and $\mathcal{K}_{t^{\prime }}:$%
\begin{equation}
dt^{\prime }=dt(1-\left\vert dr_{f}/dt^{\prime }\right\vert ^{2})^{1/2},
\label{2.9}
\end{equation}%
as well as between the reference frames $\mathcal{K}_{t^{\prime }}$ and $%
\mathcal{K}_{\tau }:$ 
\begin{equation}
d\tau =dt^{\prime }(1-\left\vert dr/dt^{\prime }-dr_{f}/dt^{\prime
}\right\vert ^{2})^{1/2},  \label{2.10}
\end{equation}%
give rise, up to $O(1/c^{2}),$ as $c\rightarrow \infty ,$ to $dt^{\prime
}\simeq dt$ and $d\tau \simeq dt^{\prime },$ respectively, it is easy to
verify that the least action condition $\delta S_{p}^{(t^{\prime })}=0$ is
equivalent to the dynamical equation%
\begin{equation}
d\pi /dt=\nabla \mathcal{L}_{p}^{(t^{\prime })}[r;dr/dt]=(\frac{1}{2}%
\left\vert dr/dt-dr_{f}/dt\right\vert ^{2})\nabla m-\xi \nabla \varphi (r),
\label{2.11}
\end{equation}%
where we have defined the generalized canonical momentum as 
\begin{equation}
\pi :=\partial \mathcal{L}_{p}^{(t^{\prime })}[r;dr/dt]/\partial
(dr/dt)=m(dr/dt-dr_{f}/dt),  \label{2.12}
\end{equation}%
with the dash signs dropped and denoted by \textquotedblleft $\nabla $%
\textquotedblright\ the usual gradient operator in $\mathbb{E}^{3}.$
Equating the canonical momentum expression \ (\ref{2.12}) with respect to
the reference frame $\mathcal{K}_{t^{\prime }}$ to that of (\ref{2.4a}) \ {%
with respect to the canonical reference frame $\mathcal{\tilde{K}}%
_{t^{\prime }},$ \ and identifying the reference frame $\mathcal{\tilde{K}}$}%
$_{{t}^{\prime }}${\ with }$\mathcal{K}_{t^{\prime }},$ one obtains that 
\begin{equation}
m(dr/dt-dr_{f}/dt)=mdr/dt-\xi A(r),  \label{2.12a}
\end{equation}%
giving rise to the important inertial particle mass determining expression%
\begin{equation}
m\ =-\xi \varphi (r),\   \label{2.13}
\end{equation}%
which right away follows from the relationship%
\begin{equation}
\varphi (r)dr_{f}/dt=A(r).\   \label{2.14}
\end{equation}%
The latter is well known in the classical electromagnetic theory \cite%
{Jack,LaLi} for potentials $(\varphi ,A)\in T^{\ast }(M^{4})$ satisfying the
Lorentz condition%
\begin{equation}
\partial \varphi (r)/\partial t+<\nabla ,A(r)>=0,  \label{2.15}
\end{equation}%
yet the expression (\ref{2.13}) looks very nontrivial in relating the 
\textit{\textquotedblleft inertial\textquotedblright } \ mass of the charged
point particle $\xi $ \ to the electric potential, being both generated by
the ambient charged point particles $\xi _{f}.$ \ As was argued in articles 
\cite{BoPr-foun,BoPr-Feyn, BoPrTa-acti}, the above mass phenomenon is
closely related and from a physical perspective shows its deep relationship
to the classical electromagnetic mass problem.

Before further analysis of the completely relativistic the charge $\xi $
motion under consideration, we substitute the mass expression (\ref{2.13})
into the quasi-relativistic action functional (\ref{2.6b}) with the
Lagrangian (\ref{2.8}). As a result, we obtain two possible action
functional expressions, taking into account two main temporal parameters
choices:%
\begin{equation}
S_{p}^{(t^{\prime })}=-\int_{t_{1}^{\prime }}^{t_{2}^{\prime }}\xi \varphi
^{\prime }(r)(1+\frac{1}{2}\left\vert dr/dt^{\prime }-dr_{f}/dt^{\prime
}\right\vert ^{2})dt^{\prime }  \label{2.15a}
\end{equation}%
on an interval $[t_{1}^{\prime },t_{2}^{\prime }]\subset \mathbb{R},$ or 
\begin{equation}
S_{p}^{(\tau )}=-\int_{\tau _{1}}^{\tau _{2}}\xi \varphi ^{\prime }(r)(1+%
\frac{1}{2}\left\vert dr/d\tau -dr_{f}/d\tau \right\vert ^{2})d\tau
\label{2.15b}
\end{equation}%
on an $[\tau _{1},\tau _{2}]\subset \mathbb{R}$. \ The direct relativistic
transformations of (\ref{2.15b}) entail that 
\begin{align}
S_{p}^{(\tau )}& =-\int_{\tau _{1}}^{\tau _{2}}\xi \varphi ^{\prime }(r)(1+%
\frac{1}{2}\left\vert dr/d\tau -dr_{f}/d\tau \right\vert ^{2})d\tau \simeq
\label{2.15c} \\
& \simeq -\int_{\tau _{1}}^{\tau _{2}}\xi \varphi ^{\prime }(r)(1+\left\vert
dr/d\tau -dr_{f}/d\tau \right\vert ^{2})^{1/2}d\tau =  \notag \\
& =-\int_{\tau _{1}}^{\tau _{2}}\xi \varphi ^{\prime }(r)(1-|dr/dt^{\prime
}-dr_{f}/dt^{\prime }|)^{-1/2}d\tau =-\int_{t_{1}^{\prime }}^{t_{2}^{\prime
}}\xi \varphi ^{\prime }(r)dt^{\prime },  \notag
\end{align}%
giving rise to the correct, from the physical point of view, relativistic
action functional form (\ref{2.6b}), suitably transformed to the proper time
reference frame representation (\ref{2.6c}) via the Feynman proper time
paradigm. Thus, we have shown that the true action functional procedure
consists in a physically motivated choice of either the action functional
expression form (\ref{2.6a}) or (\ref{2.6b}). Then, it is transformed to the
proper time action functional representation form (\ref{2.6c}) within the
Feynman paradigm, and the least action principle is applied.

Concerning the above discussed problem of describing the motion of a charged
point particle $\xi $ \ in the electromagnetic field generated by another
moving charged point particle $\xi _{f},$ it must be mentioned that we have
chosen the quasi-relativistic functional expression (\ref{2.8}) in the form (%
\ref{2.6b}) with respect to the moving reference frame $\mathcal{K}%
_{t^{\prime }},$ because its form is physically reasonable and acceptable,
since the charged point particle $\xi _{f}$ \ is then at rest, generating no
magnetic field.

Based on the above relativistic action functional expression 
\begin{equation}
S_{p}^{(\tau )}:=-\int_{\tau _{1}}^{\tau _{2}}\xi \varphi ^{\prime
}(r)(1+\left\vert dr/d\tau -dr_{f}/d\tau \right\vert ^{2})^{1/2}d\tau
\label{2.15d}
\end{equation}%
written with respect to the proper reference from $\mathcal{K}_{\tau },$ one
finds the following evolution equation:%
\begin{equation}
d\pi _{p}/d\tau =-\xi \nabla \varphi ^{\prime }(r)(1+\left\vert dr/d\tau
-dr_{f}/d\tau \right\vert ^{2})^{1/2},  \label{2.15e}
\end{equation}%
where the generalized momentum is given exactly by the relationship \ (\ref%
{2.12}): 
\begin{equation}
\pi _{p}=m(dr/dt-dr_{f}/dt).  \label{2.16}
\end{equation}%
Making use of the relativistic transformation (\ref{2.9}) and the next one \
(\ref{2.10}), the equation (\ref{2.15e}) is easily transformed to 
\begin{equation}
\frac{d}{dt}(p+\xi A)=-\nabla \varphi (r)(1-\left\vert u_{f}\right\vert
^{2}),  \label{2.17}
\end{equation}%
where we took into account the related definitions: (\ref{2.13}) for the
charged particle $\xi $ mass, (\ref{2.14}) for the magnetic vector potential
and $\varphi (r)=$ $\varphi ^{\prime }(r)/(1-\left\vert u_{f}\right\vert
^{2})^{1/2}$ for the scalar electric potential with respect to the
laboratory reference frame $\mathcal{K}_{t}.$ Equation (\ref{2.17}) can be
further transformed, using elementary vector algebra, to the classical
Lorentz type form:%
\begin{equation}
dp/dt=\xi E+\xi u\times B-\xi \nabla <u-u_{f},A>,  \label{2.18}
\end{equation}%
where 
\begin{equation}
E:=-\partial A/\partial t-\nabla \varphi  \label{2.19}
\end{equation}%
is the related electric field and 
\begin{equation}
B:=\nabla \times A  \label{2.20}
\end{equation}%
is the related magnetic field, exerted by the moving charged point particle $%
\xi _{f}$ \ on the charged point particle $\xi $ \ with respect to the
laboratory reference frame $\mathcal{K}_{t}.$ The Lorentz type force
equation \ (\ref{2.18}) was obtained in \cite{BoPr-Feyn,BoPr-foun} in terms
of the moving reference frame $\mathcal{K}_{t^{\prime }},$ and recently
reanalyzed in \cite{BoPr-foun,Pryk-Ampe}. The obtained results follow in
part \cite{Rous-1,Rous-2} from Amp\`{e}re's classical works on constructing
the magnetic force between two neutral conductors with stationary currents.

\subsection{Analysis of the Maxwell and Lorentz force equations}

As a moving charged particle $\xi _{f}$ generates the suitable electric
field (\ref{2.19}) and magnetic field (\ref{2.20}) via their electromagnetic
potential $(\varphi ,A)\in T^{\ast }(M^{4})$ with respect to a laboratory
reference frame $\mathcal{K}_{t},$ we will supplement them naturally by
means of the external material equations describing the relativistic charge
conservation law:%
\begin{equation}
\partial \rho /\partial t+<\nabla ,J>=0,  \label{3.1}
\end{equation}%
where $(\rho ,J)\in T^{\ast }(M^{4})$ is a related four-vector for the
charge and current distribution in the space $\mathbb{R}^{3}.$ Moreover, one
can augment the equation (\ref{3.1}) with the experimentally well
established Gauss law%
\begin{equation}
<\nabla ,E>=\rho  \label{3.2}
\end{equation}%
to calculate the quantity $\Delta \varphi :=<\nabla ,\nabla \varphi >$ from
the expression (\ref{2.19}): 
\begin{equation}
\Delta \varphi =-\frac{\partial }{\partial t}<\nabla ,A>-<\nabla ,E>.
\label{3.3}
\end{equation}%
Having taken into account the relativistic Lorentz condition (\ref{2.15})
and the expression (\ref{3.2}) one easily finds that the wave equation 
\begin{equation}
\partial ^{2}\varphi /\partial t^{2}-\Delta \varphi =\rho  \label{3.4}
\end{equation}%
holds with respect to the laboratory reference frame $\mathcal{K}_{t}.$
Applying the rot-operation \textquotedblleft $\nabla \times $%
\textquotedblright\ to the expression (\ref{2.19}) we obtain, owing to the
expression (\ref{2.20}), the equation 
\begin{equation}
\nabla \times E+\partial B/\partial t=0,  \label{3.5}
\end{equation}%
giving rise, together with (\ref{3.2}), to the first pair of the classical
Maxwell equations. To obtain the second pair of the Maxwell equations, it is
first necessary to apply the rot-operation \textquotedblleft $\nabla \times $%
\textquotedblright to the expression (\ref{2.20}): 
\begin{equation}
\nabla \times B=\partial E/\partial t+(\partial ^{2}A/\partial t^{2}-\Delta
A)  \label{3.6}
\end{equation}%
and then apply $-\partial /\partial t$ to the wave equation (\ref{3.4}) to
obtain%
\begin{equation}
\begin{array}{c}
-\frac{\partial ^{2}}{\partial t^{2}}(\frac{\partial \varphi }{\partial t}%
)+<\nabla ,\nabla \frac{\partial \varphi }{\partial t}>=\frac{\partial ^{2}}{%
\partial t^{2}}<\nabla ,A>- \\ 
\\ 
-<\nabla ,\nabla <\nabla ,A>>=<\nabla ,\frac{\partial ^{2}A}{\partial t^{2}}%
-\nabla \times (\nabla \times A)-\Delta A>= \\ 
\\ 
=<\nabla ,\frac{\partial ^{2}A}{\partial t^{2}}-\Delta A>=<\nabla ,J>.%
\end{array}%
\   \label{3.7}
\end{equation}%
The result (\ref{3.7}) leads to the relationship 
\begin{equation}
\partial ^{2}A/\partial t^{2}-\Delta A=J,  \label{3.8}
\end{equation}%
if we take into account that both the vector potential $A\in \mathbb{E}^{3}$
and the vector of current $J\in \mathbb{E}^{3}$ are determined up to a
rot-vector expression $\nabla \times S$ for some smooth vector-function $%
S:M^{4}\rightarrow \mathbb{E}^{3}.$ Inserting the relationship \ (\ref{3.8})
into \ (\ref{3.6}), we obtain \ (\ref{3.5}) and the second pair of the
Maxwell equations:%
\begin{equation}
\nabla \times B=\partial E/\partial t+J,\text{ \ }\nabla \times E=\partial
B/\partial t.  \label{3.9}
\end{equation}%
It is important that the system of equations \ (\ref{3.9}) can be
represented by means of the least action principle $\delta S_{f-p}^{(t)}=0,$
where the action functional 
\begin{equation}
S_{f-p}^{(t)}:=\int_{t_{1}}^{t_{2}}dt\mathcal{L}_{f-p}^{(t)}  \label{3.10}
\end{equation}%
is defined on an interval $[t_{1},t_{2}]\subset \mathbb{R}$ \ by the
Landau-Lifschitz type \cite{LaLi} Lagrangian function 
\begin{equation}
\mathcal{L}_{f-p}^{(t)}=\int_{\mathbb{R}%
^{3}}d^{3}r((|E|^{2}-|B|^{2})/2+<J,A>-\rho \varphi )  \label{3.11}
\end{equation}%
with respect to the laboratory reference frame $\mathcal{K}_{t},$ which is
subject to the electromagnetic field a unique and physically reasonable.
From \ (\ref{3.11}) we deduce that the generalized field momentum 
\begin{equation}
\pi _{f}:=\partial \mathcal{L}_{f-p}^{(t)}/\partial (\partial A/\partial
t)=-E  \label{3.12}
\end{equation}%
and its evolution is given as 
\begin{equation}
\partial \pi _{f}/\partial t:=\delta \mathcal{L}_{f-p}^{(t)}/\delta A\
=J-\nabla \times B,  \label{3.13}
\end{equation}%
which is equivalent to the first Maxwell equation of \ (\ref{3.9}). As the
Maxwell equations allow the least action representation, it is easy to
derive\ \cite{AbMa,Arno,BlPrSa,BoPr-foun,BoPrTa-acti} their dual Hamiltonian
formulation with the Hamiltonian function 
\begin{equation}
H_{f-p}:=\int_{\mathbb{R}^{3}}d^{3}r<\pi _{f},\partial A/\partial t>-%
\mathcal{L}_{f-p}^{(t)}=\int_{\mathbb{R}%
^{3}}d^{3}r((|E|^{2}-|B|^{2})/2-<J,A>),  \label{3.14}
\end{equation}%
satisfying the invariant condition 
\begin{equation}
dH_{f-p}/dt=0  \label{3.15}
\end{equation}%
for all $t\in \mathbb{R}.$

It is worth noting here that the Maxwell equations were derived under the
important condition\ that the charged system $(\rho ,J)\in T(M^{4})$ exerts
no influence on the ambient electromagnetic field potentials $(\varphi
,A)\in T^{\ast }(M^{4}).$ As this is not actually the case owing to the
damping radiation reaction on accelerated charged particles, one can try to
describe this self-interacting influence by means of the modified least
action principle, making use of the Lagrangian expression \ (\ref{3.11})
recalculated with respect to the separately chosen charged particle $\xi $
endowed with the uniform shell model$\ $geometric structure and generating
this electromagnetic field.

Following the slightly modified well-known approach from \cite{LaLi} and
reasonings from \cite{Beck,Moro} this Landau-Lifschitz type Lagrangian \ (%
\ref{3.11}) can be recast (further in the Gauss units) as 
\begin{equation}
\begin{array}{c}
\mathcal{L}_{f-p}^{(t)}=\int_{\mathbb{R}^{3}}d^{3}r((|E|^{2}-|B|^{2})/2+%
\int_{\mathbb{R}^{3}}d^{3}r(\frac{1}{c}<J,A>-\rho \varphi )-<k(t),dr/dt>= \\ 
\\ 
=\int_{\mathbb{R}^{3}}d^{3}r(\frac{1}{2}<-\nabla \varphi -\frac{1}{c}%
\partial A/\partial t,-\nabla \varphi -\frac{1}{c}\partial A/\partial t>- \\ 
\\ 
-\frac{1}{2}<\nabla \times (\nabla \times A),A>)+\int_{\mathbb{R}^{3}}d^{3}r(%
\frac{1}{c}<J,A>-\rho \varphi )-<k(t),dr/dt>= \\ 
\\ 
=\int_{\mathbb{R}^{3}}d^{3}r(\frac{1}{2}<-\nabla \varphi ,E>-\frac{1}{2c}%
<\partial A/\partial t,E>-\frac{1}{2}<A,\nabla \times B>)+ \\ 
\\ 
+\int_{\mathbb{R}^{3}}(\frac{1}{c}<J,A>-\rho \varphi )-<k(t),dr/dt>=%
\end{array}
\label{3.16}
\end{equation}%
\begin{equation*}
\begin{array}{c}
=\int_{\mathbb{R}^{3}}d^{3}r(\frac{1}{2}\varphi <\nabla ,E>+\frac{1}{2c}%
<A,\partial E/\partial t>-\frac{1}{2c}<A,\ J+\partial E/\partial t>)+ \\ 
\\ 
+\int_{\mathbb{R}^{3}}(\frac{1}{c}<J,A>-\rho \varphi )-\frac{1}{2c}\frac{d}{%
dt}\int_{\mathbb{R}^{3}}d^{3}r<A,E>- \\ 
\\ 
-\frac{1}{2}\lim_{r\rightarrow \infty }\int_{\mathbb{S}_{r}^{2}}<\varphi
E+A\times B,dS_{r}^{2}>-<k(t),dr/dt>\ = \\ 
\\ 
=-\frac{1}{2}\int_{\mathbb{\Omega }_{+}(\xi )}d^{3}r(\frac{1}{c}<J,A>-\rho
\varphi )+\int_{\mathbb{\Omega }_{+}(\xi )\cup \Omega _{-}(\xi )}(\frac{1}{c}%
<J,A>-\rho \varphi )-<k(t),dr/dt>- \\ 
\\ 
-\frac{1}{2c}\frac{d}{dt}\int_{\mathbb{R}^{3}}d^{3}r<A,E>-\frac{1}{2}%
\lim_{r\rightarrow \infty }\int_{\mathbb{S}_{r}^{2}}<\varphi E+A\times
B,dS_{r}^{2}>=%
\end{array}%
\end{equation*}%
\begin{equation*}
\begin{array}{c}
=-\frac{1}{2}\int_{\mathbb{\Omega }_{+}(\xi )}d^{3}r(\frac{1}{c}<J,A>-\rho
\varphi )-\frac{1}{2}\int_{\mathbb{\Omega }_{-}(\xi )}d^{3}r(\frac{1}{c}%
<J,A>-\rho \varphi )+ \\ 
\\ 
+\frac{1}{2}\int_{\mathbb{\Omega }_{-}(\xi )}d^{3}r(\frac{1}{c}<J,A>-\rho
\varphi )+\int_{\mathbb{\Omega }_{+}(\xi )\cup \Omega _{-}(\xi )}(\frac{1}{c}%
<J,A>-\rho \varphi )-<k(t),dr/dt>- \\ 
\\ 
-\frac{1}{2c}\frac{d}{dt}\int_{\mathbb{R}^{3}}d^{3}r<A,E>-\frac{1}{2}%
\lim_{r\rightarrow \infty }\int_{\mathbb{S}_{r}^{2}}<\varphi E+A\times
B,dS_{r}^{2}>= \\ 
\\ 
=\frac{1}{2}\int_{\mathbb{\Omega }_{-}(\xi )}d^{3}r(\frac{1}{c}<J,A>-\rho
\varphi )-\frac{1}{2}\int_{\mathbb{\Omega }_{+}(\xi )\cup \Omega _{-}(\xi
)}d^{3}r(\frac{1}{c}<J,A>-\rho \varphi )+ \\ 
\\ 
+\int_{\mathbb{\Omega }_{+}(\xi )\cup \Omega _{-}(\xi )}(\frac{1}{c}%
<J,A>-\rho \varphi )-<k(t),dr/dt>- \\ 
\\ 
-\frac{1}{2c}\frac{d}{dt}\int_{\mathbb{R}^{3}}d^{3}r<A,E>-\frac{1}{2}%
\lim_{r\rightarrow \infty }\int_{\mathbb{S}_{r}^{2}}<\varphi E+A\times
B,dS_{r}^{2}>= \\ 
\\ 
=\frac{1}{2}\int_{\mathbb{\Omega }_{-}(\xi )}d^{3}r(\frac{1}{c}<J,A>-\rho
\varphi )+\frac{1}{2}\int_{\mathbb{\Omega }_{+}(\xi )\cup \Omega _{-}(\xi
)}d^{3}r(\frac{1}{c}<J,A>-\rho \varphi )- \\ 
\\ 
-\frac{1}{2c}\frac{d}{dt}\int_{\mathbb{R}^{3}}d^{3}r<A,E>-\frac{1}{2}%
\lim_{r\rightarrow \infty }\int_{\mathbb{S}_{r}^{2}}<\varphi E+A\times
B,dS_{r}^{2}>,%
\end{array}%
\end{equation*}%
where we have introduced still not determined a radiation damping force $%
k(t)\in \mathbb{E}^{3},$ have denoted by $\mathbb{\Omega }_{+}\mathbb{(\xi )}%
:=supp$ $\xi _{+}\subset \mathbb{R}^{3}\ $and $\ \mathbb{\Omega }_{-}\mathbb{%
(\xi )}:=supp$ $\xi _{-}$ $\subset \mathbb{R}^{3}$ the corresponding charge $%
\xi $ supports, located\ on the electromagnetic field shadowed rear and
electromagnetic field exerted front semispheres (see Fig.1) of the electron
shell, respectively to its motion with the fixed velocity $u(t)\in \mathbb{E}%
^{3},$ as well as we denoted by $\mathbb{S}_{r}^{2}$ \ a two-dimensional
sphere of radius $r\rightarrow \infty .$ 
\begin{eqnarray*}
&&...\text{ } \\
&&Fig.1
\end{eqnarray*}%
Having naturally assumed that the radiated charged particle energy at
infinity is negligible, the Lagrangian function \ (\ref{3.16}) becomes
equivalent to 
\begin{equation}
\begin{array}{c}
\mathcal{L}_{f-p}^{(t)}=\frac{1}{2}\int_{\mathbb{\Omega }_{-}(\xi )}d^{3}r(%
\frac{1}{c}<J,A>-\rho \varphi )+\frac{1}{2c}\int_{\mathbb{\Omega }_{+}(\xi
)\cup \Omega _{-}(\xi )}(<J,A>-\rho \varphi )-<k(t),dr/dt>,%
\end{array}
\label{3.17}
\end{equation}%
which we now need to additionally recalculate taking into account that the
electromagnetic potentials $(\varphi ,A)\in T^{\ast }(M^{4})$ are retarded,
generated by only the front part of the electron shell and given as $%
1/c^{2}\rightarrow 0$ in the following expanded into Lienard-Wiechert series
form:

\begin{equation}
\begin{array}{c}
\varphi =\left. \int_{\mathbb{R}^{3}}d^{3}r^{\prime }\frac{\rho (t^{\prime
},r^{\prime })}{|r-r^{\prime }|}\right\vert _{t^{\prime }=t-|r-r^{\prime
}|/c}=\lim_{\varepsilon \downarrow 0}\int_{\mathbb{R}^{3}}d^{3}r^{\prime }%
\frac{\rho (t-\varepsilon ,r^{\prime })}{|r-r^{\prime }|}+ \\ 
\\ 
+\lim_{\varepsilon \downarrow 0}\frac{1}{2c^{2}}\int_{\mathbb{R}%
^{3}}d^{3}r^{\prime }|r-r^{\prime }|\partial ^{2}\rho (t-\varepsilon
,r^{\prime })/\partial t^{2}+ \\ 
\\ 
+\lim_{\varepsilon \downarrow 0}\frac{1}{6c^{3}}\int_{\mathbb{R}%
^{3}}d^{3}r^{\prime }|r-r^{\prime }|^{2}\partial \rho (t-\varepsilon
,r^{\prime })/\partial t+O(1/c^{4})= \\ 
\\ 
=\int_{\mathbb{\Omega }_{+}(\xi )}d^{3}r^{\prime }\frac{\rho (t,r^{\prime })%
}{|r-r^{\prime }|}+\ \frac{1}{2c^{2}}\int_{\mathbb{\Omega }_{+}(\xi
)}d^{3}r^{\prime }|r-r^{\prime }|\partial ^{2}\rho (t,r^{\prime })/\partial
t^{2}+ \\ 
\\ 
+\frac{1}{6c^{3}}\int_{\mathbb{\Omega }_{+}(\xi )}d^{3}r^{\prime
}|r-r^{\prime }|^{2}\partial \rho (t,r^{\prime })/\partial t+O(1/c^{4}), \\ 
\end{array}
\label{3.18}
\end{equation}%
\begin{equation*}
\begin{array}{c}
A=\left. \frac{1}{c}\int_{\mathbb{R}^{3}}d^{3}r^{\prime }\frac{J(t^{\prime
},r^{\prime })}{|r-r^{\prime }|}\right\vert _{t^{\prime }=t-|r-r^{\prime
}|/c}=\lim_{\varepsilon \downarrow 0}\frac{1}{c}\int_{\mathbb{R}%
^{3}}d^{3}r^{\prime }\frac{J(t-\varepsilon ,r^{\prime })}{|r-r^{\prime }|}-
\\ 
\\ 
-\lim_{\varepsilon \downarrow 0}\frac{1}{c^{2}}\int_{\mathbb{R}%
^{3}}d^{3}r^{\prime }\partial J(t-\varepsilon ,r^{\prime })/\partial t+ \\ 
\\ 
+\lim_{\varepsilon \downarrow 0}\frac{1}{2c^{3}}\int_{\mathbb{R}%
^{3}}d^{3}r^{\prime }|r-r^{\prime }|\partial ^{2}J(t-\varepsilon ,r^{\prime
})/\partial t^{2}+O(1/c^{4})= \\ 
\\ 
=\frac{1}{c}\int_{\mathbb{\Omega }_{+}(\xi )}d^{3}r^{\prime }\frac{%
J(t,r^{\prime })}{|r-r^{\prime }|}-\frac{1}{c^{2}}\int_{\mathbb{\Omega }%
_{+}(\xi )}d^{3}r^{\prime }\partial J(t,r^{\prime })/\partial t+ \\ 
\\ 
+\frac{1}{2c^{3}}\int_{\mathbb{\Omega }_{+}(\xi )}d^{3}r^{\prime
}|r-r^{\prime }|\partial ^{2}J(t,r^{\prime })/\partial t^{2}+O(1/c^{4}),%
\end{array}%
\end{equation*}%
where \textit{\ }the current density $J(t,r)=\rho (t,r)dr/dt$ for all $t\in 
\mathbb{R}$ and $r\in \Omega (\xi ):=\mathbb{\Omega }_{+}(\xi )\cup \mathbb{%
\Omega }_{+}(\xi )\simeq \mathbb{S}^{2}:=supp$ $\rho (t;r)\subset \mathbb{R}%
^{3},$ being the spherical compact support of the charged particle density
distribution, and the limit $\lim_{\varepsilon \downarrow 0}$ \textit{was \
treated physically, that is taking into account the assumed shell modell of
the charged particle }$\xi $\textit{\ and its corresponding charge density
self interaction}. \ Moreover, the potentials \ (\ref{3.18}) are both
considered to be retarded and non singular, moving in space with the
velocity $u\in T(\mathbb{R}^{3})$ subject to the laboratory reference frame $%
\mathcal{K}_{t}.$ As a result of simple enough calculations like in \cite%
{Jack}, making use of \ the expressions (\ref{3.18}) one obtains that the
Lagranfian function (\ref{3.17}) brings about 
\begin{equation}
\mathcal{L}_{f-p}^{(t)}=\frac{\mathcal{E}_{es}}{2c^{2}}|u|^{2}-<k(t),dr/dt>,
\label{3.18b}
\end{equation}%
where we took into account that owing to the reasonings from \cite{Beck,Moro}
the only front half the electric charge interacts with the whole virtually
identical charge charge $\xi ,$ \ as well as made use of the following up to 
$O(1/c^{4})\ $limiting integral expressions: 
\begin{equation}
\ 
\begin{array}{c}
\ \int_{\mathbb{\Omega }_{+}(\xi )\cup \mathbb{\Omega }_{-}(\xi
)}d^{3}r\int_{\mathbb{\Omega }_{+}(\xi )\cup \mathbb{\Omega }_{-}(\xi )}\
d^{3}r^{\prime }\rho (t,r^{\prime })\rho (t,r^{\prime }):=\xi ^{2}, \\ 
\\ 
\mathcal{\ }\frac{1}{2}\int_{\mathbb{\Omega }_{+}(\xi )\cup \mathbb{\Omega }%
_{-}(\xi )}d^{3}r\int_{\mathbb{\Omega }_{+}(\xi )\cup \mathbb{\Omega }%
_{-}(\xi )}\ d^{3}r^{\prime }\frac{\rho (t,r^{\prime })\rho (t,r^{\prime })}{%
|r-r^{\prime }|}:=\mathcal{E}_{es}, \\ 
\\ 
\ \ \int_{\mathbb{\Omega }_{+}(\xi )}d^{3}r\rho (t,r)\int_{\mathbb{\Omega }%
_{+}(\xi )}d^{3}r^{\prime }\frac{\rho (t;r^{\prime })}{|r^{\prime }-r|}\ =%
\frac{1}{2}\mathcal{E}_{es}, \\ 
\\ 
\ \int_{\mathbb{\Omega }_{-}(\xi )}d^{3}r\rho (t,r)\int_{\mathbb{\Omega }%
_{-}(\xi )}d^{3}r^{\prime }\frac{\rho (t;r^{\prime })}{|r^{\prime }-r|}\ =%
\frac{1}{2}\mathcal{E}_{es}, \\ 
\\ 
\ \int_{\mathbb{\Omega }_{-}(\xi )}d^{3}r\rho (t,r)\ \int_{\mathbb{\Omega }%
_{+}(\xi )}d^{3}r^{\prime }\frac{\rho (t;r^{\prime })}{|r-r^{\prime }|}|%
\frac{<r^{\prime }-r,u>}{|r^{\prime }-r|}|^{2}>:=\frac{\mathcal{E}_{es}}{6}%
|u|^{2}, \\ 
\\ 
\ \int_{\mathbb{\Omega }_{+}(\xi )}d^{3}r\rho (t,r)\ \int_{\mathbb{\Omega }%
_{+}(\xi )}d^{3}r^{\prime }\frac{\rho (t;r^{\prime })}{|r-r^{\prime }|}|%
\frac{<r^{\prime }-r,u>}{|r^{\prime }-r|}|^{2}>:=\frac{\mathcal{E}_{es}}{6}%
|u|^{2}.%
\end{array}
\label{3.19}
\end{equation}

To obtain the corresponding evolution equation for our charged particle $\xi 
$ \ we need, within the Feynman proper time paradigm, to transform the
Lagrangian function \ (\ref{3.18b}) to the one with respect to the proper
time reference frame $\mathcal{K}_{\tau }:$%
\begin{equation}
\mathcal{L}_{f-p}^{(\tau )}=\ (m_{es}/2)|\dot{r}|^{2}(1+|\dot{r}%
|^{2}/c^{2})^{-1/2}-<k(t),\dot{r}>,  \label{3.20}
\end{equation}%
where, for brevity, we have denoted by $\dot{r}$ $:=dr/d\tau $ the charged
particle velocity with respect to the proper reference frame $\mathcal{K}%
_{\tau }$ and by, definition, $m_{es}:=\mathcal{E}_{es}/c^{2}$ its so called
electrostatic mass with respect to the laboratory refrence frame $\mathcal{K}%
_{t}.$

Thus, the generalized charged particle $\xi $ momentum \ (up to $O(1/c^{4}))$
equals 
\begin{equation}
\begin{array}{c}
\pi _{p}:=\partial \mathcal{L}_{f-p}^{(\tau )}/\partial \dot{r}=\ \frac{%
m_{es}\dot{r}}{(1+|\dot{r}|^{2}/c^{2})^{1/2}}-\frac{m_{es}|\dot{r}|^{2}\dot{r%
}}{2c^{2}(1+|\dot{r}|^{2}/c^{2})^{3/2}}-k(t)= \\ 
\\ 
=\ m_{es}u(1-\frac{|u|^{2}}{2c^{2}})-k(t)\simeq
m_{es}u(1-|u|^{2}/c^{2})^{1/2}-k(t)=\bar{m}_{es}u-k(t),%
\end{array}
\label{3.21}
\end{equation}%
where we denoted, as before, by $u:=dr/dt$ the charged particle $\xi $
velocity with respect to the laboratory reference frame $\mathcal{K}_{t}\ $%
and put, by definition, $\ $%
\begin{equation}
\bar{m}_{es}:=m_{es}(1-|u|^{2})^{1/2}\ \   \label{3.22}
\end{equation}%
its mass parameter $\bar{m}_{es}\in \mathbb{R}_{+}$ with respect to the
proper reference frame $\mathcal{K}_{\tau }.$

The generalized momentum \ (\ref{3.22}) satisfies with respect to the proper
reference frame $\mathcal{K}_{\tau }$ the evolution equation 
\begin{equation}
d\pi _{p}/d\tau :=\partial \mathcal{L}_{f-p}^{(\tau )}/\partial r=0,
\label{3.23}
\end{equation}%
$\ $being \ equivalent, with respect to the laboratory reference frame $%
\mathcal{K}_{t},\ $ \ to the Lorentz type equation 
\begin{equation}
\frac{d}{dt}(\bar{m}_{es}u)=\ dk(t)/dt.  \label{3.23a}
\end{equation}%
The evolution equation \ (\ref{3.23a}) allows the corresponding canonical
Hamiltonian formulation on the phase space $T^{\ast }(\mathbb{R}^{3})$ with
the Hamiltonian function 
\begin{equation}
\begin{array}{c}
H_{f-p}^{\ }:=<\pi _{p},r>-\mathcal{L}_{f-p}^{(\tau )}\simeq <\frac{m_{es}%
\dot{r}}{(1+|\dot{r}|^{2}/c^{2})^{1/2}}-\frac{m_{es}|\dot{r}|^{2}\dot{r}}{%
2c^{2}(1+|\dot{r}|^{2}/c^{2})^{3/2}}-k(t),\dot{r}>- \\ 
\\ 
-(m_{es}/2)|\dot{r}|^{2}(1+|\dot{r}|^{2}/c^{2})^{-1/2}+<k(t),\dot{r}>=\bar{m}%
_{es}|u|^{2}/2,%
\end{array}
\label{3.23b}
\end{equation}%
naturally looking and satisfying up to $O(1/c^{4})$ for all $\tau $ and $t$ $%
\in $ $\mathbb{R}$ the conservation conditions 
\begin{equation}
\frac{d}{d\tau }H_{f-p}=0=\frac{d}{dt}H_{f-p}.\ \   \label{3.23c}
\end{equation}%
Looking at the equation \ (\ref{3.23a}) and \ (\ref{3.23b}), one can state
that the physically observable inertial charged particle $\xi $ mass
parameter 
\begin{equation}
m_{phys}:=\ \bar{m}_{es},  \label{3.23d}
\end{equation}%
being exactly equal to the relativistic charged particle $\xi $
electromagnetic mass, as it was assumed by H. Lorentz and Abraham.

To determine the damping radiation force $k(t)\in \mathbb{E}^{3},$ we can
make use of the Lorentz type force expression \ (\ref{3.17}) and obtain,
similarly to \cite{Jack}, up to $O(1/c^{4})$ accuracy, the resulting
self-interecting Abraham-Lorentz type force expression. Thus, owing to the
zero net foirce condition, we have that 
\begin{equation}
d\tilde{\pi}_{p}/dt+F_{s}=0,  \label{3.23e}
\end{equation}%
where the Lorentz force $\ $%
\begin{eqnarray}
\ F_{s} &=&-\ \frac{1}{2c}\int_{\mathbb{\Omega }_{-}(\xi )}d^{3}r\rho (t,r)%
\frac{d}{dt}A(t,r)-\frac{1}{2c}\int_{\mathbb{\Omega }_{+}(\xi )\cup \mathbb{%
\Omega }_{-}(\xi )}d^{3}r\rho (t,r)\frac{d}{dt}A(t,r)-  \label{3.24} \\
&&  \notag \\
&&-\frac{1}{2}\int_{\mathbb{\Omega }_{-}(\xi )}d^{3}r\rho (t,r)\nabla
\varphi (t,r)\ (1-|u/c|^{2})-\frac{1}{2}\int_{\mathbb{\Omega }_{+}(\xi )\cup 
\mathbb{\Omega }_{-}(\xi )}d^{3}r\rho (t,r)\nabla \varphi (t,r)\
(1-|u/c|^{2}).  \notag
\end{eqnarray}%
This expression easily follows from the least action condition $\delta
S^{(t)}=0,\ $\ where $\ S^{(t)}:=\int_{t_{1}}^{t_{2}}\mathcal{L}%
_{f-p}^{(t)}dt$ \ \ with the Lagrangian function given by the derived above
Landau-Lifschitz type expression \ (\ref{3.19}), and the potentials $%
(\varphi ,A)\in T^{\ast }(M^{4})$ given by the Lienard-Wiechert expressions
\ (\ref{3.18}). Followed by calculations similar to those of \cite{Jack,Bart}%
, \ from (\ref{3.24}) and (\ref{3.18}) one can obtain, within the assumed
before uniform shell electron model, for small $|u/c|\ll 1\ $ and slow
enough acceleration \ \ that%
\begin{equation}
\begin{array}{c}
F_{s}=\ \sum_{n\in \mathbb{Z}_{+}}\frac{(-1)^{n+1}}{2n!c^{n}}%
(1-|u/c|^{2})[\int_{\mathbb{\Omega }_{-}(\xi )}\rho (t,r)d^{3}r(\cdot )+ \\ 
\\ 
+\int_{\mathbb{\Omega }_{+}(\xi )\cup \mathbb{\Omega }_{-}(\xi )}\rho
(t,r)d^{3}r(\cdot )]\ \int_{\mathbb{\Omega }_{+}(\xi )}d^{3}r^{\prime }\frac{%
\partial ^{n}}{\partial t^{n}}\rho (t,r^{\prime })\nabla |r-r^{\prime
}|^{n-1\ }+ \\ 
\\ 
+\ \sum_{n\in \mathbb{Z}_{+}}\frac{(-1)^{n+1}}{2n!c^{n+2}}[\int_{\mathbb{%
\Omega }_{-}(\xi )}\rho (t,r)d^{3}r(\cdot )+ \\ 
\\ 
+\int_{\mathbb{\Omega }_{+}(\xi )\cup \mathbb{\Omega }_{-}(\xi )}\rho
(t,r)d^{3}r(\cdot )]\ \int_{\mathbb{\Omega }_{+}(\xi )}d^{3}r^{\prime
}|r-r^{\prime }|^{n-1}\frac{\partial ^{n+1}}{\partial t^{n+1}}J(t,r^{\prime
}]= \\ 
\\ 
=\ \sum_{n\in \mathbb{Z}_{+}}\frac{(-1)^{n+1}}{2n!c^{n+2}}%
(1-|u/c|^{2})[\int_{\mathbb{\Omega }_{-}(\xi )}\rho (t,r)d^{3}r(\cdot )+ \\ 
\\ 
+\int_{\mathbb{\Omega }_{+}(\xi )\cup \mathbb{\Omega }_{-}(\xi )}\rho
(t,r)d^{3}r(\cdot )]\int_{\mathbb{\Omega }_{+}(\xi )}d^{3}r^{\prime }\frac{%
\partial ^{n=2}}{\partial t^{n+2}}\rho (t,r^{\prime })\nabla |r-r^{\prime
}|^{n+1}+ \\ 
\\ 
+\sum_{n\in \mathbb{Z}_{+}}\frac{(-1)^{n+1}}{2n!c^{n+2}}[\int_{\mathbb{%
\Omega }_{-}(\xi )}\rho (t,r)d^{3}r(\cdot )+ \\ 
\\ 
+\int_{\mathbb{\Omega }_{+}(\xi )\cup \mathbb{\Omega }_{-}(\xi )}\rho
(t,r)d^{3}r(\cdot )]\int_{\mathbb{\Omega }_{+}(\xi )}d^{3}r^{\prime
}|r-r^{\prime }|^{n-1}\frac{\partial ^{n+1}}{\partial t^{n+1}}J(t,r^{\prime
}).%
\end{array}
\label{3.24a}
\end{equation}%
The relationship above can be rewritten, owing to the charge continuity
equation \ (\ref{3.1}), giving rise to the radiation force expression 
\begin{equation}
\begin{array}{c}
\\ 
F_{s}=\ \sum_{n\in \mathbb{Z}_{+}}\frac{(-1)^{n\ }}{2n!c^{n+2}}%
(1-|u/c|^{2})[\int_{\mathbb{\Omega }_{\mathbb{-}}}\rho (t,r)d^{3}r(\cdot
)+\int_{\Omega _{+}\cup \mathbb{\Omega }_{\mathbb{-}}}\rho (t,r)d^{3}r(\cdot
)]\times \\ 
\\ 
\times \int_{\mathbb{\Omega }_{\mathbb{+}}}d^{3}r^{\prime }|r-r^{\prime
}|^{n-1}\frac{\partial ^{n+1}}{\partial t^{n+1}}\left( \frac{\ J(t,r^{\prime
})}{n+2}+\frac{n-1}{n+2}\frac{<r-r^{\prime },J(t,r^{\prime })>(r-r^{\prime })%
}{|r-r^{\prime }|^{2}}\right) + \\ 
\\ 
+\ \sum_{n\in \mathbb{Z}_{+}}\frac{(-1)^{n+1}}{2n!c^{n+2}}[\int_{\mathbb{%
\Omega }_{\mathbb{-}}}\rho (t,r)d^{3}r(\cdot )+\int_{\Omega _{+}\cup \mathbb{%
\Omega }_{\mathbb{-}}}\rho (t,r)d^{3}r(\cdot )]\ \int_{\mathbb{\Omega }_{%
\mathbb{+}}}d^{3}r^{\prime }|r-r^{\prime }|^{n-1}\frac{\partial ^{n+1}}{%
\partial t^{n+1}}J(t,r^{\prime })= \\ 
\\ 
=\ \sum_{n\in \mathbb{Z}_{+}}\frac{(-1)^{n+1}}{2n!c^{n+2}}%
(1-|u/c|^{2})[\int_{\mathbb{\Omega }_{\mathbb{-}}}\rho (t,r)d^{3}r(\cdot
)+\int_{\Omega _{+}\cup \mathbb{\Omega }_{\mathbb{-}}}\rho (t,r)d^{3}r(\cdot
)]\ \times \\ 
\\ 
\times \int_{\mathbb{\Omega }_{\mathbb{+}}}d^{3}r^{\prime }|r-r^{\prime
}|^{n-1}\frac{\partial ^{n+1}}{\partial t^{n+1}}\left( \frac{\ J(t,r^{\prime
})}{n+2}+\frac{n-1}{n+2}\frac{|r-r^{\prime },u|^{2}J(t,r^{\prime })}{%
|r-r^{\prime }|^{2}|u|^{2}}\right) + \\ 
\\ 
+\ \sum_{n\in \mathbb{Z}_{+}}\frac{(-1)^{n+1}}{2n!c^{n+2}}[\int_{\mathbb{%
\Omega }_{\mathbb{-}}}\rho (t,r)d^{3}r(\cdot )+\int_{\Omega _{+}\cup \mathbb{%
\Omega }_{\mathbb{-}}}\rho (t,r)d^{3}r(\cdot )]\ \int_{\mathbb{\Omega }_{%
\mathbb{+}}}d^{3}r^{\prime }|r-r^{\prime }|^{n-1}\frac{\partial ^{n+1}}{%
\partial t^{n+1}}J(t,r^{\prime }). \\ 
\end{array}
\label{3.24b}
\end{equation}%
Now, having applied to \ (\ref{3.24b}) the rotational symmetry property for
calculation of the internal integrals, one easily obtains in the case of a
charged particle $\xi $ uniform shell model that 
\begin{equation*}
\begin{array}{c}
F_{s}=\ \sum_{n\in \mathbb{Z}_{+}}\frac{(-1)^{n}}{2n!c^{n+2}}%
(1-|u/c|^{2})[\int_{\mathbb{\Omega }_{\mathbb{-}}}\rho (t,r)d^{3}r(\cdot
)+\int_{\Omega _{+}\cup \mathbb{\Omega }_{\mathbb{-}}}\rho (t,r)d^{3}r(\cdot
)]\ \times \\ 
\\ 
\times \int_{\mathbb{\Omega }_{\mathbb{+}}}d^{3}r^{\prime }|r-r^{\prime
}|^{n-1}\frac{\partial ^{n+1}}{\partial t^{n+1}}\left( \frac{\ J(t,r^{\prime
})}{n+2}+\frac{(n-1)J(t,r^{\prime })}{3(n+2)}\ \right) + \\ 
\\ 
+\ \ \sum_{n\in \mathbb{Z}_{+}}\frac{(-1)^{n+1}}{2n!c^{n}}[\int_{\mathbb{%
\Omega }_{\mathbb{-}}}\rho (t,r)d^{3}r(\cdot )+\int_{\Omega _{+}\cup \mathbb{%
\Omega }_{\mathbb{-}}}\rho (t,r)d^{3}r(\cdot )]\ \int_{\mathbb{\Omega }_{%
\mathbb{+}}}d^{3}r^{\prime }\frac{|r-r^{\prime }|^{n+1}}{c^{2}}\frac{%
\partial ^{n+1}}{\partial t^{n+1}}J(t,r^{\prime })\ = \\ 
\end{array}%
\end{equation*}%
\begin{equation}
\begin{array}{c}
=\frac{d}{dt}[\sum_{n\in \mathbb{Z}_{+}}\frac{\ (-1)^{n+1}}{6n!c^{n+2}}%
[\int_{\mathbb{\Omega }_{\mathbb{-}}}\rho (t,r)d^{3}r(\cdot )+\int_{\Omega
_{+}\cup \mathbb{\Omega }_{\mathbb{-}}}\rho (t,r)d^{3}r(\cdot )]\ \times \\ 
\\ 
\times \int_{\mathbb{\Omega }_{\mathbb{+}}}d^{3}r^{\prime }|r-r^{\prime
}|^{n-1}\frac{\partial ^{n}}{\partial t^{n}}J(t,r^{\prime })-\sum_{n\in 
\mathbb{Z}_{+}}\frac{(-1)^{n\ }|u|^{2}}{6n!c^{n+4}})[\int_{\mathbb{\Omega }_{%
\mathbb{-}}}\rho (t,r)d^{3}r(\cdot )+ \\ 
\\ 
+\int_{\Omega _{+}\cup \mathbb{\Omega }_{\mathbb{-}}}\rho (t,r)d^{3}r(\cdot
)]\ \int_{\mathbb{\Omega }_{\mathbb{+}}}d^{3}r^{\prime }|r-r^{\prime }|^{n-1}%
\frac{\partial ^{n}}{\partial t^{n}}J(t,r^{\prime })].%
\end{array}
\label{3.24c}
\end{equation}%
Now, having took into account the integral expressions \ (\ref{3.19}), one
finds from \ (\ref{3.24c}) that up to the $O(1/c^{4})$ accuracy the
following radiation reaction force expression 
\begin{align}
\ F_{s}& =-\frac{d}{dt}\left( \frac{\mathcal{E}_{es}}{c^{2}}u\right) +\frac{d%
}{dt}\left( \frac{\mathcal{E}_{es}}{2c^{2}}\ |u/c|^{2}\ u(t)\right) +\frac{%
2\xi ^{2}}{3c^{3}}\frac{d^{2}u}{dt^{2}}+O(1/c^{4})=  \label{3.24d} \\
&  \notag \\
& =-\frac{d}{dt}\left( m_{es}(1-\frac{|u/c|^{2}}{2})u\right) +\frac{2\xi ^{2}%
}{3c^{3}}\frac{d^{2}u}{dt^{2}}+O(1/c^{4})=  \notag \\
&  \notag \\
& =-\frac{d}{dt}\left( m_{es}(1-|u/c|^{2})^{1/2}u\right) \ \ +\frac{2\xi ^{2}%
}{3c^{3}}\frac{d^{2}u}{dt^{2}}+O(1/c^{4})=  \notag \\
&  \notag \\
& =-\frac{d}{dt}(\bar{m}_{es}u-\frac{2\xi ^{2}}{3c^{3}}\frac{du}{dt}%
)+O(1/c^{4})\   \notag
\end{align}%
holds. \ We mention here that following the reasonings from \cite%
{Beck,Moro,Puth}, in the expressions above there is taken into account an
additional hidden and the velocity $u\in T(\mathbb{R}^{3})$ directed
electrostatic Coulomb surface self-force, acting only on the \textit{front
half part} of the spherical electron shell. As a result, from (\ref{3.23e}),
\ (\ref{3.24}) and the relationship \ (\ref{3.21}) one obtains that the
electron momentum 
\begin{equation}
\pi _{p}:=\bar{m}_{es}u-\frac{2\xi ^{2}}{3c^{3}}\frac{du}{dt}=\bar{m}%
_{es}u-k(t),  \label{3.24e}
\end{equation}%
thereby defyning both the radiation reaction momentum $k(t)=\frac{2\xi ^{2}}{%
3c^{3}}\frac{du}{dt}$ and the corresponding radiation reaction force 
\begin{equation}
F_{r}=\frac{2\xi ^{2}}{3c^{3}}\frac{d^{2}u}{dt^{2}}+O(1/c^{4}),  \label{3.25}
\end{equation}%
coincides exactly with the classical Abraham--Lorentz--Dirac expression.
Moreover, it also follows that the observable physical shell model electron
inertial mass 
\begin{equation}
m_{ph}=\ m_{es}:=\mathcal{E}_{es}/c^{2},  \label{3.25a}
\end{equation}%
$\ $\ being completely of the electromagnetic origin, giving rise to the
final force expression 
\begin{equation}
\frac{d}{dt}(m_{ph}u)=\frac{2\xi ^{2}}{3c^{3}}\frac{d^{2}u}{dt^{2}}%
+O(1/c^{4}).  \label{3.26}
\end{equation}%
This means, in particular, that the real physically observed
\textquotedblleft inertial\textquotedblright\ mass $m_{ph}$\ of an electron
within the uniform shell model is strongly determined by its electromagnetic
self-interaction energy $\mathcal{E}_{es}$. A similar statement there was
recently demonstrated using completely different approaches in \cite%
{Puth,Moro}, based on the vacuum Casimir effect considerations. Moreover,
the assumed above boundedness of the electrostatic self-energy $\mathcal{E}%
_{es}$ appears to be completely equivalent to the existence of so-called
intrinsic Poincar\'{e} type \textquotedblleft \textit{tensions}%
\textquotedblright , analyzed in \cite{Beck,Moro}, and to the existence of a
special compensating Coulomb \textquotedblleft \textit{pressure}%
\textquotedblright , suggested in \cite{Puth}, guaranteeing the observable
electron stability.

\begin{remark}
Some years ago there was suggested in the work \cite{MaPi} a "solution" to
the mentioned before $"4/3$-electron mass" problem, expressed by the
physical mass mass relationship \ (\ref{3.25a}) and formulated more than one
hundred years ago by H. Lorentz and M. Abraham. To the regret, the above
mentioned "solution" appeared to be fake that one can easily observe from
the main not correct assumptions \ on which the work \cite{MaPi} \ has been
based: the first one is about the particle-field momentum conservation,
taken in the form 
\begin{equation}
\frac{d}{dt}(p+\xi A)=0,  \label{3.27}
\end{equation}%
and the second one is a speculation about the $1/2$-coefficient imbedded
into the calculation of the Lorentz type self-interaction force 
\begin{equation}
F:=-\frac{1}{2c}\int_{\mathbb{R}^{3}}d^{3}r\rho (t;r)\partial
A(t;r)/\partial t,  \label{3.28}
\end{equation}%
being not correctly argued by the reasoning that the expression \ (\ref{3.28}%
) represents "... the interaction of a given element of charge with all
other parts, otherwise we count twice that reciprocal action" (cited from \ 
\cite{MaPi}, page 2710). This claim is fake as there was not taken into
account the important fact that the interaction in the integral \ (\ref{3.28}%
) is, in reality, retarded and its impact into it should be considered as
that calculated for two virtually different charged particles, as it has
been done in the classical works of H. Lorentz and M. Abraham. Subject to
the first assumption \ (\ref{3.27}) it is enough to recall that a vector of
the field momentum $\xi A\in \mathbb{E}^{3}$ is not independent and is,
within the charged particle model considered, strongly related with the
local flow of the electromagnetic potential energy in the Lorentz constraint
form:%
\begin{equation}
\partial \varphi /\partial t+<\nabla ,A\ >=0,  \label{3.29}
\end{equation}%
under which there hold \ the exploited in the work \cite{MaPi} the
Lienard-Wiechert expressions \ (\ref{3.17}) potentials\ for calculation of
the integral \ (\ref{3.28}). Thus, the equation \ (\ref{3.27}), following
the classical Newton second law, should be replaced by%
\begin{equation}
\frac{d}{dt^{\prime }}(p^{\prime }+\xi A^{\prime })=-\nabla (\xi \varphi
^{\prime }),  \label{3.30}
\end{equation}%
written with respect to the reference frame $\mathcal{K}(t^{\prime };r%
\mathcal{)}$ subject to which the charged particle $\xi \ $ is at rest.
Taking into account that with respect to the laboratory reference frame $%
\mathcal{K}_{t}$ there hold the relativistic relationships $dt=dt^{\prime
}(1-|u|^{2}/c^{2})^{1/2}\ $\ and $\varphi ^{\prime }=\varphi
(1-|u|^{2}/c^{2})^{1/2},$ \ \ from (\ref{3.30}) one easily obtains that 
\begin{equation}
\begin{array}{c}
\frac{d}{dt}(p+\xi A)=-\xi \nabla \varphi (1-|u|^{2}/c^{2})= \\ 
\\ 
=-\xi \nabla \varphi +\frac{\xi }{c}\nabla <u,u\varphi /c>=-\xi \nabla
\varphi +\frac{\xi }{c}\nabla <u,A>.%
\end{array}
\label{3.31}
\end{equation}%
Here we made use of the well-known relationship $A=u\varphi /c$ for the
vector potential generated by this charged particle $\xi $ moving in space
with the velocity $u\in T(\mathbb{R}^{3})$ with respect to the laboratory
reference frame $\mathcal{K}_{t}.$ Based now on the equation \ (\ref{3.31})
one can derive the final expression for the evolution of the charged
particle $\xi $ momentum:%
\begin{eqnarray}
dp/dt &=&-\xi \nabla \varphi -\frac{\xi }{c}dA/dt+\frac{\xi }{c}\nabla <u,A>=
\label{3.32} \\
&&  \notag \\
&=&-\xi \nabla \varphi -\frac{\xi }{c}\partial A/\partial t-\frac{\xi }{c}%
<u,\nabla >A+\frac{\xi }{c}\nabla <u,A>=  \notag \\
&&  \notag \\
&=&\xi E+\frac{\xi }{c}u\times (\nabla \times A)=\xi E+\frac{\xi }{c}u\times
B,  \notag
\end{eqnarray}%
that is exactly the well known Lorentz force expression, used in the works
of H. Lorentz and M. Abraham.
\end{remark}

Recently enough there appeared other interesting works devoted to this "$4/3$%
-electron mass" problem, amongst which we would like to mention \cite%
{Moro,Puth}, whose argumentations are close to each other and based on the
charged shell electron model, within which there is assumed a virtual
interaction of the electron with the ambient "dark" radiation energy. The
latter was first clearly demonstrated in \cite{Puth}, \ \ where a suitable
compensation mechanism of the related singular electrostatic Coulomb
electron energy and the wide band vacuum electromagnetic radiation energy
fluctuations deficit inside the electron shell was shown to be harmonically
realized as the electron shell radius $a\rightarrow 0.$ Moreover, this
compensation happens exactly when the induced outward directed electrostatic
Coulomb pressure on the whole electron coincides, up to the sign, with that
induced by the mentioned above vacuum electromagnetic energy fluctuations
outside the electron shell, since there was manifested their absence inside
the electron shell.

Really, the outward directed electrostatic spatial Coulomb pressure on the
electron equals 
\begin{equation}
\eta _{coul}:=\left. \lim_{a\rightarrow 0}\frac{\varepsilon _{0}|E|^{2}}{2}%
\right\vert _{r=a}=\lim_{a\rightarrow 0}\frac{\xi ^{2}}{32\varepsilon
_{0}\pi ^{2}a^{4}},  \label{3.33}
\end{equation}%
where $E=\frac{\xi r}{4\pi \varepsilon _{0}|r|^{3}}\in \mathbb{E}^{3}$ is
the electrostatic field at point $r\in \mathbb{R}\ $\ subject to the
electron center at the point $r=0\in \mathbb{R}.$ The related inward
directed vacuum electromagnetic fluctuations spatial pressure equals 
\begin{equation}
\eta _{vac}:=\lim_{\Omega \rightarrow \infty }\frac{1}{3}\int_{0}^{\Omega }d%
\mathcal{E}(\omega ),  \label{3.34}
\end{equation}%
where $d\mathcal{E}(\omega )$ is the electromagnetic energy fluctuations
density for a frequency $\omega \in \mathbb{R},$ and $\Omega \in \mathbb{R}$
is the corresponding electromagnetic frequency cutoff. The integral \ (\ref%
{3.34}) can be calculated if to take into account the quantum statistical
recipe \cite{Feyn,Huan,BoBo} that 
\begin{equation}
d\mathcal{E}(\omega ):=\hbar \omega \frac{d^{3}p(\omega )}{h^{3}},
\label{3.35}
\end{equation}%
where the Plank constant $h:=2\pi \hbar $ and the electromagnetic wave
momentum $p(\omega )\in \mathbb{E}^{3}$ satisfies the relativistic
relationship 
\begin{equation}
|p(\omega )|=\hbar \omega /c.  \label{3.36}
\end{equation}%
Whence by substituting \ (\ref{3.36}) into \ (\ref{3.35}) one obtains 
\begin{equation}
d\mathcal{E}(\omega )=\frac{\hbar \omega ^{3}}{2\pi ^{2}c^{3}}d\omega ,
\label{3.37}
\end{equation}%
which entails, owing to \ (\ref{3.34}), the following vacuum electromagnetic
energy fluctuations spatial pressure%
\begin{equation}
\eta _{vac}=\lim_{\Omega \rightarrow \infty }\frac{\hbar \Omega ^{4}}{24\pi
^{2}c^{3}}.  \label{3.38}
\end{equation}

For the charged electron shell model to be stable at rest it is necessary to
equate the inward \ (\ref{3.38}) and outward \ (\ref{3.33}) spatial
pressures:%
\begin{equation}
\lim_{\Omega \rightarrow \infty }\frac{\hbar \Omega ^{4}}{24\pi ^{2}c^{3}}%
=\lim_{a\rightarrow 0}\frac{\xi ^{2}}{32\varepsilon _{0}\pi ^{2}a^{4}},
\label{3.39}
\end{equation}%
giving rise to the balance electron shell radius $a_{b}\rightarrow 0$
limiting condition:%
\begin{equation}
a_{b}=\lim_{\Omega \rightarrow \infty }\ \left[ \Omega ^{-1}\left( \frac{%
3\xi ^{2}c^{2}}{2\hbar }\right) ^{1/4}\right] .  \label{3.40}
\end{equation}

Simultaneously we can calculate the corresponding Coulomb and
electromagnetic fluctuations energy deficit inside the electron shell:%
\begin{equation}
\Delta W_{b}:=\frac{1}{2}\int_{a_{b}}^{\infty }\varepsilon
_{0}|E|^{2}d^{3}r-\int_{0}^{a_{b}}d^{3}r\int_{0}^{\Omega }d\mathcal{E}%
(\omega )=\frac{\xi ^{2}}{8\pi \varepsilon _{0}a_{b}}-\frac{\hbar \Omega
^{4}a_{b}^{3}}{6\pi c^{3}}=0,\   \label{3.41}
\end{equation}%
additionally ensuring the electron shell model stability.

Another important consequence from this pressure-energy compensation
mechanism can be derived concerning the electron ienrtial mass $m_{ph}\in 
\mathbb{R}_{+},$ entering the momentum expression \ (\ref{3.24e}) in the
case of the electron slow enough movement. Namely, following the reasonings
from \cite{Moro}, one can observe that during the electton movement there
arises an additional hidden not compensated and velocity $u\in T(\mathbb{R}%
^{3})$ directed electrostatic Coulomb surface self-pressure acting only on
the \textit{front half part} of the electron shell and equal to 
\begin{equation}
\eta _{surf}:=\frac{|E\xi |}{4\pi a_{b}^{2}}\frac{1}{2}=\frac{\xi ^{2}}{%
32\pi \varepsilon _{0}a_{b}^{4}},  \label{3.42}
\end{equation}%
coinciding, evidently, with the already compensated outward directed
electrostatic Coulomb spatial pressure \ (\ref{3.33}). As, evidently, during
the electron motion in space its surface electric current energy flow is not
vanishing \cite{Moro}, one ensues that the electron momentum gains an
additional mechanical impact, which can be expressed as 
\begin{equation}
\pi _{\xi }:=-\eta _{surf}\frac{4\pi a_{b}^{3}}{3c^{2}}u=-\frac{1}{3}\frac{%
\xi ^{2}}{8\pi \varepsilon _{0}a_{b}c^{2}}u=-\frac{1}{3}\bar{m}_{es}u,
\label{3.43}
\end{equation}%
where we took into account that within this electron shell model the
corresponding electrostatic electron mass equals its electrostatic energy 
\begin{equation}
\bar{m}_{es}=\frac{\xi ^{2}}{8\pi \varepsilon _{0}a_{b}c^{2}}.  \label{3.44}
\end{equation}

Thus, one can claim that, owing to the structural stability of the electron
shell model, its generalized self-interaction momentum $\pi _{p}\in T^{\ast
}(\mathbb{R}^{3})$ \ gains during the movement with velocity $u=dr/dt\in T(%
\mathbb{R}^{3})$ the additional backward directed hidden impact \ (\ref{3.43}%
), which can be identified with the back-directed momentum component 
\begin{equation}
\pi _{\xi }=-\frac{1}{3}\bar{m}_{es}u,  \label{3.45}
\end{equation}%
complementing the classical \cite{Jack,Bart} momentum expression $\ $%
\begin{equation}
\pi _{p}=\frac{4}{3}\bar{m}_{es}u,  \label{3.46}
\end{equation}
which can be easily obtained from the Lagrangian expression expression, if
one not to take into account the shading property of the moving uniform
shell electron model. \ Then, owing to \ the additional momentum (\ref{3.45}%
), the full momentum becomes as 
\begin{equation}
\pi _{p}=\pi _{\xi }+\frac{4}{3}\bar{m}_{es}u\ =(-\frac{1}{3}\bar{m}_{es}+%
\frac{4}{3}\bar{m}_{es})u=\bar{m}_{es}u,  \label{3.47}
\end{equation}%
coinciding with that of \ (\ref{3.21}) modulo the radiation reaction\
momentum $k(t)=\frac{2\xi ^{2}}{3c^{3}}\frac{du}{dt},$ \ strongly supporting
the electromagnetic energy origin of the electron inertion mass for the
first time conceived by H. Lorentz and M. Abraham.

\subsection{Comments}

The electromagnetic mass origin problem was reanalyzed in details within the
Feynman proper time paradigm and related vacuum field theory approach by
means of the fundamental least action principle and the Lagrangian and
Hamiltonian formalisms. The resulting electron inertia appeared to coincide
in part, in the quasi-relativistic limit, with the momentum expression
obtained more than one hundred years ago by M. Abraham and H. Lorentz \cite%
{Abra,Lore-1,Lore-2,Lore-3}, yet it proved to contain an additional hidden
impact owing to the imposed electron stability constraint, which was taken
into account in the original action functional as some preliminarily
undetermined constant component. As it was demonstrated in \cite{Puth,Moro},
this stability constraint can be successfully realized within the charged
shell model of electron at rest, if to take into account the existing
ambient electromagnetic \textquotedblleft dark\textquotedblright\ energy
fluctuations, whose inward directed spatial pressure on the electron shell
is compensated by the related outward directed electrostatic Coulomb spatial
pressure as the electron shell radius satisfies some limiting compatibility
condition. The latter also allows to compensate simultaneously the
corresponding electromagnetic energy fluctuations deficit \ inside the
electron shell, thereby forbidding the external energy to flow into the
electron. In contrary to the lack of energy flow inside the electron shell,
during the electron movement the corresponding internal momentum flow is not
vanishing owing to the nonvanishing hidden electron momentum flow caused by
the surface pressure flow and compensated by the suitably generated surface
electric current flow. As it was shown, \ this backward directed hidden
momentum flow makes it possible to justify the corresponding
self-interaction electron mass expression and to state, within the electron
shell model, the fully electromagnetic electron mass origin, as it has been
conceived by H. Lorentz and M. Abraham and strongly supported by R. Feynman
in his Lectures \cite{Feyn-1}. This consequence is also independently
supported by means of the least action approach, based on the Feynman proper
time paradigm and the suitably calculated regularized retarded \ electric
potential impact into the charged particle Lagrangian function.

The charged particle radiation problem, revisited in this Section, allowed
to conceive the explanation of the charged particle mass as that of a
compact and stable object which should be exerted by a vacuum field
self-interaction energy. The latter can be satisfied iff the expressions (%
\ref{3.19}) hold, thereby imposing on the intrinsic charged particle
structure \cite{Medi} some nontrivial geometrical constraints. Moreover, as
follows from the physically observed particle mass expressions (\ref{3.25a}%
), the electrostatic potential energy being of the \ \ \ self-interaction
origin, contributes into the inertial mass as its main relativistic mass
component.

There exist different relativistic generalizations of the force expression \
(\ref{3.26}), which \ suffer the common physical inconsistency related to
the no radiation effect of \ a charged particle in uniform motion.

Another deeply related \ problem to the radiation reaction force analyzed
above is the search for an explanation to the Wheeler and Feynman reaction
radiation mechanism, called the absorption radiation theory, strongly based
on the Mach type interaction of a charged particle with the ambient vacuum
electromagnetic medium. Concerning this problem, one can also observe some
of its relationships with the one devised here within the vacuum field
theory approach, but this question needs a more detailed and extended
analysis.

\section{\protect\bigskip}

\section{\label{Sec_5}A charged point particle dynamics and a hadronic
string model analysis}

\subsection{The classical relativistic electrodynamics backgrounds: a
charged point particle analysis}

It is well known \cite{LaLi,Feyn-1,Paul,BaNe} that the relativistic least
action principle for a point charged particle $\xi $ in the Minkowski
space-time $M^{4}\simeq \mathbb{E}^{3}\times \mathbb{R}$ can be formulated
on a time interval $[t_{1},t_{2}]\subset \mathbb{R}$ (in the light speed
units) as 
\begin{eqnarray}
\delta S^{(t)} &=&0,\text{ \ \ }S^{(t)}:=\int_{\tau (t_{1})}^{\tau
(t_{2})}(-m_{0}d\tau -\xi <\mathcal{A},dx>_{M^{4}})=  \notag \\
&=&\int_{s(t_{1})}^{s(t_{2})}(-m_{0}<\dot{x},\dot{x}>_{M^{4}}^{1/2}-\xi <%
\mathcal{A},\dot{x}>_{M^{4}})ds.  \label{A1.1}
\end{eqnarray}%
Here $\delta x(s(t_{1}))=0=\delta x(s(t_{2}))$ are the boundary constraints, 
$m_{0}\in \mathbb{R}_{+}$ is the so called particle rest mass, the 4-vector $%
x:=(r,t)\in M^{4}$ is the particle location in $M^{4},$ $\dot{x}:=dx/ds\in
T(M^{4})$ is the particle Euclidean \textquotedblright
four-vector\textquotedblright\ velocity with respect to a laboratory
reference frame $\mathcal{K},$ parameterized by means of the Minkowski
space-time parameters $(r,s(t))\in M^{4}$ and related to each other by means
of the infinitesimal Lorentz interval relationship 
\begin{equation}
d\tau :=<dx,dx>_{M^{4}}^{1/2}:=ds<\dot{x},\dot{x}>_{M^{4}}^{1/2},
\label{A1.1a}
\end{equation}%
$\mathcal{A}\in $ $T^{\ast }(M^{4})$ is an external electromagnetic 4-vector
potential, satisfying the classical Maxwell equations \cite{Paul,LaLi,Feyn-1}%
, the sign $<\cdot ,\cdot >_{\mathcal{H}}$ \ means, in general, the
corresponding scalar product in a finite-dimensional vector space $\mathcal{H%
}$ and $T(M^{4}),T^{\ast }(M^{4})$ \ are, respectively, the tangent and
cotangent spaces \cite{AbMa,Arno,Thir,DuNoFo,HePrPr} to the Minkowski space $%
M^{4}.$ In particular, $<x,x>_{M^{4}}:=t^{2}-<r,r>_{\mathbb{E}^{3}}$ for any 
$x:=(r,t)\in M^{4}.$

The subintegral expression in (\ref{A1.1}) 
\begin{equation}
\mathcal{L}^{(t)}:=-m_{0}<\dot{x},\dot{x}>_{M^{4}}^{1/2}-\xi <\mathcal{A},%
\dot{x}>_{M^{4}}  \label{A1.2}
\end{equation}%
is the related Lagrangian function, whose first part is proportional to the
particle world line length with respect to the proper rest reference frame $%
\mathcal{K}_{\tau }$ and the second part is proportional to the pure
electromagnetic particle-field interaction with respect to the Minkowski
laboratory reference frame $\mathcal{K}.$ Moreover, the positive rest mass
parameter \ $m_{0}\in \mathbb{R}_{+}$ is introduced into (\ref{A1.2}) as an
external physical ingredient, also describing the point particle\ \ with
respect to the proper rest reference frame $\mathcal{K}_{\tau }.$ \ The
electromagnetic 4-vector potential $\mathcal{A}\in T^{\ast }(M^{4})$ $\ $is
at the same time expressed as a 4-vector, constructed and measured with
respect to the Minkowski laboratory reference frame $\mathcal{K}_{t}$ \ that
looks from physical point of view enough controversial, since the action
functional (\ref{A1.1}) is forced to be extremal with respect to the
laboratory reference frame $\mathcal{K}.$ This, in particular, means that
the real physical motion of our charged point particle, being realized with
respect to the proper rest reference frame $\mathcal{K}_{\tau },$ somehow
depends on an external observation data \cite%
{Feyn-1,Fadd-4,Logu-2,Logu-3,Bril} with respect to the occasionally chosen
laboratory reference frame $\mathcal{K}.$ \ This aspect was never discussed
in the physical literature except of very interesting reasonings by R.
Feynman in \cite{Feyn-1}, who argued that the relativistic expression for
the classical Lorentz force has a physical sense only with respect to \ the
Euclidean rest \ reference frame $\mathcal{K}_{\tau }$ variables $(r,\tau
)\in \mathbb{E}^{4}$ related with the Minkowski laboratory reference frame $%
\mathcal{K}_{t}$ \ parameters $(r,t)\in M^{4}$ by means of the infinitesimal
relationship 
\begin{equation}
d\tau :=<dx,dx>_{M^{4}}^{1/2}=dt(1-|u|^{2})^{1/2},  \label{A1.2a}
\end{equation}%
where $u:=dr/dt\in T(\mathbb{E}^{3})$ is the point particle velocity with
respect to the reference frame $\mathcal{K}.$

It is worth to point out here that to be correct, it would be necessary to
include still into the action functional the additional part describing the
electromagnetic field itself. But this part is not taken into account, since
there is generally assumed \cite{BoSt,Klym,Klei,BoSh,Weyl,Merm,Merm-1,Neum}
that the charged particle influence on the electromagnetic field is
negligible. This is true, if the particle charge value $\ \xi $ is very
small but the support $supp\mathcal{A\subset }$ $M^{4}$ of the
electromagnetic 4-vector potential is compact. It is clear that in the case
of two interacting to each other charged particles the above assumption can
not be applied, as it is necessary to take into account the relative motion
of two particles and the respectively changing interaction energy. This
aspect of \ the action functional choice problem appears to be very
important when one tries to analyze the related Lorentz type \ forces
exerted by charged particles on each other. We will return to this problem
in a separate section below.

Having calculated the least action condition (\ref{A1.1}), we easily obtain
from (\ref{A1.2}) the classical relativistic dynamical equations 
\begin{eqnarray}
dP/ds &:&=-\partial \mathcal{L}^{(t)}/\partial x=-\xi \nabla _{x}<\mathcal{A}%
,\dot{x}>_{M^{4}},  \label{A1.3} \\
P &:&=-\partial \mathcal{L}^{(t)}/\partial \dot{x}=m_{0}\dot{x}<\dot{x},\dot{%
x}>_{M^{4}}^{-1/2}+\xi \mathcal{A},  \notag
\end{eqnarray}%
where by $P\in T^{\ast }(M^{4})$ we denoted the common particle-field
momentum of the interacting system.

Now at $s=t\in \mathbb{R}$ by means of the standard infinitesimal change of
variables (\ref{A1.2a}) \ we can easily obtain from (\ref{A1.3}) the
classical Lorentz force expression%
\begin{equation}
dp/dt=\xi E+\xi u\times B  \label{A1.6}
\end{equation}%
with the relativistic particle momentum and mass 
\begin{equation}
p:=mu,\text{ \ }\ m:=m_{0}(1-|u|^{2})^{-1/2},  \label{A1.7}
\end{equation}%
respectively, the electric field 
\begin{equation}
E:=-\partial A/\partial t-\nabla \varphi  \label{A1.8}
\end{equation}%
and the magnetic field 
\begin{equation}
B:=\nabla \times A,  \label{A1.9}
\end{equation}%
where we have expressed the electromagnetic 4-vector potential as $\mathcal{A%
}:=(A,\varphi )\in T^{\ast }(M^{4}).$

The Lorentz force (\ref{A1.6}), owing to our preceding assumption, means the
force exerted by the external electromagnetic field on our charged point
particle, whose charge $\xi $ is so negligible that it does not exert the
influence on the field. This fact becomes very important if we try to make
use of the Lorentz force expression (\ref{A1.6}) for the case of two
interacting to each other charged particles, since then one can not assume
that our charge $\xi $ exerts negligible influence on other charged
particle. Thus, the corresponding Lorentz force between two charged
particles should be suitably altered. Nonetheless, the modern physics \cite%
{BoSh,Dira,LaLi,DuJa,JaPo,Barr,Jack} did not make this naturally needed \
Lorentz force modification and there is everywhere used the classical
expression (\ref{A1.6}). This situation was observed and analyzed concerning
the related physical aspects in \cite{Repc}, having shown that the
electromagnetic Lorentz force between two moving charged particles can be
modified in such a way that it ceases to be dependent on their relative
motion contrary to the classical relativistic case.

To the regret, the least action principle approach to analyzing the Lorentz
force structure was in \cite{Repc} completely ignored that gave rise to some
incorrect and physically not motivated statements concerning mathematical
physics backgrounds of the modern electrodynamics. To make the problem more
transparent we will analyze it in the section below from the vacuum field
theory approach recently devised in \cite{BoPrTa,BoPrTaPr-Lore,BoPr-1}.

\subsection{The least action principle analysis}

Consider the least action principle (\ref{A1.1}) and observe that the
extremality condition 
\begin{equation}
\delta S^{(t)}=0,\text{ \ \ \ \ }\delta x(s(t_{1}))=0=\delta x(s(t_{2})),
\label{A2.1}
\end{equation}%
is calculated with respect to the laboratory reference frame $\mathcal{K},$
whose point particle coordinates $(r,t)\in M^{4}$ are parameterized by means
\ of an arbitrary parameter $s\in \mathbb{R}$ owing to expression (\ref%
{A1.1a}). Recalling now the definition of the invariant proper rest
reference frame $\mathcal{K}_{\tau }$ time parameter (\ref{A1.2a}), we
obtain that at the critical parameter value $s=\tau \in \mathbb{R}$ the
action functional (\ref{A1.1}) on the fixed interval $[\tau _{1},\tau
_{2}]\subset \mathbb{R}$ turns into 
\begin{equation}
S^{(t)}=\int\limits_{\tau _{1}}^{\tau _{2}}(-m_{0}-\xi <\mathcal{A},\dot{x}%
>_{M^{4}})d\tau  \label{A2.2}
\end{equation}%
under the additional constraint 
\begin{equation}
<\dot{x},\dot{x}>_{M^{4}}^{1/2}=1,  \label{A2.2a}
\end{equation}%
where, by definition, $\ \dot{x}:=dx/d\tau ,$ $\tau \in \mathbb{R}.$

The expressions (\ref{A2.2}) and (\ref{A2.2a}) need some comments since the
corresponding to (\ref{A2.2}) Lagrangian function 
\begin{equation}
\mathcal{L}^{(t)}:=-m_{0}-\xi <\mathcal{A},\dot{x}>_{M^{4}}  \label{A2.3}
\end{equation}%
depends only virtually on the unobservable rest mass parameter $m_{0}\in 
\mathbb{R}$ and, evidently, it has no \ direct impact into the resulting
particle dynamical equations following from the condition $\delta S^{(t)}=0.$
\ Nonetheless, the rest mass springs up as a suitable Lagrangian multiplier
owing to the imposed constraint (\ref{A2.2a}). To demonstrate this consider
the extended Lagrangian function (\ref{A2.3}) in the form 
\begin{equation}
\mathcal{L}_{\lambda }^{(t)}:=-m_{0}-\xi <\mathcal{A},\dot{x}%
>_{M^{4}}-\lambda (<\dot{x},\dot{x}>_{M^{4}}^{1/2}-1),  \label{A2.3.1}
\end{equation}%
where $\lambda \in \mathbb{R}$ is a suitable Lagrangian multiplier. The
resulting Euler equations look as%
\begin{eqnarray}
P_{r} &:&=\partial \mathcal{L}_{\lambda }^{(t)}/\partial \dot{r}=\xi
A+\lambda \dot{r},\text{ \ }P_{t}:=\partial \mathcal{L}_{\lambda
}^{(t)}/\partial \dot{t}=-\xi \varphi -\lambda \dot{t},\text{ }  \notag \\
\text{\ }\partial \mathcal{L}_{\lambda }^{(t)}/\partial \lambda &=&<\dot{x},%
\dot{x}>_{M^{4}}^{1/2}-1=0,\text{ \ }dP_{r}/d\tau =\xi \nabla _{r}<A,\dot{r}%
>_{\mathbb{E}^{3}}-\xi \dot{t}\nabla _{r}\varphi ,  \notag \\
\text{ \ }dP_{t}/d\tau &=&\xi <\partial A/\partial t,\dot{r}>_{\mathbb{E}%
^{3}}-\xi \dot{t}\partial \varphi /\partial t,  \label{A2.3.2}
\end{eqnarray}%
giving rise, owing to relationship (\ref{A1.2a}), to the following dynamical
equations:

\begin{equation}
\frac{d}{dt}(\lambda u\dot{t})=\xi E+\xi u\times B,\text{ \ }\frac{d}{dt}%
(\lambda \dot{t})=\xi <E,u>_{\mathbb{E}^{3}},  \label{A2.3.3}
\end{equation}%
where we denoted by 
\begin{equation}
E:=-\partial A/\partial t-\nabla \varphi ,\text{ \ }B=\nabla \times A
\label{A2.3.4}
\end{equation}%
the corresponding electric and magnetic fields. As a simple consequence of (%
\ref{A2.3.3}) one obtains 
\begin{equation}
\frac{d}{dt}\ln (\lambda \dot{t})+\frac{d}{dt}\ln (1-|u|^{2})^{1/2}=0,
\label{A2.3.5}
\end{equation}%
being equivalent for all $t\in \mathbb{R},$ owing to relationship (\ref%
{A1.2a}), \bigskip to the relationship 
\begin{equation}
\lambda \dot{t}(1-|u|^{2})^{1/2}=\lambda :=m_{0,}  \label{A2.3.6}
\end{equation}%
where $m_{0}\in \mathbb{R}_{+}$ is a constant, which could be interpreted as
the rest mass of our charged point particle $\xi .$ Really, the first
equation of (\ref{A2.3.3}) can be rewritten as 
\begin{equation}
dp/dt=\xi E+\xi u\times B,  \label{A2.3.7}
\end{equation}%
where we denoted 
\begin{equation}
p:=mu,\text{ }m:=\lambda \dot{t}=m_{0}(1-|u|^{2})^{-1/2},  \label{A2.3.8}
\end{equation}%
coinciding exactly with that of (\ref{A1.2a}).

Thereby, we retrieved here all of the results obtained in section above,
making use of the action functional (\ref{A2.2}), represented with respect
to the rest reference frame $\mathcal{K}_{\tau }$ under constraint \ (\ref%
{A2.2a}). \ During these derivations, we faced with a very delicate
inconsistency property of definition \ of the action functional $S^{(t)},$
defined with respect to the rest reference frame $\mathcal{K}_{\tau },$ but
depending on the external electromagnetic potential function $\mathcal{A}%
:M^{4}\rightarrow T^{\ast }(M^{4}),$ constructed exceptionally with respect
to the laboratory reference frame $\mathcal{K}.$ Namely, this potential
function, as a physical observable quantity, is defined and, \ \
respectively, measurable only with respect to the fixed laboratory reference
frame $\mathcal{K}.$ This, in particular, means that a physically reasonable
action functional should be constructed by means of an expression strongly
calculated within the laboratory reference frame $\mathcal{K}_{t}$ \ by
means of coordinates $(r,t)\in M^{4}$ and later suitably transformed subject
to the rest reference frame $\mathcal{K}_{\tau }$ coordinates $(r,\tau )\in 
\mathbb{E}^{4},$ respective for the real charged point particle $\xi $
motion. Thus, the corresponding action functional, in reality, should be
from the very beginning written as 
\begin{equation}
S^{(\tau )}=\int\limits_{t(\tau _{1})}^{t(\tau _{2})}(-\xi <\mathcal{A},\dot{%
x}>_{\mathbb{E}^{3}})dt,  \label{A2.4}
\end{equation}%
where $\dot{x}:=dx/dt,$ $t\in \mathbb{R},$ being calculated on some time
interval $[t(\tau _{1}),t(\tau _{2})]\subset \mathbb{R},$ suitably related
with the proper motion of the charged point particle $\xi $ on the true time
interval $[\tau _{1},\tau _{2}]\subset \mathbb{R}$ with respect to the rest
reference frame $\mathcal{K}_{\tau }$ and whose charge value is assumed so
negligible that it exerts no influence on the external electromagnetic
field. The problem now arises: how to compute correctly the variation $%
\delta S^{(\tau )}=0$ of the action functional (\ref{A2.4})?

To reply to this question we will turn to the Feynman reasonings from \cite%
{Feyn-1}, where he argued, when deriving the relativistic Lorentz force
expression, that the real charged particle dynamics can be physically not
ambiguously determined only with respect to the rest reference frame time
parameter. Namely, \ Feynman wrote: "...we calculate a growth $\Delta x$ for
a small time interval $\Delta t.$ But in the other reference frame the
interval $\Delta t$ may correspond to changing both $t^{\prime }$ and $%
x^{\prime },$ thereby at the change of the only $t^{\prime }$ the suitable
change of $x$ will be\ other... Making use of the quantity $d\tau $ one can
determine a good differential operator \ $d/d\tau ,$ as it is invariant with
respect to the Lorentz reference frames transformations". This means that if
our charged particle $\xi $ moves in the Minkowski space $M^{4}$ during the
time interval $[t_{1},t_{2}]\subset \mathbb{R}$ with respect to the
laboratory reference frame $\mathcal{K},$ its proper real and invariant time
of motion with respect to the rest reference frame $\mathcal{K}_{\tau }$
will be respectively $[\tau _{1},\tau _{2}]\subset \mathbb{R}.$

As a corollary of the Feynman reasonings, we arrive at the necessity to
rewrite the action functional (\ref{A2.4}) as%
\begin{equation}
S^{(\tau )}=\int\limits_{\tau _{1}}^{\tau _{2}}(-\xi <\mathcal{A},\dot{x}%
>_{M^{4}})d\tau ,\text{ \ }\delta x(\tau _{1})=0=\delta x(\tau _{2}),
\label{A2.5}
\end{equation}%
where $\dot{x}:=dx/d\tau ,$ $\tau \in \mathbb{R},$ under the additional
constraint 
\begin{equation}
<\dot{x},\dot{x}>_{M^{4}}^{1/2}=1,  \label{A2.5a}
\end{equation}%
being equivalent to the infinitesimal transformation (\ref{A1.2a}).
Simultaneously the proper time interval $[\tau _{1},\tau _{2}]\subset 
\mathbb{R}$ \ is mapped on the time interval $[t_{1},t_{2}]\subset \mathbb{R}
$ by means of the infinitesimal transformation 
\begin{equation}
dt=d\tau (1+|\dot{r}|^{2})^{1/2},  \label{A2.6}
\end{equation}%
where $\dot{r}:=dr/d\tau ,$ $\tau \in $\ $\mathbb{R}.$ Thus, we can now pose
the true least action problem equivalent to (\ref{A2.5}) as 
\begin{equation}
\delta S^{(\tau )}=0,\text{ \ \ \ \ \ }\delta r(\tau _{1})=0=\delta r(\tau
_{2}),  \label{A2.6a}
\end{equation}%
where the functional 
\begin{equation}
S^{(\tau )}=\int\limits_{\tau _{1}}^{\tau _{2}}[-\bar{W}(1+|\dot{r}%
|^{2})^{1/2}+\xi <A,\dot{r}>_{\mathbb{E}^{3}}]d\tau  \label{A2.7}
\end{equation}%
is characterized by the Lagrangian function 
\begin{equation}
\mathcal{L}^{(\tau )}:=-\bar{W}(1+|\dot{r}|^{2})^{1/2}+\xi <A,\dot{r}>_{%
\mathbb{E}^{3}}.  \label{A2.8}
\end{equation}%
Here we denoted, for further convenience, $\bar{W}:=\xi \varphi ,$ \ being a
suitable vacuum field \cite{BoPrTa,BoPrTaPr-Lore,BoPr-1,Repc} potential
function. The resulting Euler equation gives rise to the following
relationships%
\begin{eqnarray}
P &:&=\partial \mathcal{L}^{(\tau )}/\partial \dot{r}=-\bar{W}\dot{r}(1+|%
\dot{r}|^{2})^{-1/2}+\xi A,  \label{A2.9} \\
dP/d\tau &:&=\partial \mathcal{L}^{(\tau )}/\partial r=-\nabla \bar{W}(1+|%
\dot{r}|^{2})^{1/2}+\xi \nabla <A,\dot{r}>_{\mathbb{E}^{3}}.  \notag
\end{eqnarray}%
Making now use once more of the infinitesimal transformation (\ref{A2.6})
and the crucial dynamical particle mass definition \cite%
{BoPrTa,BoPrTaPr-Lore,Repc} (in the light speed units) 
\begin{equation}
m:=-\bar{W},  \label{A2.9a}
\end{equation}%
we can easily rewrite equations (\ref{A2.9}) with respect to the parameter $%
t\in \mathbb{R}$ as the \ classical relativistic Lorentz force:%
\begin{equation}
dp/dt=\xi E+\xi u\times B,  \label{A2.10}
\end{equation}%
where we denoted 
\begin{eqnarray}
p &:&=mu,\text{ \ \ \ \ \ \ }u:=dr/dt,\text{ }  \label{A2.11} \\
B &:&=\nabla \times A,\text{ \ \ }E:=-\xi ^{-1}\nabla \bar{W}-\partial
A/\partial t.  \notag
\end{eqnarray}%
Thus, we obtained once more the relativistic Lorentz force expression (\ref%
{A2.10}), but slightly different from (\ref{A1.6}), since the classical
relativistic momentum $\ $expression of (\ref{A1.7}) \ does not completely
coincide with our modified relativistic momentum expression 
\begin{equation}
p=-\bar{W}u,  \label{A2.12}
\end{equation}%
depending strongly on the scalar vacuum field potential function $\bar{W}%
:M^{4}\rightarrow \mathbb{R}.$ But if to recall here that our action
functional (\ref{A2.5}) was written under the assumption that \ the particle
charge value $\xi $ is negligible and \ not exerting the essential influence
on the electromagnetic field source, we can make use of the before obtained
in \cite{BoPr-1,BoPrTa,Repc} result, that the vacuum field potential
function $\bar{W}:M^{4}\rightarrow \mathbb{R},$ owing to (\ref{A2.10})-(\ref%
{A2.12}), satisfies as $\xi \rightarrow 0$ the dynamical equation 
\begin{equation}
d(-\bar{W}u)/dt=-\nabla \bar{W},  \label{A2.13}
\end{equation}%
whose solution will be exactly the expression%
\begin{equation}
-\bar{W}=m_{0}(1-|u|^{2})^{-1/2},\text{ }m_{0}=-\left. \bar{W}\right\vert
_{u=0}.  \label{A2.14}
\end{equation}%
Thereby, we have arrived, owing to (\ref{A2.14}) and (\ref{A2.12}), at the
almost full coincidence of our result (\ref{A2.10}) for the relativistic \
Lorentz force with that of (\ref{A1.6}) under the condition $\xi \rightarrow
0.$

The obtained above results and inferences we will formulate as the following
proposition.

\begin{proposition}
\label{Pr_A1.2} Under the assumption of the negligible influence of a
charged point particle $\xi $ on an external electromagnetic field source a
true physically reasonable action functional can be given by expression (\ref%
{A2.4}), being equivalently defined with respect to the rest reference frame 
$\mathcal{K}_{\tau }$ in form (\ref{A2.5}),(\ref{A2.5a}). The resulting
relativistic Lorentz force (\ref{A2.10}) coincides almost exactly with that
of (\ref{A1.6}), obtained from the classical Einstein type action functional
(\ref{A1.1}), but the momentum expression (\ref{A2.12}) differs from the
classical expression (\ref{A1.7}), taking into account the related vacuum
field potential interaction energy impact.
\end{proposition}

\bigskip As an important corollary we make the following.

\begin{corollary}
\label{Cor_A2.2}The Lorentz force expression (\ref{A2.10}) \ should be in due course corrected in the case when the weak charge $\xi$ influence
assumption made above does not hold.
\end{corollary}

\begin{remark}
\label{Rem_A2.3}Concerning the infinitesimal relationship (\ref{A2.6}) one
can observe that it reflects the Euclidean nature of transformations $%
\mathbb{R\ni }$ $t\rightleftharpoons \tau \in $ \bigskip $\mathbb{R}.$
\end{remark}

In spite of the results obtained above by means of two different least
action principles (\ref{A1.1}) and (\ref{A2.5}), we must claim here that the
first one possesses some logical controversies, which may give rise to
unpredictable, unexplainable and even nonphysical effects. Amongst these
controversies we mention: $i)$ the definition of Lagrangian function (\ref%
{A1.2}) as an expression, depending on the external and undefined rest mass
parameter with respect to the rest \ reference frame $\mathcal{K}_{\tau }$
time $\tau \in \mathbb{R},$ but serving as an variational integrand with
respect to the laboratory \ reference frame $\mathcal{K}_{t}$ \ time $t\in 
\mathbb{R};$ $ii)$ the least action condition (\ref{A1.1}) is calculated
with respect to the fixed boundary conditions at the ends of a time interval 
$[t_{1},t_{2}]\subset \mathbb{R},$ thereby the resulting dynamics becomes
strongly \ dependent on the chosen laboratory reference frame $\mathcal{K},$
what is, following the Feynman arguments \cite{Feyn,Feyn-1}, physically
unreasonable; $iii)$ the resulting relativistic particle mass and its energy
depend only on the particle velocity in the laboratory reference frame $%
\mathcal{K},$ not taking into account the present vacuum field potential
energy, exerting not trivial action on the particle motion; $iv)$ the
assumption concerning the negligible influence of a charged point particle
on the external electromagnetic field source is also physically inconsistent.

\section{An alternative hadronic string model least action formulation}

A classical relativistic hadronic string model was first proposed in \cite%
{BaCh,Namb,Godb} and deeply studied in \cite{BaNe}, making use of the least
action principle and related Lagrangian and Hamiltonian formalisms. We will
not discuss here this classical string model and will not comment the
physical problems accompanying it, especially those related to its diverse
quantization versions, but proceed to formulating a new relativistic
hadronic string model, constructed by means of the vacuum field theory
approach, devised in \cite{BoPrTa-acti,BoPrTaPr-Lore,BoPr-1}. The
corresponding least action principle is, following \cite{BaNe}, formulated
as 
\begin{equation}
\delta S^{(\tau )}=0,\text{ \ \ }S^{(\tau )}:=\int_{s(\tau _{1})}^{s(\tau
_{2})}\int_{\sigma _{1}(s)}^{\sigma _{2}(s)}\bar{W}(x(\xi ))(|\dot{\xi}%
|^{2}|\xi ^{\prime }|^{2}-<\dot{\xi},\xi ^{\prime }>_{\mathbb{E}%
^{4}}^{2})^{1/2}d\sigma \wedge ds,  \label{A5.1}
\end{equation}%
where $\bar{W}:M^{4}\rightarrow \mathbb{R}$ is a vacuum field potential
function, characterizing the interaction of the vacuum medium with our
charged string object, the differential 2-form $d\Sigma ^{(2)}:=(|\dot{\xi}%
|^{2}|\xi ^{\prime }|^{2}-<\dot{\xi},\xi ^{\prime }>_{\mathbb{E}%
^{4}}^{2})^{1/2}d\sigma \wedge ds$ $=\sqrt[2]{g(\xi )}d\sigma \wedge ds,$ $%
g(\xi ):=$ $\det (\left. g_{ij}(\xi )\right\vert _{i,j=\overline{1,2}}),$ $|%
\dot{\xi}|^{2}:=$ $<\dot{\xi},\dot{\xi}>_{\mathbb{E}^{4}},$ \ $|\xi ^{\prime
}|^{2}:=$ $\ <\xi ^{\prime },\xi ^{\prime }>_{\mathbb{E}^{4}},$ being
related with the induced positive define Riemannian infinitesimal metrics $%
dz^{2}:=<d\xi ,d\xi >_{\mathbb{E}^{4}}=g_{11}(\xi )d\sigma ^{2}+g_{12}(\xi
)d\sigma ds+g_{21}(\xi )dsd\sigma $ $+g_{22}(\xi )ds^{2}$ on the string,
means \cite{AbMa,BaNe,DuNoFo,Thir} the infinitesimal two-dimensional world
surface element, parameterized by variables $(s,\sigma )\in \mathbb{R}^{2}$
and embedded into the 4-dimensional Euclidean space-time $\mathbb{E}^{4}$
with coordinates $\xi :=(\tau (s,\sigma ),r(s,\sigma ))\in \mathbb{E}^{4}$
subject to the proper rest reference frame $\mathcal{K}_{\tau },$ $\dot{\xi}%
:=\partial \xi /\partial s,$ $\xi ^{\prime }:=\partial \xi /\partial \sigma $
are the corresponding partial derivatives. The related boundary conditions
are chosen as 
\begin{equation}
\delta \xi (s,\sigma (s))=0  \label{A5.2}
\end{equation}%
at string parameter $\sigma (s)\in \mathbb{R}$ for all $s\in \mathbb{R}.$
The action functional expression is strongly motivated by that constructed
for the point particle action functional (\ref{A5.1}): 
\begin{equation}
S^{(\tau )}:=-\int_{\sigma _{1}}^{\sigma _{2}}dl(\sigma )\int_{t(\sigma
,\tau _{1})}^{t(\sigma ,\tau _{2})}\bar{W}dt(\tau ,\sigma ),  \label{A5.3}
\end{equation}%
where the laboratory reference time parameter $t(\tau ,\sigma )\in \mathbb{R}
$ is related to the proper rest string reference frame time parameter $\tau
\in \mathbb{R}$ by means of the standard Euclidean infinitesimal
relationship 
\begin{equation}
dt(\tau ,\sigma ):=(1+|\dot{r}_{\perp }|^{2}(\tau ,\sigma ))^{1/2}d\tau ,%
\text{ \ \ \ }|\dot{r}_{\perp }|^{2}:=<\dot{r}_{\bot },\dot{r}_{\bot }>_{%
\mathbb{E}^{3}},  \label{A5.4}
\end{equation}%
with $\sigma \in \lbrack \sigma _{1},\sigma _{2}]\subset \mathbb{R},$ being
a spatial variable \ parameterizing the string length measure $dl(\sigma )$
on the real axis $\mathbb{R},$ $\dot{r}_{\perp }(\tau ,\sigma ):=\hat{N}$ $%
\dot{r}(\tau ,\sigma )\in \mathbb{E}^{3}$ being the orthogonal to the string
velocity component, and 
\begin{equation}
\text{\ }\ \hat{N}:=(1-|r^{\prime }|^{-2}r^{\prime }\otimes r^{\prime }),%
\text{ \ \ }|r^{\prime }|^{-2}:=<r^{\prime },r^{\prime }>_{\mathbb{E}%
^{3}}^{-1},  \label{A5.5}
\end{equation}%
being the corresponding projector operator in $\mathbb{E}^{3}$ on the
orthogonal to the string direction, expressed for brevity by means of the
standard tensor product $"\otimes "$ in the Euclidean space $\mathbb{E}^{3}.$
The result of calculation of (\ref{A5.3}) gives rise to the following
expression 
\begin{equation}
S^{(\tau )}=-\int_{\tau _{1}}^{\tau _{2}}d\tau \int_{\sigma _{1}(\tau
)}^{\sigma _{2}(\tau )}\bar{W}(|r^{\prime }|^{2}(1+|\dot{r}|^{2})-<\dot{r}%
,r^{\prime }>_{\mathbb{E}^{3}}^{2})^{1/2}d\ \sigma ,  \label{A5.6}
\end{equation}%
where we made use of the infinitesimal measure representation $dl(\sigma
)=<r^{\prime },r^{\prime }>_{\mathbb{E}^{3}}^{1/2}d\sigma ,$ $\sigma \in
\lbrack \sigma _{1},\sigma _{2}].$ If now to introduce on the string world
surface local coordinates $(s(\tau ,\sigma ),\sigma )\in \mathbb{E}^{2}$ and
the related Euclidean string position vector $\xi :=(\tau ,r(s,\sigma ))\in 
\mathbb{E}^{4},$ the string action functional reduces equivalently to that
of (\ref{A5.1}).

Below we will proceed to Lagrangian and Hamiltonian analyzing the least
action conditions for expressions (\ref{A5.1}) and (\ref{A5.6}).

\subsection{ Lagrangian and Hamiltonian analysis}

First we will obtain the corresponding to (\ref{A5.1}) \ \ Euler equations
with respect to the special \cite{BaNe,DuNoFo,BoPr-foun} internal conformal
variables $(s,\sigma )\in \mathbb{E}^{2}$ on the world string surface, with
respect to which the metrics on it becomes equal to $dz^{2}=|\xi ^{\prime
}|^{2}d\sigma ^{2}+|\dot{\xi}|^{2}ds^{2},$ where $<$ $\xi ^{\prime },\dot{\xi%
}$ $>_{\mathbb{E}^{4}}=0=|\xi ^{\prime }|^{2}-|\dot{\xi}|^{2}$ are the
imposed constraints, and the corresponding infinitesimal world surface
measure $d\Sigma ^{(2)}$ becomes $d\Sigma ^{(2)}=|\xi ^{\prime }|\dot{\xi}%
|d\sigma \wedge ds.$ As a result of simple calculations one finds the linear
second order partial differential equation 
\begin{equation}
\partial (\bar{W}\dot{\xi})/\partial s+\partial (\bar{W}\xi ^{\prime
})/\partial \sigma =\partial (|\xi ^{\prime }|\text{ }|\dot{\xi}|\bar{W}%
)/\partial \sigma  \label{A6.1}
\end{equation}%
under the suitably chosen boundary conditions%
\begin{equation}
\xi ^{\prime }-\dot{\xi}\text{ }\dot{\sigma}=0  \label{A6.2}
\end{equation}%
for all $s\in \mathbb{R}.$ It is interesting to mention that equation \ (\ref%
{A6.1}) \ is of \textit{elliptic type}, contrary to the case considered
before in \cite{BaNe}. This is, evidently, owing to the fact that the
resulting metrics on the string world surface is Euclidean, as we took into
account that the real string motion is, in reality, realized with respect to
its proper rest reference frame $\mathcal{K}_{\tau },$ being not dependent
on the string motion observation data, measured with respect to any external
laboratory reference frame $\mathcal{K}.$ The latter can be used for
physically motivated evidence of the dynamical stability of the relativistic
charged string object, modeling a charged hadronic particle \cite%
{Barb,Grov,Namb,Wilc-1} with non-trivial internal structure.

The differential equation (\ref{A6.1}) strongly depends on the vacuum field
potential function $\bar{W}:M^{4}\rightarrow \mathbb{R},$ which, as a
function of the Minkowski 4-vector variable $x:=(t(s,\sigma ),r)\in M^{4}$
of the laboratory reference frame $\mathcal{K},$ should be expressed as a
function of the variables $(s,\sigma )\in \mathbb{E}^{2},$ making use of the
infinitesimal relationship (\ref{A5.4}) in the following form:%
\begin{equation}
dt=<\hat{N}\partial \xi /\partial \tau ,\hat{N}\partial \xi /\partial \tau
>_{\mathbb{E}^{3}}^{1/2}(\frac{\partial \tau }{\partial s}ds+\frac{\partial
\tau }{\partial \sigma }d\sigma ),  \label{A6.3}
\end{equation}%
defined on the string world surface. The function $\bar{W}:M^{4}\rightarrow 
\mathbb{R}$ itself should be simultaneously found, following ideas of \cite%
{Bril,Tram} and recent results of \cite{BoPr-1,BoPr-foun}, by means of a
suitable solution to the Maxwell equation $\partial ^{2}W/\partial
t^{2}-\Delta W=\rho ,$ where $\rho \in \mathbb{R}$ is an ambient charge
density and, by definition, $\ \bar{W}(r(t))$ $:=$ $\lim_{r\rightarrow
r(t)}\left. W(r,t)\right\vert ,$ with $r(t)\in \mathbb{E}^{3}$ being the
position of the string element with a proper rest reference coordinates $%
(\tau ,\sigma )\in \mathbb{E}^{2}$ at the time moment $t=t(\tau ,\sigma )\in 
\mathbb{R}.$

We proceed now to constructing the dynamical Euler equations for our string
model, making use of the general action functional (\ref{A5.6}) in the
following form:%
\begin{equation}
S^{(\tau )}=-\int_{\tau _{1}}^{\tau _{2}}d\tau \int_{\sigma _{1}(\tau
)}^{\sigma _{2}(\tau )}\bar{W}|r^{\prime }|(1+|\dot{r}|^{2}-<r^{\prime
}|r^{\prime }|^{-1},\dot{r}>_{\mathbb{E}^{3}}^{2})^{1/2}d\ \sigma ,
\label{A6.3a}
\end{equation}%
It is easy to calculate that the generalized momentum density 
\begin{eqnarray}
p &:&=\partial \mathcal{L}^{(\tau )}/\partial \dot{r}=\frac{-\bar{W}%
|r^{\prime }|(\dot{r}-r^{\prime }|r^{\prime }|^{-2}<r^{\prime },\dot{r}>_{%
\mathbb{E}^{3}})}{(|\dot{r}|^{2}+1-<r^{\prime }|r^{\prime }|^{-1},\dot{r}>_{%
\mathbb{E}^{3}}^{2})^{1/2}}=  \notag \\
&=&\frac{-\bar{W}|r^{\prime }|\hat{N}dr\ }{d\tau (|\dot{r}|^{2}+1-<r^{\prime
}|r^{\prime }|^{-1},\dot{r}>_{\mathbb{E}^{3}}^{2})^{1/2}}=-|r^{\prime }|\bar{%
W}\hat{N}dr/dt\ =-|r^{\prime }|\hat{N}(\bar{W}u)  \label{A6.4}
\end{eqnarray}%
satisfies the dynamical equation 
\begin{eqnarray}
dp/d\tau &:&=\delta \mathcal{L}^{(\tau )}/\delta r=-(|r^{\prime }|^{2}(|\dot{%
r}|^{2}+1)-<r^{\prime },\dot{r}>_{\mathbb{E}^{3}}^{2})^{1/2}\ \nabla \bar{W}+
\label{A6.5} \\
&&+\frac{\partial }{\partial \sigma }\left\{ \frac{\bar{W}(1+|\dot{r}|^{2}%
\hat{T})r^{\prime }/r^{\prime }}{(1+<|\dot{r}|^{2}\hat{T}r^{\prime
}|r^{\prime }|^{-1},r^{\prime }|r^{\prime }|^{-1}>_{\mathbb{E}^{3}})^{1/2}}%
\right\} ,  \notag
\end{eqnarray}%
where we denoted by 
\begin{equation}
\mathcal{L}^{(\tau )}:=-\bar{W}(|r^{\prime }|^{2}(1+|\dot{r}|^{2})-<\dot{r}%
,r^{\prime }>_{\mathbb{E}^{3}}^{2})^{1/2}=-\bar{W}(|r^{\prime }|^{2}+|\dot{r}%
|^{2}<r^{\prime },\hat{T}_{\dot{r}}r^{\prime }>_{\mathbb{E}^{3}})^{1/2}
\label{A6.6}
\end{equation}%
the corresponding Lagrangian function, and for any vector $w\in \mathbb{E}%
^{3}$ 
\begin{equation}
\hat{T}_{w}:=1-|w|^{-2}\text{ }w\otimes w,\text{ \ \ \ }|w|^{2}:=<w,w>_{%
\mathbb{E}^{3}}^{2}\   \label{A6.7}
\end{equation}%
the usual projector operator in $\mathbb{E}^{3}.$ As a result of (\ref{A6.5}%
) one finds that%
\begin{eqnarray}
dp/dt &=&-|r^{\prime }|\ \nabla \bar{W}+  \label{A6.7a} \\
+(1-|u|^{2}) &<&u,r^{\prime }>_{\mathbb{E}^{3}}^{2})^{-1/2}\frac{\partial }{%
\partial \sigma }\left\{ \frac{\bar{W}(1-|u|^{2}+<u,r^{\prime }>_{\mathbb{E}%
^{3}}^{2}+|u|^{2}\hat{T}_{u})r^{\prime }/r^{\prime }}{(1-|u|^{2}+<u,r^{%
\prime }>_{\mathbb{E}^{3}}^{2})^{1/2}}\right\} ,  \notag
\end{eqnarray}%
where we took into account that owing to (\ref{A5.4}) 
\begin{equation}
\dot{r}=dr/d\tau =dr/dt\cdot (dt/d\tau )=u(1-|u|^{2}+<u,r^{\prime }>_{%
\mathbb{E}^{3}}^{2})^{-1/2}.  \label{A6.7b}
\end{equation}

The Lagrangian function is degenerate \cite{BaNe,DuNoFo}, satisfying the
obvious identity 
\begin{equation}
<p,r^{\prime }>_{\mathbb{E}^{3}}=0  \label{A6.8}
\end{equation}%
for all $\tau \in \mathbb{R}.$ Concerning the Hamiltonian formulation of the
dynamics (\ref{A6.5}) we construct the corresponding Hamiltonian functional
as%
\begin{equation}
\begin{array}{c}
\mathcal{H}:=\int_{\sigma _{1}}^{\sigma _{2}}(<p,\dot{r}>_{\mathbb{E}^{3}}-%
\mathcal{L}^{(\tau )})d\sigma = \\ 
\\ 
=\int_{\sigma _{1}}^{\sigma _{2}}\left( \frac{-\bar{W}|r^{\prime }|(|\dot{r}%
|^{2}-r^{\prime }|r^{\prime }|^{-2}<r^{\prime },\dot{r}>_{\mathbb{E}%
^{3}}^{2})}{(|\dot{r}|^{2}+1-<r^{\prime }|r^{\prime }|^{-1},\dot{r}>_{%
\mathbb{E}^{3}}^{2})^{1/2}}+\frac{\bar{W}|r^{\prime }|(|\dot{r}%
|^{2}+1-<r^{\prime }|r^{\prime }|^{-1},\dot{r}>_{\mathbb{E}^{3}}^{2})\ }{(|%
\dot{r}|^{2}+1-<r^{\prime }|r^{\prime }|^{-1},\dot{r}>_{\mathbb{E}%
^{3}}^{2})^{1/2}}\ \right) d\sigma = \\ 
\\ 
=\int_{\sigma _{1}}^{\sigma _{2}}\left( \frac{\ \bar{W}|r^{\prime }|}{(|\dot{%
r}|^{2}+1-<r^{\prime }|r^{\prime }|^{-1},\dot{r}>_{\mathbb{E}^{3}}^{2})^{1/2}%
}\right) d\sigma =-\int_{\sigma _{1}}^{\sigma _{2}}(\bar{W}^{2}|r^{\prime
}|^{2}-|p|^{2})^{1/2}d\sigma ,%
\end{array}
\label{A6.9}
\end{equation}%
satisfying the canonical Hamiltonian equations 
\begin{equation}
dr/d\tau :=\delta \mathcal{H}/\delta p,\text{ \ }dp/d\tau :=-\delta \mathcal{%
H}/\delta r,  \label{A6.10}
\end{equation}%
where 
\begin{equation}
d\mathcal{H}/d\tau =0,  \label{A6.10a}
\end{equation}%
holding only with respect to the proper rest reference frame \ $\mathcal{K}%
_{\tau }$ time parameter $\tau \in \mathbb{R}.$ Now making use of identity (%
\ref{A6.8}) the Hamiltonian functional (\ref{A6.9}) can be equivalently
represented \cite{BaNe} in the symbolic form as 
\begin{equation}
\mathcal{H}=\int_{\sigma _{1}}^{\sigma _{2}}|\bar{W}r^{\prime }\pm ip|_{%
\mathbb{C}}d\sigma ,\text{ }  \label{A6.11}
\end{equation}%
where \ $i:=\sqrt{-1}.$ Moreover, concerning the result obtained above we
need to mention here that one can not construct the suitable Hamiltonian
function expression and relationship of type (\ref{A6.10a}) with respect to
the laboratory reference frame $\mathcal{K},$ since expression (\ref{A6.11})
is not defined on the whole for a separate laboratory time parameter $t\in 
\mathbb{R}$ \ locally dependent both on the spatial parameter $\sigma \in 
\mathbb{R}$ and the proper rest reference frame time parameter $\tau \in 
\mathbb{R}.$

Thereby, one can formulate the following proposition.

\begin{proposition}
The hadronic string model (\ref{A5.1}) allows, on the related world surface,
the conformal local coordinates, with respect to which the resulting
dynamics is described by means of the linear second order elliptic equation (%
\ref{A6.1}). Subject to the proper Euclidean rest reference frame $\mathcal{K%
}_{\tau }$ coordinates the corresponding dynamics is equivalent to the
canonical Hamiltonian equations (\ref{A6.10}) with Hamiltonian functional (%
\ref{A6.9}).
\end{proposition}

We proceed now to construct the action functional expression for a charged
string under an external electromagnetic field, generated by a point
velocity charged particle $\mathbb{\xi }_{f},$ moving with some velocity $%
u_{f}:=dr_{f}/dt\in $ $\mathbb{E}^{3}$\ subject to a laboratory reference
frame $\mathcal{K}.$ To solve this problem we make use of the trick of
Section \ref{Sec_1}, passing to the string, \ considered with respect to the
proper rest reference frame $\mathcal{K}_{\tau }$ moving under the external
vacuum field potential $\bar{W}^{\prime }$ with respect to the relative
reference frame $\mathcal{K}_{f}^{\prime },$ specified by its own Euclidean
coordinates $(t^{\prime },r_{f})\in \mathbb{E}^{4},$ which simultaneously
moves with velocity $u_{f}\in $ $\mathbb{E}^{3}$ with respect to the
laboratory reference frame $\mathcal{K}.$ As a result of this reasoning we
can write down the action functional: 
\begin{equation}
S^{(\tau )}=-\int_{\tau _{1}}^{\tau _{2}}d\tau \int_{\sigma _{1}(\tau
)}^{\sigma _{2}(\tau )}\bar{W}^{\prime }(|r^{\prime }|^{2}(1+|\dot{r}-\dot{r}%
_{f}|^{2})-<\dot{r}-\dot{r}_{f},r^{\prime }>_{\mathbb{E}^{3}}^{2})^{1/2}d\
\sigma ,  \label{A6.12}
\end{equation}%
giving rise to the following dynamical equation 
\begin{eqnarray}
dP/d\tau &:&=\delta \mathcal{L}^{(\tau )}/\delta r=-(|r^{\prime }|^{2}(1+|%
\dot{r}-\dot{r}_{f}|^{2})-<\dot{r}-\dot{r}_{f},r^{\prime }>_{\mathbb{E}%
^{3}}^{2})^{1/2}\ \nabla \bar{W}^{\prime }+  \label{A6.13} \\
&&+\frac{\partial }{\partial \sigma }\left\{ \frac{\bar{W}^{\prime }(1+|\dot{%
r}-\dot{r}_{f}|^{2}\hat{T}_{\dot{r}-\dot{r}_{f}})r^{\prime }}{(|r^{\prime
}|^{2}(1+|\dot{r}-\dot{r}_{f}|^{2})-<\dot{r}-\dot{r}_{f},r^{\prime }>_{%
\mathbb{E}^{3}}^{2})^{1/2}}\right\} ,  \notag
\end{eqnarray}%
where the generalized momentum density 
\begin{equation}
P:=\frac{-\bar{W}^{\prime }|r^{\prime }|^{2}\hat{N}(\dot{r}-\dot{r}_{f})}{%
(|r^{\prime }|^{2}(1+|\dot{r}-\dot{r}_{f}|^{2})-<\dot{r}-\dot{r}%
_{f},r^{\prime }>_{\mathbb{E}^{3}\mathbb{E}^{3}}^{2})^{1/2}}.  \label{A6.14}
\end{equation}%
Owing to (\ref{A6.14}), one can define 
\begin{eqnarray}
p &:&=\frac{-\bar{W}^{\prime }|r^{\prime }|^{2}\hat{N}\dot{r}}{(|r^{\prime
}|^{2}(1+|\dot{r}-\dot{r}_{f}|^{2})-<\dot{r}-\dot{r}_{f},r^{\prime }>_{%
\mathbb{E}^{3}}^{2})^{1/2}}=  \label{A6.16} \\
&=&-\frac{-\bar{W}^{\prime }|r^{\prime }|\hat{N}dr}{d\tau (1+|\dot{r}-\dot{r}%
_{f}|^{2}-<\dot{r}-\dot{r}_{f},r^{\prime }/|r^{\prime }|>_{\mathbb{E}%
^{3}}^{2})^{1/2}}=-\bar{W}^{\prime }|r^{\prime }|\hat{N}u^{\prime }=-\bar{W}%
|r^{\prime }|\hat{N}u  \notag
\end{eqnarray}%
as the local string momentum density and 
\begin{eqnarray}
\mathbb{\xi }|r^{\prime }|A &:&=\frac{\bar{W}^{\prime }|r^{\prime }|^{2}\hat{%
N}\dot{r}_{f}}{(|r^{\prime }|^{2}(1+|\dot{r}-\dot{r}_{f}|^{2})-<\dot{r}-\dot{%
r}_{f},r^{\prime }>_{\mathbb{E}^{3}}^{2})^{1/2}}=  \label{A6.17} \\
&=&\frac{\bar{W}^{\prime }|r^{\prime }|\hat{N}dr_{f}}{d\tau (1+|\dot{r}-\dot{%
r}_{f}|^{2}-<\dot{r}-\dot{r}_{f},r^{\prime }/|r^{\prime }|>_{\mathbb{E}%
^{3}}^{2})^{1/2}}=\bar{W}^{\prime }|r^{\prime }|\hat{N}u_{f}^{\prime }=\bar{W%
}|r^{\prime }|\hat{N}u_{f}  \notag
\end{eqnarray}%
as the external vector magnetic potential density, where $\mathbb{\xi }\in 
\mathbb{R}$ is a uniform charge density, distributed along the string
length. Thus, equation (\ref{A6.13}) reduces to%
\begin{equation}
\begin{array}{c}
d(p+\mathbb{\xi }|r^{\prime }|A)/dt^{\prime }=-|r^{\prime }|\ \nabla \bar{W}%
^{\prime }+(1-|u^{\prime }-u_{f}^{\prime }|^{2}+<u^{\prime }-u_{f}^{\prime
},r^{\prime }>_{\mathbb{E}^{3}}^{2})^{-1/2}\times \\ 
\end{array}
\label{A6.17a}
\end{equation}%
\begin{equation*}
\times \frac{\partial }{\partial \sigma }\left\{ \frac{\bar{W}^{\prime
}(1-|u^{\prime }-u_{f}^{\prime }|^{2}+<u^{\prime }-u_{f}^{\prime },r^{\prime
}>_{\mathbb{E}^{3}}^{2}+|u^{\prime }-u_{f}^{\prime }|^{2}\hat{T}_{u^{\prime
}-u_{f}^{\prime }}r^{\prime }/r^{\prime })}{(1-|u^{\prime }-u_{f}^{\prime
}|^{2}+<u^{\prime }-u_{f}^{\prime },r^{\prime }>_{\mathbb{E}^{3}}^{2})^{1/2}}%
\right\}
\end{equation*}%
with respect to the moving reference frame $\mathcal{K}_{f}^{\prime },$ or
equivalently, to%
\begin{equation}
\begin{tabular}{l}
$d(p+\mathbb{\xi }|r^{\prime }|A)/dt=-|r^{\prime }|$ $\nabla \bar{W}%
(1-|u_{f}|^{2})+\ \ $ \\ 
\\ 
$+(1-|u_{f}|^{2})(1-|u_{f}|^{2}-|u\ -u_{f}|^{2}+<u-u_{f},r^{\prime }>_{%
\mathbb{E}^{3}}^{2})^{-1/2}\times $ \\ 
\\ 
$\times \frac{\partial }{\partial \sigma }\left\{ \frac{\bar{W}%
(1-|u_{f}|^{2}-|u-u_{f}|^{2}+<u-u_{f},r^{\prime }>_{\mathbb{E}%
^{3}}^{2}+|u-u_{f}|^{2}\hat{T}_{u-u_{f}}r^{\prime }/r^{\prime })}{%
(1-|u_{f}|^{2}-|u-u_{f}|^{2}+<u-u_{f},r^{\prime }>_{\mathbb{E}%
^{3}}^{2})^{1/2}}\right\} $%
\end{tabular}
\label{A6.17b}
\end{equation}%
with respect to the moving laboratory frame $\mathcal{K}.$ \ The latter can
be easily rewritten also as the Lorentz type force expression

\begin{equation}
\begin{array}{c}
dp/dt=\mathbb{\xi }|r^{\prime }|E+\mathbb{\xi }|r^{\prime }|u\times B-%
\mathbb{\xi }|r^{\prime }|\ \nabla <u-u_{f},A>_{\mathbb{E}^{3}}+ \\ 
\\ 
+(1-|u_{f}|^{2})(1-|u_{f}|^{2}-|u\ -u_{f}|^{2}+<u-u_{f},r^{\prime }>_{%
\mathbb{E}^{3}}^{2})^{-1/2}\times \\ 
\\ 
\times \frac{\partial }{\partial \sigma }\left\{ \frac{\bar{W}%
(1-|u_{f}|^{2}-|u-u_{f}|^{2}+<u-u_{f},r^{\prime }>_{\mathbb{E}%
^{3}}^{2}+|u-u_{f}|^{2}\hat{T}_{u-u_{f}}r^{\prime }/r^{\prime })}{%
(1-|u_{f}|^{2}-|u-u_{f}|^{2}+<u-u_{f},r^{\prime }>_{\mathbb{E}%
^{3}}^{2})^{1/2}}\right\} ,%
\end{array}
\label{A6.18}
\end{equation}%
where $B=\ \nabla \times A$ means, as usual, the external magnetic field and 
\begin{equation}
E\ =\ \partial A\ /\partial t-\ \nabla \bar{W}\ \   \label{A6.19a}
\end{equation}%
means the corresponding electric field, acting on the string. Making use of
the standard scheme, one can derive, as above, the Hamiltonian
interpretation of dynamical equations \ (\ref{A6.13}), but which will not be
here discussed.

\section{\label{Sec_7}The generalized Fock spaces, quantum currents algebra
representations and electrodynamics}

\subsection{Preliminaries: Fock space and its realizations}

Let $\Phi $ be a separable Hilbert space, $F$ be a topological real linear
space and $\mathcal{A}:=\left\{ A(f):f\in F\right\} $ a family of commuting
self-adjoint operators in $\Phi $ (i.e. these operators commute in the sense
of their resolutions of the identity). Consider the Gelfand rigging \cite%
{Bere-1,GeVi} of the Hilbert space ${\Phi ,}$ i.e., a chain 
\begin{equation}
\mathcal{D}\subset {\Phi }_{+}\subset {\Phi }\subset {\Phi }_{-}\subset 
\mathcal{D^{^{\prime }}}  \label{FQeq2.1}
\end{equation}%
in which ${\Phi }_{+}$ and ${\Phi }_{-}$ are further Hilbert spaces, and the
inclusions are dense and continuous, i.e. ${\Phi }_{+}$ is topologically
(densely and continuously) and quasi-nucleus (the inclusion operator $i:{%
\Phi }_{+}\longrightarrow {\Phi }$ is of the Hilbert - Schmidt type)
embedded into ${\Phi },$ the space $\ {\Phi }_{-}$ is the dual of ${\Phi }%
_{+}$ with respect to the scalar product $<.,.>_{{\Phi }}$ in ${\Phi },$ and 
$\mathcal{D}$ is a separable projective limit of Hilbert spaces,
topologically embedded into ${\Phi }_{+}.$ Then, the following structural
theorem \cite{Gold,GoGrPoSh,GoMeSh,GoMeSh-1} holds:

\begin{theorem}
\label{FQth2.1} Assume that the family of operators $\mathcal{A}$ \
satisfies the following conditions:

a)$\quad \mathcal{D}\subset DomA(f),\;f\in F,$ and the closure of the
operator $A(f)\uparrow \mathcal{D}$ coincides with $A(f)$ for any $f\in F,$
that is $A(f)\uparrow \mathcal{D}=A(f)$ in ${\Phi };$

b) the Range $A(f)\uparrow \mathcal{D}\subset {\Phi }_{+}$ \ for any $f\in F$%
;

c) for every $\psi \in \mathcal{D}$ the mapping $F\ni f\longrightarrow
A(f)\psi \in {\Phi }_{+}$ is linear and continuous;

d) there exists a strong cyclic (vacuum) vector $|\Omega \rangle \in
\bigcap_{f\in F}DomA(f),$ such that the set of all vectors $|\Omega \rangle
, $ $\prod_{j=1}^{n}A(f_{j})|\Omega \rangle ,$ $n\in \mathbb{Z}_{+},$ is
total in ${\Phi }_{+}$ (i.e. their linear hull is dense in ${\Phi }_{+}$).

Then there exists a probability measure $\mu $ on $(F^{\prime },C_{\sigma
}(F^{\prime }))$, where $F^{\prime }$ is the dual of $F$ and $C_{\sigma
}(F^{\prime })$ is the $\sigma -$algebra generated by cylinder sets in $%
F^{\prime }$ such that, for $\mu -$almost every $\eta \in F^{\prime }$ there
is a generalized joint eigenvector $\omega (\eta )\in {\Phi }_{-}$ of the
family $\mathcal{A},$ corresponding to the joint eigenvalue $\eta \in
F^{\prime },$ that is 
\begin{equation}
<\omega (\eta ),A(f)\psi >_{{\Phi }}=\eta (f)<\omega (\eta ),\psi >_{{\Phi }}
\label{FQeq2.1a}
\end{equation}%
with $\eta (f)\in \mathbb{R}$ denoting the pairing between $F$ and $%
F^{\prime }$.

The mapping 
\begin{equation}
{\Phi }_{+}\ni \psi \longrightarrow <\omega (\eta ),\psi >_{{\Phi }}:=\psi
(\eta )\in \mathbb{C}  \label{FQeq2.1b}
\end{equation}%
for any $\eta \in F^{\prime }$ can be continuously extended to a unitary
surjective operator $\mathcal{F}:{\Phi }\longrightarrow L_{2}^{(\mu
)}(F^{\prime };\mathbb{C}),$ where 
\begin{equation}
\mathcal{F}\text{ }\psi (\eta ):=\psi (\eta )  \label{FQeq2.1c}
\end{equation}%
for any $\eta \in F^{\prime }$ is a generalized Fourier transform,
corresponding to the family $\mathcal{A}$. Moreover, the image of the
operator $A(f)$, $f\in F^{\prime }$, under the $\mathcal{F}-$mapping is the
operator of multiplication by the function $F^{\prime }\ni \eta \rightarrow
\eta (f)\in \mathbb{C}.$
\end{theorem}

We assume additionally that the main Hilbert space $\Phi $ possesses the
standard Fock space (bose)-structure \cite{BoBo,Bere,BoPrTa}, that is 
\begin{equation}
{\Phi }=\oplus _{n\in \mathbb{Z}_{+}}{\Phi }_{(s)}^{\otimes n},
\label{FQeq2.1d}
\end{equation}%
where subspaces ${\ \Phi }_{(s)}^{\otimes n},$ $n\in \mathbb{Z}_{+}$, are
the symmetrized tensor products of a Hilbert space $\mathcal{H}:=L_{2}(%
\mathbb{R}^{m};\mathbb{C}).$ If a vector $g:=(g_{0},g_{1},...,g_{n},...)\in
\Phi $, its norm 
\begin{equation}
\Vert g\Vert _{\Phi }:=\left( \sum_{n\in \mathbb{Z}_{+}}\Vert g_{n}\Vert
_{n}^{2}\right) ^{1/2},  \label{FQeq2.2}
\end{equation}%
where $g_{n}\in {\Phi }_{(s)}^{\otimes n}\simeq L_{2,(s)}((\mathbb{R}%
^{m})^{\otimes n};\mathbb{C})$ and $\parallel ...\parallel _{n}$ is the
corresponding norm in ${\Phi }_{(s)}^{\otimes n}$ for all $n\in \mathbb{Z}%
_{+}$. Note here that concerning the rigging structure (\ref{FQeq2.1}),
there holds the corresponding rigging for the Hilbert spaces ${\Phi }%
_{(s)}^{\otimes n}$, $n\in \mathbb{Z}_{+}$, that is 
\begin{equation}
\mathcal{D}_{(s)}^{n}\subset {\Phi }_{(s),+}^{\otimes n}\subset {\Phi }%
_{(s)}^{\otimes n}\subset {\Phi }_{(s),-}^{\otimes n}  \label{FQeq2.3}
\end{equation}%
with some suitably chosen dense and separable topological spaces of
symmetric functions $\mathcal{D}_{(s)}^{n},$ $n\in \mathbb{Z}_{+}.$
Concerning expansion (\ref{FQeq2.1}) we obtain by means of projective and
inductive limits \cite{Bere-1,GeVi} the quasi-nucleus rigging of the Fock
space $\Phi $ in the form (\ref{FQeq2.1}):%
\begin{equation*}
\mathcal{D}\subset {\Phi }_{+}\subset {\Phi }\subset {\Phi }_{-}\subset 
\mathcal{D^{^{\prime }}}.
\end{equation*}

Consider now any vector $|(\alpha )_{n}\rangle \in {\Phi }_{(s)}^{\otimes
n}, $ $n\in \mathbb{Z}_{+},$ which can be written \cite{Bere,BoBo,KoSt} in
the following canonical Dirac ket-form: 
\begin{equation}
|(\alpha )_{n}\rangle :=|\alpha _{1},\alpha _{2},...,\alpha _{n}\rangle ,
\label{FQeq2.4}
\end{equation}%
where, by definition, \ 
\begin{equation}
|\alpha _{1},\alpha _{2},...,\alpha _{n}\rangle :=\frac{1}{\sqrt{n!}}%
\sum_{\sigma \in S_{n}}|\alpha _{\sigma (1)}\rangle \otimes |\alpha _{\sigma
(2)}\rangle ...|\alpha _{\sigma (n)}\rangle  \label{FQeq2.5}
\end{equation}%
and $|\alpha _{j}\rangle \in {\Phi }_{(s)}^{\otimes 1}(\mathbb{R}^{m};%
\mathbb{C}):=\mathcal{H}$ for any fixed $j\in \mathbb{Z}_{+}$. The
corresponding scalar product of base vectors as (\ref{FQeq2.5}) is given as
follows: 
\begin{equation}
\begin{array}{c}
\langle (\beta )_{n}|(\alpha )_{n}\rangle :=\langle \beta _{n},\beta
_{n-1},...,\beta _{2},\beta _{1}|\alpha _{1},\alpha _{2},...,\alpha
_{n-1},\alpha _{n}\rangle \\[5pt] 
=\sum_{\sigma \in S_{n}}\langle \beta _{1}|\alpha _{\sigma (1)}\rangle
...\langle \beta _{n}|\alpha _{\sigma (n)}\rangle :=per\{\langle \beta
_{i}|\alpha _{j}\rangle :i,j=\overline{1,n}\},%
\end{array}
\label{FQeq2.6}
\end{equation}%
where \textquotedblleft $per$\textquotedblright\ denotes the permanent of
matrix and $\langle .|.\rangle $ is the corresponding product in the Hilbert
space $\mathcal{H}$. Based now on representation (\ref{FQeq2.4}) one can
define an operator $a^{+}(\alpha ):{\Phi }_{(s)}^{\otimes n}\longrightarrow {%
\Phi }_{(s)}^{\otimes (n+1)}$ for any $|\alpha \rangle \in \mathcal{H}$ as
follows: 
\begin{equation}
a^{+}(\alpha )|\alpha _{1},\alpha _{2},...,\alpha _{n}\rangle :=|\alpha
,\alpha _{1},\alpha _{2},...,\alpha _{n}\rangle ,  \label{FQeq2.7}
\end{equation}%
which is called the "\textit{creation}" operator in the Fock space $\Phi $.
The adjoint operator $a(\beta ):=(a^{+}(\beta ))^{\ast }:{\Phi }%
_{(s)}^{\otimes (n+1)}\longrightarrow {\Phi }_{(s)}^{\otimes n}$ with
respect to the Fock space ${\Phi }$ (\ref{FQeq2.1}) for any $|\beta \rangle
\in \mathcal{H},$ called the "\textit{annihilation}" operator, acts as
follows: 
\begin{equation}
a(\beta )|\alpha _{1},\alpha _{2},...,\alpha _{n+1}\rangle :=\sum_{\sigma
\in S_{n}}\langle \beta ,\alpha _{j}\rangle |\alpha _{1},\alpha
_{2},...,\alpha _{j-1},\hat{\alpha}_{j},\alpha _{j+1},...,\alpha
_{n+1}\rangle ,  \label{FQeq2.8}
\end{equation}%
where the $"hat"$ over a vector denotes that it should be omitted from the
sequence.

It is easy to check that the commutator relationship 
\begin{equation}
\lbrack a(\alpha ),a^{+}(\beta )]=\langle \alpha ,\beta \rangle
\label{FQeq2.9}
\end{equation}%
holds for any vectors $|\alpha \rangle \in \mathcal{H}$ and $|\beta \rangle
\in \mathcal{H}$. Expression (\ref{FQeq2.9}), owing to the rigged structure (%
\ref{FQeq2.1}), can be naturally extended to the general case, when vectors $%
\ |\alpha \rangle $ and $|\beta \rangle \in \mathcal{H}_{-}$, conserving its
form. In particular, taking $|\alpha \rangle :=|\alpha (x)\rangle =\frac{1}{%
\sqrt{2\pi }}e^{i\langle \lambda ,x\rangle }\in \mathcal{H}_{-}:=L_{2,-}(%
\mathbb{R}^{m};\mathbb{C})$ for any $x\in \mathbb{R}^{m}$, one easily gets
from (\ref{FQeq2.9}) that 
\begin{equation}
\lbrack a(x),a^{+}(y)]=\delta (x-y),  \label{FQeq2.10}
\end{equation}%
where we put, by definition, $a^{+}(x):=a^{+}(\alpha (x))$ and $%
a(y):=a(\alpha (y))$ for all $x,y\in \mathbb{R}^{m}$ and denoted by $\delta
(\cdot )$ the classical Dirac delta-function.

The construction above makes it possible to observe easily that there exists
a unique vacuum vector $|\Omega \rangle \in \mathcal{H}_{+}$, such that for
any $x\in \mathbb{R}^{m}$ 
\begin{equation}
a(x)|\Omega \rangle =0,  \label{FQeq2.11}
\end{equation}%
and the set of vectors 
\begin{equation}
\left( \prod_{j=1}^{n}a^{+}(x_{j})\right) |\Omega \rangle \in {\Phi }%
_{(s)}^{\otimes n}  \label{FQeq2.12}
\end{equation}%
is total in ${\Phi }_{(s)}^{\otimes n}$, that is their linear integral hull
over the dual functional spaces $\hat{\Phi}_{(s)}^{\otimes n}$ is dense in
the Hilbert space ${\Phi }_{(s)}^{\otimes n}$ for every $n\in \mathbb{Z}_{+}$%
. This means that for any vector $g\in \Phi $ the following representation 
\begin{equation}
g=\oplus _{n\in \mathbb{Z}_{+}}\int_{(\mathbb{R}^{m})^{n}}\hat{g}%
_{n}(x_{1},...,x_{n})a^{+}(x_{1})a^{+}(x_{2})...a^{+}(x_{n})|\Omega \rangle
\label{FQeq2.13}
\end{equation}%
holds with the Fourier type coefficients $\hat{g}_{n}\in \hat{\Phi}%
_{(s)}^{\otimes n}$ for all $n\in \mathbb{Z}_{+}$, with $\hat{\Phi}%
_{(s)}^{\otimes 1}:=\mathcal{H}\simeq {L}_{2}(\mathbb{R}^{m};\mathbb{C}).$
The latter is naturally endowed with the Gelfand type quasi-nucleus rigging
dual to 
\begin{equation}
\mathcal{H}_{+}\subset \mathcal{H}\subset \mathcal{H}_{-},  \label{FQeq2.14}
\end{equation}%
making it possible to construct a quasi-nucleolus rigging of the dual Fock
space ${\hat{\Phi}}:=\oplus _{n\in \mathbb{Z}_{+}}{\hat{\Phi}}%
_{(s)}^{\otimes n}.$ Thereby, chain (\ref{FQeq2.14}) generates the dual Fock
space quasi-nucleolus rigging 
\begin{equation}
\mathcal{\hat{D}}\subset {\hat{\Phi}}_{+}\subset {\hat{\Phi}}\subset {\hat{%
\Phi}}_{-}\subset \mathcal{\hat{D}}^{\prime }  \label{FQeq2.15}
\end{equation}%
with respect to the central Fock type Hilbert space ${\hat{\Phi}},$ where $%
\mathcal{\hat{D}}\simeq \mathcal{D},$ easily following from (\ref{FQeq2.1})
and (\ref{FQeq2.14}).

Construct now the following self-adjoint operator $\ \ \rho (x):\Phi
\rightarrow \Phi $ \ \ as 
\begin{equation}
\rho (x):=a^{+}(x)a(x),  \label{FQeq2.16}
\end{equation}%
called the density operator at a point $x\in \mathbb{R}^{m},$ satisfying the
commutation properties: 
\begin{equation}
\begin{array}{c}
\lbrack \rho (x),\rho (y)]=0, \\[5pt] 
\lbrack \rho (x),a(y)]=-a(y)\delta (x-y), \\[5pt] 
\lbrack \rho (x),a^{+}(y)]=a^{+}(y)\delta (x-y)%
\end{array}
\label{FQeq2.17}
\end{equation}%
for all $y\in \mathbb{R}^{m}$.

Now, if to construct the following self-adjoint family $\mathcal{A}:=\left\{
\int_{\mathbb{R}^{m}}\rho (x)f(x)dx:f\in F\right\} $ of linear operators in
the Fock space $\Phi ,$\ where $F$ $:=\mathcal{S}(\mathbb{R}^{m};\mathbb{R})$
is the Schwartz functional space, one can derive, making use of Theorem \ref%
{FQth2.1}, that there exists the generalized Fourier transform (\ref%
{FQeq2.1c}), such that 
\begin{equation}
{\Phi }(\mathcal{H})=L_{2}^{(\mu )}(\mathcal{S}^{\prime };\mathbb{C})\simeq
\int_{\mathcal{S}^{\prime }}^{\oplus }\Phi _{\eta }d\mu (\eta )
\label{FQeq2.17a}
\end{equation}%
for some Hilbert space sets $\Phi _{\eta },$ $\eta \in F^{\prime },$ \ and a
suitable measure $\mu $ on $\mathcal{S}^{\prime },$ \ with respect to which
the corresponding joint eigenvector $\omega (\eta )\in \Phi _{+}$ for any $%
\eta \in F^{\prime }$ generates the Fourier transformed family $\hat{u}%
=\left\{ \eta (f)\in \mathbb{R}:\;\;f\in F\right\} $. Moreover, if $\dim
\Phi _{\eta }=1$ for all $\ \eta \in F,$ the Fourier transformed eigenvector 
$\hat{\omega}(\eta ):=\Omega (\eta )=1$ for all $\eta \in F^{^{\prime }}.$

Now we will consider the family of self-adjoint operators $\hat{u}$ as
generating a unitary family $\mathcal{U}:=\left\{ U(f):f\in F\right\} =\exp
(i\hat{u}),$ where for any $\rho (f)\in \hat{u},$ $f\in F,$ the operator 
\begin{equation}
U(f):=\exp [i\rho (f)]  \label{FQeq2.18}
\end{equation}%
is unitary, satisfying the abelian commutation condition 
\begin{equation}
U(f_{1})U(f_{2})=U(f_{1}+f_{2})  \label{FQeq2.19}
\end{equation}%
for any $f_{1},f_{2}\in F.$

Since, in general, the unitary family $\mathcal{U}=\exp (i\hat{u})$ is
defined in some Hilbert space $\Phi $, not necessarily being of Fock type,
the important problem of describing its Hilbertian cyclic representation
spaces arises, within which the factorization 
\begin{equation}
\rho (f)=\int_{\mathbb{R}^{m}}a^{+}(x)a(x)f(x)dx  \label{FQeq2.20}
\end{equation}%
jointly with relationships (\ref{FQeq2.17}) hold for any $f\in F$. This
problem can be treated using mathematical tools devised both within the
representation theory of $C^{\ast }$-algebras \cite{Dira,Dira-1} and the
Gelfand--Vilenkin \cite{Gold,GeVi} approach. Below we will describe the main
features of the Gelfand--Vilenkin formalism, being much more suitable for
the task, providing a reasonably unified framework of constructing the
corresponding representations.

\begin{definition}
\label{FQdef2.1} Let $F$ be a locally convex topological vector space, $%
F_{0}\subset F$ be a finite dimensional subspace of $F$. \ Let $%
F^{0}\subseteq F^{\prime }$ be defined by 
\begin{equation}
F^{0}:=\left\{ \xi \in F^{\prime }:\;\;\xi |_{F_{0}}=0\right\} ,
\label{FQeq2.21}
\end{equation}%
and called the annihilator of $F_{0}$.
\end{definition}

The quotient space $F^{\prime 0}:=F^{\prime }/F^{0}$ may be identified with $%
F_{0}^{\prime }\subset F^{\prime }$, the adjoint space of $F_{0}$.

\begin{definition}
\label{FQdef2.2} Let $A\subseteq F^{\prime };$ then the subset

\begin{equation}
X_{F^{0}}^{(A)}:=\left\{ \xi \in F^{\prime }:\xi +F^{0}\subset A\right\}
\label{FQeq2.22}
\end{equation}%
is called the cylinder set with base $A$ and generating subspace $F^{0}$.
\end{definition}

\begin{definition}
\label{FQdef2.3} Let $n=\dim F_{0}=\dim F_{0}^{\prime }=\dim F^{\prime 0}.$
One says that a cylinder set $X^{(A)}$ has Borel base, if $A$ is Borel, when
regarded as a subset of $\mathbb{R}^{m}$.
\end{definition}

The family of cylinder sets with Borel base forms an algebra of sets.

\begin{definition}
\label{FQdef2.4} The measurable sets in $F^{\prime }$ are the elements of
the $\sigma -$ algebra generated by the cylinder sets with Borel base.
\end{definition}

\begin{definition}
\label{FQdef2.5} A cylindrical measure in $F^{\prime }$ is a real-valued $%
\sigma -$pre-additive function $\mu $ defined on the algebra of cylinder
sets with Borel base and satisfying the conditions $0\leq \mu (X)\leq 1$ for
any $X,$ $\mu (F^{\prime })=1$ and $\mu \left( \coprod_{j\in \mathbb{Z}%
_{+}}X_{j}\right) =\sum_{j\in \mathbb{Z}_{+}}\mu (X_{j}),$ if all sets $%
X_{j}\subset F^{\prime }$, $j\in \mathbb{Z}_{+}$, have a common generating
subspace $F_{0}\subset F$.
\end{definition}

\begin{definition}
\label{FQdef2.6} A cylindrical measure $\mu $ satisfies the commutativity
condition if and only if for any bounded continuous function $\alpha
:F^{n}\longrightarrow \mathbb{R}$ of $n\in \mathbb{Z}_{+}$ real variables
the function 
\begin{equation}
\alpha \lbrack f_{1},f_{2},...,f_{n}]:=\int_{F^{\prime }}\alpha (\eta
(f_{1}),\eta (f_{2}),...,\eta (f_{n}))d\mu (\eta )  \label{FQeq2.23}
\end{equation}%
is sequentially continuous in $f_{j}\in F,$ $j=\overline{1,m}.$ (It is well
known \cite{GeVi,Gold} that in countably normalized spaces the properties of
sequential and ordinary continuity are equivalent).
\end{definition}

\begin{definition}
\label{FQdef2.7} A cylindrical measure $\mu $ is countably additive if and
only if for any cylinder set $X=\coprod_{j\in \mathbb{Z}_{+}}X_{j}$, which
is the union of countably many mutually disjoints cylinder sets $%
X_{j}\subset F^{\prime },j\in \mathbb{Z}_{+},$ $\mu (X)=\sum_{j\in \mathbb{Z}%
_{+}}\mu (X_{j})$.
\end{definition}

The following propositions hold.

\begin{proposition}
\label{FQpr2.8} A countably additive cylindrical measure $\mu $ can be
extended to a countably additive measure on the $\sigma $- algebra,
generated by the cylinder sets with Borel base. Such a measure will also be
called a cylindrical measure.
\end{proposition}

\begin{proposition}
\label{FQpr2.9} Let $F$ be a nuclear space. Then any cylindrical measure $%
\mu $ on $F^{\prime }$, satisfying the continuity condition, is countably
additive.
\end{proposition}

\begin{definition}
\label{FQdef2.10} Let $\mu $ be a cylindrical measure in $F^{\prime }$. The
Fourier transform of $\mu $ is the nonlinear functional 
\begin{equation}
\mathcal{L}(f):=\int_{F^{\prime }}\exp [i\eta (f)]d\mu (\eta ).
\label{FQeq2.24}
\end{equation}
\end{definition}

\begin{definition}
\label{FQdef2.11} The nonlinear functional $\mathcal{L}:F\longrightarrow 
\mathbb{C}$ on $F,$ defined by (\ref{FQeq2.24}), is called positive
definite, if and only if for all $f_{j}\in F$ and $\lambda _{j}\in \mathbb{C}%
,$ $j=\overline{1,n},$ the condition 
\begin{equation}
\sum_{j,k=1}^{n}\bar{\lambda}_{j}\mathcal{L}(f_{k}-f_{j})\lambda _{k}\geq 0
\label{FQeq2.25}
\end{equation}%
holds for any $n\in \mathbb{Z}_{+}$.
\end{definition}

\begin{proposition}
\label{FQpr2.12} The functional $\mathcal{L}:F\longrightarrow\mathbb{C}$ on $%
F $, defined by (\ref{FQeq2.24}), is the Fourier transform of a cylindrical
measure on $F^{\prime }$, if and only if it is positive definite,
sequentially continuous and satisfying the condition $\mathcal{L}(0)=1$.
\end{proposition}

Suppose now that we have a continuous unitary representation of the unitary
family $\mathcal{U}$ in a Hilbert space $\Phi $ with a cyclic vector $%
|\Omega \rangle \in \Phi $. Then we can put 
\begin{equation}
\mathcal{L}(f):=\langle \Omega |U(f)|\Omega \rangle  \label{FQeq2.26}
\end{equation}%
for any $f\in F:=\mathcal{S}$, being the Schwartz space on $\mathbb{R}^{m}$,
and observe that functional (\ref{FQeq2.26}) is continuous on $F$ owing to
the continuity of the representation. Therefore, this functional is the
generalized Fourier transform of a cylindrical measure $\mu $ on $\mathcal{S}%
^{\prime }:$ 
\begin{equation}
\langle \Omega |U(f)|\Omega \rangle =\int_{\mathcal{S}^{\prime }}\exp [i\eta
(f)]d\mu (\eta ).  \label{FQeq2.27}
\end{equation}%
From the spectral point of view, based on Theorem \ref{FQth2.1}, there is an
isomorphism between the Hilbert spaces $\Phi $ and $L_{2}^{(\mu )}(\mathcal{S%
}^{\prime };\mathbb{C})$, defined by $|\Omega \rangle \longrightarrow \Omega
(\eta )=1$ and $U(f)|\Omega \rangle \longrightarrow \exp [i\eta (f)]$ and
next extended by linearity upon the whole Hilbert space $\Phi $.

In the case of the non-cyclic case there exists a finite or countably
infinite family of measures $\left\{ \mu _{k}:k\in \mathbb{Z}_{+}\right\} $
on $\mathcal{S}^{\prime },$ with ${\Phi }\simeq \oplus _{k\in \mathbb{Z}%
_{+}}L_{2}^{(\mu _{k})}(\mathcal{S}^{\prime };\mathbb{C})$ and the unitary
operator $U(f):{\Phi }\longrightarrow {\Phi }$ for any $f\in \mathcal{S}%
^{\prime }$ corresponds in all $L_{2}^{(\mu _{k})}(\mathcal{S}^{\prime };%
\mathbb{C})$, $k\in \mathbb{Z}_{+}$, to $\exp [i\eta (f)]$. This means that
there exists a single cylindrical measure $\mu $ on $\mathcal{S}^{\prime }$
and a $\mu -$measurable field of Hilbert spaces ${\Phi }_{\eta }$ on $%
\mathcal{S}^{\prime }$, such that 
\begin{equation}
{\Phi }\simeq \int_{\mathcal{S}^{\prime }}^{\oplus }{\Phi }_{\eta }d\mu
(\eta ),  \label{FQeq2.28}
\end{equation}%
with $U(f):{\Phi }\longrightarrow {\Phi }$, corresponding \cite{GeVi} to the
operator of multiplication by $\exp [i\eta (f)]$ for any $f\in \mathcal{S}$
and $\eta \in \mathcal{S}^{\prime }$. Thereby, having constructed the
nonlinear functional (\ref{FQeq2.24}) in an exact analytical form, one can
retrieve the representation of the unitary family $\mathcal{U}$ in the
corresponding Hilbert space ${\Phi }$ of the Fock type, making use of the
suitable factorization (\ref{FQeq2.20}) as follows: ${\Phi }=\oplus _{n\in 
\mathbb{Z}_{+}}{\Phi }_{n}$, where 
\begin{equation}
{\Phi }_{n}=\underset{f_{n}\in L_{2,s}(\mathbb{R}^{m\times }{}^{n};\mathbb{C}%
)}{span}\left\{ \prod_{j=\overline{1,n}}a^{+}(x_{j})|\Omega \rangle \right\}
,  \label{FQeq2.29}
\end{equation}%
for all $n\in \mathbb{Z}_{+}$. The cyclic vector $|\Omega \rangle \in {\Phi }
$ can be, in particular, obtained as the ground state vector of some
unbounded self-adjoint positive definite Hamilton operator $\mathbb{H}:{\Phi 
}\longrightarrow {\Phi }$, commuting with the self-adjoint particles number
operator 
\begin{equation}
\mathbb{N}:=\int_{\mathbb{R}^{m}}\rho (x)dx,  \label{FQeq2.30}
\end{equation}%
that is $[\mathbb{H},\mathbb{N}]=0$. Moreover, the conditions 
\begin{equation}
\mathbb{H}|\Omega \rangle =0  \label{FQeq2.31}
\end{equation}%
and 
\begin{equation}
\inf_{g\in dom\mathbb{H}}\langle g,\mathbb{H}g\rangle =\langle \Omega |%
\mathbb{H}|\Omega \rangle =0  \label{FQeq2.32}
\end{equation}%
hold for the operator $\mathbb{H}:{\Phi }\longrightarrow {\Phi },$ where $%
dom $ $\mathbb{H}$ denotes its domain of definition.

To find the functional (\ref{FQeq2.26}), which is called the generating
Bogolubov type functional for moment distribution functions 
\begin{equation}
F_{n}(x_{1},x_{2},...,x_{n}):=\langle \Omega |:\rho (x_{1})\rho
(x_{2})...\rho (x_{n}):|\Omega \rangle ,  \label{FQeq2.33}
\end{equation}%
where $x_{j}\in \mathbb{R}^{m}$, $j=\overline{1,n}$, and the normal ordering
operation $:\cdot :$ is defined as 
\begin{equation}
:\rho (x_{1})\rho (x_{2})...\rho (x_{n}):\;=\prod_{j=1}^{n}\left( \rho
(x_{j})-\sum_{k=1}^{j}\delta (x_{j}-x_{k})\right) ,  \label{FQeq2.34}
\end{equation}%
it is convenient to choose the Hamilton operator $\mathbb{H}:{\Phi }%
\longrightarrow {\Phi }$ in the following \cite{GoGrPoSh,Gold,BoPr-1}
algebraic form: 
\begin{equation}
\mathbb{H}:=\frac{1}{2}\int_{\mathbb{R}^{m}}K^{+}(x)\rho
^{-1}(x)K(x)dx+V(\rho ),  \label{FQeq2.35}
\end{equation}%
being equivalent in the Hilbert space $\Phi $ to the positive definite
operator expression 
\begin{equation}
\mathbb{H}:=\frac{1}{2}\int_{\mathbb{R}^{m}}(K^{+}(x)-A(x;\rho ))\rho
^{-1}(x)(K(x)-\mathcal{A}(x;\rho ))dx,  \label{FQeq2.35a}
\end{equation}%
satisfying conditions (\ref{FQeq2.31}) and (\ref{FQeq2.32}), where $A(x;\rho
):\Phi \rightarrow \Phi ,$ $\ x\in \mathbb{R}^{m},$ is some specially chosen
linear self-adjoint operator. The \textquotedblleft potential" operator $%
V(\rho ):{\Phi }\longrightarrow {\Phi }$ \ is, in general, a polynomial (or
analytical) functional of the density operator $\rho (x):{\Phi }%
\longrightarrow {\Phi }$ and the operator is given as 
\begin{equation}
K(x):=\nabla _{x}\rho (x)/2\ +\ iJ(x),  \label{FQeq2.36}
\end{equation}%
where the self-adjoint \textquotedblleft current" operator $J(x):{\Phi }%
\longrightarrow {\Phi }$ can be defined (but non-uniquely) from the equality 
\begin{equation}
\partial \rho /\partial t=\frac{1}{i}[\mathbb{H},\rho (x)]=-<\nabla
_{x},J(x)>_{,}  \label{FQeq2.37}
\end{equation}%
holding for all $x\in \mathbb{R}^{m}.$ Such an operator $J(x):{\Phi }%
\longrightarrow {\Phi },$ $x\in \mathbb{R}^{m},$ can exist owing to the
commutation condition $[\mathbb{H},\mathbb{N}]=0,$ giving rise to the
continuity relationship (\ref{FQeq2.37}), if taking into account that
supports $supp\;\rho $ of the density operator $\rho (x):{\Phi }%
\longrightarrow {\Phi },$ $x\in \mathbb{R}^{m},$ can be chosen arbitrarily
owing to the independence of (\ref{FQeq2.37}) on the potential operator $%
V(\rho ):{\Phi }\longrightarrow {\Phi },$ but its strict dependence on the
corresponding representation (\ref{FQeq2.28}). In particular, based on the
Fock space $\Phi ,$ defined by (\ref{FQeq2.1}) and generated by the
creation-annihilation operators (\ref{FQeq2.7}) and (\ref{FQeq2.8}), the
current operator $\ \ J(x):{\Phi }\longrightarrow {\Phi },$ $x\in \mathbb{R}%
^{m},$ can be constructed as follows: 
\begin{equation}
J(x)=\frac{1}{2i}[a^{+}(x)\nabla a(x)-\nabla a^{+}(x)a(x)],
\label{FQeq2.37a}
\end{equation}%
satisfying jointly with the density operator $\rho (x):{\Phi }%
\longrightarrow {\Phi },$ $x\in \mathbb{R}^{m},$ defined by \ (\ref{FQeq2.16}%
), the following quantum current Lie algebra \cite%
{BoPr-1,GoMeSh,GoMeSh-1,PrBoGoTa,PrMy} relationships:%
\begin{eqnarray}
\lbrack J(g_{1}),J(g_{2})] &=&iJ([g_{1,}g_{2}]),\text{\ \ }
\label{FQeq2.37b} \\
\lbrack J(g_{1}),\rho (f_{1})] &=&i\rho (<g_{1},\nabla f_{1}>),  \notag \\
\lbrack \rho (f_{1}),\rho (f_{2})] &=&0,\text{ }  \notag
\end{eqnarray}%
holding for all $f_{1},f_{1}\in F$ and $g_{1},g_{2}\in F^{3},$ where \ we
put, by definition, 
\begin{equation}
\lbrack g_{1,}g_{2}]:=<g_{1},\nabla )g_{2}-<g_{2},\nabla )g_{1},
\label{FQeq2.37c}
\end{equation}%
being the usual commutator of vector fields in the Euclidean space $\mathbb{E%
}^{m}.$ It is easy to observe that the current algebra (\ref{FQeq2.37b}) is
\ the Lie algebra $\mathcal{G},$ corresponding to the Banach Lie group $%
G=Diff$ $\mathbb{E}^{3}\rightthreetimes F,$\ the semidirect product of the \
Banach Lie group of diffeomorphisms \ \ $Diff$ $\mathbb{E}^{3}$ of the
three-dimensional Euclidean space $\mathbb{E}^{3}$ and the abelian subject
to the multiplicative operation Banach group of smooth functions $F$ . We
note also that representation (\ref{FQeq2.35a}) holds only under the
condition that there exists such a self-adjoint operator $\mathcal{A}(x;\rho
):{\Phi }\longrightarrow {\Phi },$ $x\in \mathbb{R}^{m},$ that 
\begin{equation}
K(x)|\Omega \rangle =\mathcal{A}(x;\rho )|\Omega \rangle  \label{FQeq2.38}
\end{equation}%
for all ground states $|\Omega \rangle \in {\Phi },$ correspond to suitably
chosen potential operators $V(\rho ):{\Phi }\longrightarrow {\Phi }.$

The self-adjointness of the operator $\mathcal{A}(x;\rho ):{\Phi }%
\longrightarrow {\Phi },$ $x\in \mathbb{R}^{m},$ can be stated following
schemes from works \cite{GoGrPoSh,BoPr-1}, under the additional condition of
the existence of such a linear anti-unitary mapping $T:{\Phi }%
\longrightarrow {\Phi }$ that the following invariance conditions hold: 
\begin{equation}
T\rho (x)T^{-1}=\rho (x),\qquad T\;J(x)\;T^{-1}=-J(x),\qquad T|\Omega
\rangle =|\Omega \rangle  \label{FQeq2.39}
\end{equation}%
for any $x\in \mathbb{R}^{m}$. Thereby, owing to conditions (\ref{FQeq2.39}%
), the following expressions 
\begin{equation}
K^{\ast }(x)|\Omega \rangle =\mathcal{A}(x;\rho )|\Omega \rangle
=K(x)|\Omega \rangle  \label{FQeq2.40}
\end{equation}%
hold for any $x\in \mathbb{R}^{m}$, giving rise to the self-adjointness of
the operator $\mathcal{A}(x;\rho ):{\Phi }\longrightarrow {\Phi },$ $x\in 
\mathbb{R}^{m}$.

Based now on the construction above one easily deduces from expression (\ref%
{FQeq2.37}) that the generating Bogolubov type functional (\ref{FQeq2.26})
obeys for all $x\in \mathbb{R}^{m}$ the following functional-differential
equation: 
\begin{equation}
\left[ \nabla _{x}-i\nabla _{x}f\right] \frac{1}{2i}\frac{\delta \mathcal{L}%
(f)}{\delta f(x)}=\mathcal{A}\left( x;\frac{1}{i}\frac{\delta }{\delta f}%
\right) \mathcal{L}(f),  \label{FQeq2.41}
\end{equation}%
whose solutions should satisfy the Fourier transform representation (\ref%
{FQeq2.27}). In particular, a wide class of special so-called Poissonian
white noise type solutions to the functional-differential equation (\ref%
{FQeq2.41}) was obtained in \cite{GoGrPoSh,BoPr-1} by means of
functional-operator methods in the following generalized form: 
\begin{equation}
\mathcal{L}(f)=\exp \left\{ 2\mathcal{A}\left( \frac{1}{i}\frac{\delta }{%
\delta f}\right) \right\} \exp \left( \bar{\rho}\int_{\mathbb{R}^{m}}\{\exp
[if(x)]-1\}dx\right) ,  \label{FQeq2.41a}
\end{equation}%
where $\bar{\rho}:=\langle \Omega |\rho |\Omega \rangle \in \mathbb{R}_{+}$
is a Poisson distribution density parameter.

It is worth to remark here that solutions to equation (\ref{FQeq2.41})
realize the suitable physically motivated representations of the abelian
Banach subgroup $F$ of the Banach group $G=Diff$ $\mathbb{E}%
^{3}\rightthreetimes F,$ mentioned above. In the general case of the Banach
group $\ G=Diff$ $\mathbb{E}^{3}\rightthreetimes F$ one can also construct 
\cite{BoPr-1,GoMeSh} a generalized Bogolubov type functional equation, whose
solutions give rise to \ suitable physically motivated representations of
the corresponding current Lie algebra $\mathcal{G}.$

Consider now the case, when the basic Fock space ${\Phi }=\otimes _{j=1}^{s}{%
\Phi }^{(j)}$, where ${\Phi }^{(j)},$ $j=\overline{1,s}$, are Fock spaces
corresponding to the different types of independent cyclic vectors $|\Omega
_{j}\rangle \in {\Phi }^{(j)},$ $j=\overline{1,s}.$ This, in particular,
means that the suitably constructed creation and annihilation operators $%
a_{j}(x),a_{k}^{+}(y):{\Phi }\longrightarrow {\Phi },$ $j,k=\overline{1,s},$
satisfy the following commutation relations: 
\begin{equation}
\begin{array}{c}
\lbrack a_{j}(x),a_{k}(y)]=0, \\[5pt] 
\lbrack a_{j}(x),a_{k}^{+}(y)]=\delta _{jk}\delta (x-y)%
\end{array}
\label{FQeq2.42}
\end{equation}%
for any $x,y\in \mathbb{R}^{m}$.

\begin{definition}
\label{FQdef2.13} A vector $|u\rangle \in {\Phi },$ $x\in \mathbb{R}^{m},$
is called coherent \cite{Bere,Glau,PrBoGoTa,KoSt} with respect to a mapping $%
u\in L_{2}(\mathbb{R}^{m};\mathbb{R}^{s}):=M,$ if it satisfies the
eigenfunction condition 
\begin{equation}
a_{j}(x)|u\rangle =u_{j}(x)|u\rangle  \label{FQeq2.43}
\end{equation}%
for each $j=\overline{1,s}$ and all $x\in \mathbb{R}^{m}$.
\end{definition}

It is easy to check that the coherent vectors $|u\rangle \in {\Phi }$ exist.
Really, the following vector expression 
\begin{equation}
|u\rangle :=\exp \{(u,a^{+})\}|\Omega \rangle ,  \label{FQeq2.44}
\end{equation}%
where $(.,.)$ is the standard scalar product in the Hilbert space $M$,
satisfies the defining condition (\ref{FQeq2.43}), and moreover, the norm 
\begin{equation}
\Vert u\Vert _{{\Phi }}:=\langle u|u\rangle ^{1/2}=\exp (\frac{1}{2}\Vert
u\Vert ^{2})<\infty ,  \label{FQeq2.45}
\end{equation}%
since $u\in M$ and its norm $\Vert u\Vert :=(u,u)^{1/2}$ is bounded.

\subsection{The Fock space embedding method, nonlinear dynamical systems and
their complete linearization}

Consider any function $u\in M:=L_{2}(\mathbb{R}^{m};\mathbb{R}^{s})$ and
observe that the Fock space embedding mapping 
\begin{equation}
\phi :M\ni u\longrightarrow |u\rangle \in {\Phi },  \label{FQeq3.1}
\end{equation}%
defined by means of the coherent vector expression (\ref{FQeq2.44}) realizes
a smooth isomorphism between Hilbert spaces $M$ and ${\Phi }.$ The inverse
mapping $\phi ^{-1}:{\Phi }\longrightarrow M$ is given by the following
exact expression: 
\begin{equation}
u(x)=\langle \Omega |a(x)|u\rangle ,  \label{FQeq3.2}
\end{equation}%
holding for almost all $x\in \mathbb{R}^{m}$. Owing to condition (\ref%
{FQeq2.45}), one finds from (\ref{FQeq3.2}) that, the corresponding function 
$u\in M.$

In the Hilbert space $M,$ let now define a nonlinear dynamical system (which
can, in general, be non-autonomous) in partial derivatives 
\begin{equation}
du/dt=K[u],  \label{FQeq3.3}
\end{equation}%
where $t\in \mathbb{R}_{+}$ is the corresponding evolution parameter, $%
[u]:=(t,x;u,u_{x},u_{xx},...,u_{rx}),r\in \mathbb{Z}_{+}$, and a mapping $%
K:M\longrightarrow T(M)$ is Frechet smooth. Assume also that the Cauchy
problem 
\begin{equation}
u|_{t=+0}=u_{0}  \label{FQeq3.4}
\end{equation}%
is solvable for any $u_{0}\in M$ in an interval $[0,T)\subset \mathbb{R}%
_{+}^{1}$ for some $T>0.$ Thereby, the smooth evolution mapping is defined 
\begin{equation}
T_{t}:M\ni u_{0}\longrightarrow u(t|u_{0})\in M,  \label{FQeq3.5}
\end{equation}%
for all $t\in \lbrack 0,T).$

It is now natural to consider the following commuting diagram 
\begin{equation}
\begin{array}{ccc}
M & \overset{\phi }{\longrightarrow } & \Phi \\ 
T_{t}\downarrow &  & \downarrow {\mathbb{T}_{t}} \\ 
M & \overset{\phi }{\longrightarrow } & \Phi ,%
\end{array}
\label{FQeq3.6}
\end{equation}%
where the mapping $\mathbb{T}_{t}:{\Phi }\longrightarrow {\Phi }$, $t\in
\lbrack 0,T)$, is defined from the conjugation relationship 
\begin{equation}
\phi \circ T_{t}=\mathbb{T}_{t}.\circ \phi  \label{FQeq3.7}
\end{equation}

Now take coherent vector $|u_{0}\rangle \in {\Phi },$ corresponding to $%
u_{0}\in M,$ and construct the vector 
\begin{equation}
|u\rangle :=\mathbb{T}_{t}\text{ }|u_{0}\rangle  \label{FQeq3.8}
\end{equation}%
for all $t\in \lbrack 0,T).$ Since vector (\ref{FQeq3.8}) is, by
construction, coherent, that is 
\begin{equation}
a_{j}(x)|u\rangle :=u_{j}(x,t|u_{0})|u\rangle  \label{FQeq3.9}
\end{equation}%
for each $j=\overline{1,s}$, $t\in \lbrack 0,T)$ and almost all $x\in 
\mathbb{R}^{m}$, owing to the smoothness of the mapping $\xi
:M\longrightarrow {\Phi }$ with respect to the corresponding norms in the
Hilbert spaces $M$ and $\Phi ,$ we derive that coherent vector (\ref{FQeq3.8}%
) is differentiable with respect to the evolution parameter $t\in \lbrack
0,T)$. Thus, one can easily find \cite{Kowa,KoSt,PrBoGoTa,BoPr-1} that 
\begin{equation}
\frac{d}{dt}|u\rangle =\hat{K}[a^{+},a]|u\rangle ,  \label{FQeq3.10}
\end{equation}%
where 
\begin{equation}
|u\rangle |_{t=+0}=|u_{0}\rangle  \label{FQeq3.11}
\end{equation}%
and a mapping $\hat{K}[a^{+},a]:{\Phi }\longrightarrow {\Phi }$ is defined
by the exact analytical expression 
\begin{equation}
\hat{K}[a^{+},a]:=(a^{+},K[a]).  \label{FQeq3.12}
\end{equation}

As a result of the consideration above we obtain the following theorem.

\begin{theorem}
\label{FQth3.1} Any smooth nonlinear dynamical system (\ref{FQeq3.3}) in
Hilbert space $M:=L_{2}(\mathbb{R}^{m};\mathbb{R}^{s})$ is representable by
means of the Fock space embedding isomorphism $\phi :M\longrightarrow {\Phi }
$ in the completely linear form (\ref{FQeq3.10}).
\end{theorem}

We now make some comments concerning the solution to the linear equation (%
\ref{FQeq3.10}) under the Cauchy condition (\ref{FQeq3.11}). Since any
vector $|u\rangle \in {\Phi }$ allows the series representation 
\begin{equation}
\begin{array}{l}
|u\rangle =\underset{n=\sum_{j=1}^{s}n_{j}\in \mathbb{Z}_{+}}{\bigoplus }%
\frac{1}{(n_{1}!n_{2}!...n_{s}!)^{1/2}}\int_{(\mathbb{R}%
^{m})^{n}}f_{n_{1}n_{2}...n_{s}}^{(n)}(x_{1}^{(1)},x_{2}^{(1)},...,x_{{n}%
_{1}}^{(1)}; \\[10pt] 
\qquad x_{1}^{(2)},x_{2}^{(2)},...,x_{{n}%
_{2}}^{(2)};...;x_{1}^{(s)},x_{2}^{(s)},...,x_{{n}_{s}}^{(s)})%
\prod_{j=1}^{s}\left(
\prod_{k=1}^{n_{j}}dx_{k}^{(j)}a_{j}^{+}(x_{k}^{(j)})\right) |\Omega \rangle
,%
\end{array}
\label{FQeq3.13}
\end{equation}%
where for any $n=\sum_{j=1}^{s}n_{j}\in \mathbb{Z}_{+}$ functions 
\begin{equation}
f_{n_{1}n_{2}...n_{s}}^{(n)}\in \bigotimes_{j=1}^{s}L_{2,s}((\mathbb{R}%
^{m})^{n_{j}};\mathbb{C})\simeq L_{2,s}(\mathbb{R}^{mn_{1}}\times \mathbb{R}%
^{mn_{2}}\times ...\mathbb{R}^{mn_{s}};\mathbb{C}),  \label{FQeq3.14}
\end{equation}%
and the norm 
\begin{equation}
\Vert u\Vert _{{\Phi }}^{2}=\sum_{n=\sum_{j=1}^{s}n_{j}\in \mathbb{Z}%
_{+}}\Vert f_{n_{1}n_{2}...n_{s}}^{(n)}\Vert _{2}^{2}=\exp (\Vert u\Vert
^{2}).  \label{FQeq3.15}
\end{equation}%
By substituting (\ref{FQeq3.13}) into equation (\ref{FQeq3.10}), reduces (%
\ref{FQeq3.10}) to an infinite recurrent set of linear evolution equations
in partial derivatives on coefficient functions (\ref{FQeq3.14}). The latter
can often be solved \cite{Kowa,PrBoGoTa} step by step analytically in exact
form, thereby, making it possible to obtain, owing to representation (\ref%
{FQeq3.2}), the exact solution $u\in M$ to the Cauchy problem (\ref{FQeq3.4}%
) for our nonlinear dynamical system in partial derivatives (\ref{FQeq3.3}).

\begin{remark}
\label{FQre3.2} Concerning some applications of nonlinear dynamical systems
like (\ref{FQeq3.1}) in mathematical physics problems, it is very important
to construct their so called conservation laws or smooth invariant
functionals $\gamma :M\longrightarrow \mathbb{R}$ on $M$. Making use of the
quantum mathematics technique described above one can suggest an effective
algorithm for constructing these conservation laws in exact form.
\end{remark}

Indeed, consider a vector $|\gamma \rangle \in {\Phi },$ satisfying the
linear equation: 
\begin{equation}
\frac{\partial }{\partial t}|\gamma \rangle +\hat{K}^{\ast }[a^{+},a]|\gamma
\rangle =0.  \label{FQeq3.16}
\end{equation}

Then, the following proposition \cite{Kowa,PrBoGoTa} holds.

\begin{proposition}
\label{FQpr3.3} The functional 
\begin{equation}
\gamma :=\langle u|\gamma \rangle  \label{FQeq3.17}
\end{equation}%
is a conservation law for dynamical system (\ref{FQeq3.1}), that is 
\begin{equation}
d\gamma /dt|_{K}=0  \label{FQeq3.18}
\end{equation}%
along any orbit of the evolution mapping (\ref{FQeq3.5}).
\end{proposition}

\bigskip

\subsection{The quantum current Lie algebra and \ the magnetic Aharonow-Bohm
effect}

\bigskip

In the Section above we could get convinced that different representations
of the equal-time current algebra (\ref{FQeq2.37b}) 
\begin{eqnarray}
\lbrack \rho (f_{1}),\rho (f_{2})] &=&0,\text{ \ \ }  \label{FQeq3.19} \\
\lbrack J(g_{1}),\rho (f_{1})] &=&i\rho (<g_{1},\nabla f_{1}>),  \notag \\
\lbrack J(g_{1}),J(g_{2})] &=&iJ([g_{2},g_{1}])+i\rho (<B,g_{1}\times
g_{2}>),  \notag
\end{eqnarray}%
where $f_{1},f_{2}\in F$ and $g_{2},g_{1}\in F^{3},$ acting in the Fock
space $\Phi $ and describing a non-relativistic quasi-stationary system
consisting of a test charged particle $q,$ imbedded into a cylindrical
region $\Gamma \subset \mathbb{E}^{3},$ being under influence of the
magnetic field $B=\nabla \times A.$ \ Here $A\in \mathbb{E}^{3}$ is a
magnetic vector potential, the sign \textquotedblleft $\times $%
\textquotedblright\ means the vector product in $\mathbb{E}^{3}$ and the
current $\ \ J(x):{\Phi }\longrightarrow {\Phi },$ $x\in \mathbb{R}^{m},$ is
defined, owing to the \textit{minimal interaction} principle (\ref{FQeq2.13}%
), as 
\begin{equation}
J(x):=\frac{1}{2}a^{+}(x)(\frac{1}{i}\nabla -A)a(x)-[(\frac{1}{i}\nabla
+A)a^{+}(x)]a(x).  \label{FQeq3.20}
\end{equation}%
In particular, it is assumed that $supp$ $B\subset \Gamma $ that gives rise
to the equality $\rho (<B,g_{1}\times g_{2}>=0$ \ for all points $x\in 
\mathbb{E}^{3}\backslash \Gamma .$

As the suitable representations of the current algebra $\mathcal{G},$
defined by (\ref{FQeq3.20}), describe the physical quantum states of the
system $\Gamma $ under regards, we consider them following \cite%
{GoMeSh,GoMeSh-1} as those realized in the Hilbert space $L_{2}(\mathbb{E}%
^{3};\mathbb{C})$ under the condition that the charged test particle $q$ can
penetrate the boundary $\partial \Gamma $ \ of the region $\Gamma .$ Namely,
let for $\psi \in L_{2}($ $\mathbb{\Gamma };\mathbb{C})$ 
\begin{eqnarray}
\rho (f)\psi (x) &:&=f(x)\psi (x),  \label{FQeq3.21} \\
J(g)\psi (x) &=&\frac{1}{2i}\left\{ [<g(x),\cdot \nabla >+<\nabla ,\cdot
g(x)>]-<g(x),\int_{\Gamma }d^{3}y\frac{\nabla \times B(y)}{4\pi |x-y|}%
>\right\} \psi (x),  \notag
\end{eqnarray}%
for all $x\in \mathbb{E}^{3},$ where the sign \textquotedblleft $\cdot $%
\textquotedblright\ means that the natural operator composition. When
deriving (\ref{FQeq3.21}) there was imposed the invariant Coulomb gauge
constraint $\ <\nabla ,A>=0$ allowing to determine the vector potential $%
A\in \mathbb{E}^{3}$ as 
\begin{equation}
A=\int_{\Gamma }d^{3}y\frac{\nabla \times B(y)}{4\pi |x-y|}  \label{FQeq3.22}
\end{equation}%
using the classical Maxwell equations (\ref{L0.6}), since the electric
displacement current component $\ \partial E/\partial t=0.$ The wave
function $\psi \in L_{2}(\mathbb{\Gamma };\mathbb{C})$ satisfies \ in the
cylindrical coordinates $x(r,\theta ,z)\in \Gamma $ \ the natural
quasi-periodical condition%
\begin{equation}
\psi (r,\theta +2\pi n,z)=\exp (i\lambda n)\psi (r,\theta ,z)
\label{FQeq3.23}
\end{equation}%
for some $\lambda \in \mathbb{R}$ and any $n\in \mathbb{Z},$ which should be
determined from the physically realizable representation (\ref{FQeq3.21}).
To do this, we need preliminarily to define the following \cite{GoMeSh-1}
unitary operator in the Hilbert space $L_{2}(\mathbb{E}^{3};\mathbb{C}):$%
\begin{equation}
Q\psi (x):=\left\{ 
\begin{array}{c}
\psi (x),\text{ \ \ \ \ \ \ \ \ \ \ \ \ \ \ \ \ \ \ \ \ \ \ \ \ \ \ \ \ \ \
\ \ \ \ \ \ \ }x\in \Gamma ; \\ 
\exp \left\{ -i\int_{l_{\infty }}^{x}dl(y)\int_{E^{3}\backslash \Gamma
}dy^{\prime }\frac{\nabla \times B(y^{\prime })}{|y-y^{\prime }|}\right\}
\psi (x),\text{ \ \ }x\in \mathbb{E}^{3}\backslash \Gamma ;\text{\ }%
\end{array}%
\right.  \label{FQeq3.24}
\end{equation}%
where the path $l_{\infty }\subset \mathbb{E}^{3}\backslash \Gamma $
connects an infinite point $\infty \in \mathbb{E}^{3}$ with the chosen point 
$x\in \mathbb{E}^{3}\backslash \Gamma .$ Making use now of the unitary
transformation $\ \tilde{J}(g):=QJ(g)Q^{-1}$ and the fact that the magnetic
field 
\begin{equation}
B(x)=\nabla \times \int_{\Gamma }\frac{\nabla \times B(y)d^{3}y}{4\pi |x-y|}
\label{FQeq3.25}
\end{equation}%
for any $x\in \mathbb{E}^{3},$ one obtains easily that the current operator $%
\ \tilde{J}(g):L_{2}(\mathbb{E}^{3};\mathbb{C})\longrightarrow L_{2}(\mathbb{%
E}^{3};\mathbb{C})$ is self-adjoint for any $g\in F^{3}$ and 
\begin{equation}
\tilde{J}(g)\psi (x)=\frac{1}{2i}[<g(x),\cdot \nabla >+<\nabla ,\cdot
g(x)>]\psi (x),  \label{FQeq3.26}
\end{equation}%
and whose domain of definition $dom$ $\tilde{J}(g)\subset $ $L_{2}(\mathbb{E}%
^{3};\mathbb{C})$ is constrained by the functions $\psi \in L_{2}(\mathbb{E}%
^{3};\mathbb{C}),$ satisfying the condition 
\begin{equation}
\psi (r,\theta +2\pi ,z)=\exp [i\lambda (B)])\psi (r,\theta ,z),
\label{FQeq3.27}
\end{equation}%
where, owing to (\ref{FQeq3.24}), 
\begin{equation}
\lambda (B)=-\int_{\partial \Gamma }<B,dS>.  \label{FQeq3.28}
\end{equation}%
Thus, the found above representation (\ref{FQeq3.26}) of the current Lie
algebra (\ref{FQeq3.19}) in the Hilbert space $L_{2}(\mathbb{E}^{3};\mathbb{C%
}),$ \ in the case when $supp$ $B$ $\subset \Gamma ,$ describes the complete
set of observables if \ the charged particle $q$ is not excluded from the
region $\Gamma .$ \ In contrast, if the region $\Gamma $ \ possesses a
potential barrier at the boundary $\partial \Gamma ,$ such that the charged
particle $q$ can not penetrate it and enter the \ region $\Gamma ,$ \ the
value of $\lambda (B)\in \mathbb{R},$ defined by (\ref{FQeq3.28}), remains
constant. This entails that a suitable outside the region $\Gamma $
measurement can certainly indicate the presence of the magnetic field inside 
$\Gamma .$ So, as it was mentioned in \cite{GoMeSh}, the constructed above
current algebra representation completely describes our non-relativistic
quasi-stationary system not giving rise to the Aharonov-Bohm \cite{AhBo}
paradox. \ Moreover, the outside measurements results simply depend on the
representation of the current algebra (\ref{FQeq3.19}), which in turn
depends on the history of the system and the topology of the space outside
the barrier. As a related physical aspect of the explanation above it is
necessary to stress that \ vanishing \ of the magnetic field outside the
region $\Gamma ,$ possessing a nontrivial topology, does not imply the
simultaneous vanishing of the corresponding magnetic potential outside the
region $\Gamma .$ Namely the latter in somewhat obscured form was used in
the analysis of the current algebra representation, suitable for describing
the complete set of physical observables both inside and outside the region $%
\Gamma .$

\bigskip

\subsection{Comments}

Within the scope of this Section we have described the main mathematical
preliminaries and properties of the quantum mathematics techniques suitable
for analytical studying of the important linearization problem for a wide
class of nonlinear dynamical systems in partial derivatives in Hilbert
spaces. This problem was analyzed in much detail using the Gelfand-Vilenkin
representation theory \cite{GeVi,BoLoOkTo} of infinite dimensional groups
and the Goldin-Menikoff-Sharp theory \cite{GoGrPoSh,Gold,GoMeSh} of
generating Bogolubov type functionals, classifying these representations.
The related problem of constructing Fock type space representations and
retrieving their creation-annihilation generating structure still needs a
deeper investigation within the approach devised. Based on the quantum
current Lie algebra description of a bounded non-relativistic quantum system
under an external electromagnetic field within the minimal interaction
principle, the magnetic Aharonow-Bohm effect has been interpreted. As it was
mentioned in \cite{GoMeSh}, the suitably constructed current algebra
representation completely describes our non-relativistic quasi-stationary
system not giving rise to the Aharonov-Bohm \cite{AhBo} paradox. We have
also presented main mathematical preliminaries and properties of the related
quantum mathematics techniques suitable for analytical studying of the
important linearization problem for a wide class of nonlinear dynamical
systems in partial derivatives in Hilbert spaces. Concerning this direction
it is worthy to mention the related results also obtained in \cite%
{SaPrTaPrBl,PrBoGoTa,Kowa,KoSt}, devoted to the application of the Fock
space embedding method to studying solutions to a wide nonlinear dynamical
systems and to constructing quantum computing algorithms.

\newpage

\section{\protect\bigskip Acknowledgements}

A.P. is acknowledged \ to Prof. Edward Kapuscik and Dr. Yuriy Yaremko for
friendly cooperation and important discussions. D.B. thanks the NSF for a
partial support.

\end{document}